\documentclass[12pt,english]{article}
\usepackage[T1]{fontenc}
\usepackage[latin9]{inputenc}
\usepackage{geometry}
\usepackage{color}
\usepackage{babel}
\usepackage{float}
\usepackage{booktabs}
\usepackage{amsmath}
\usepackage{amsthm}
\usepackage{amssymb}
\usepackage{graphicx}
\usepackage{setspace}
\usepackage[authoryear]{natbib}
\usepackage[unicode=true,pdfusetitle,
 bookmarks=true,bookmarksnumbered=false,bookmarksopen=false,
 breaklinks=false,pdfborder={0 0 1},backref=false,colorlinks=true]
 {hyperref}
\hypersetup{
 linkcolor=red, citecolor=blue, unicode, psdextra, pdfencoding=auto}

\makeatletter

\providecommand{\tabularnewline}{\\}

\theoremstyle{plain}
\newtheorem{prop}{\protect\propositionname}[section]
\theoremstyle{plain}
\newtheorem{assumption}{\protect\assumptionname}
\theoremstyle{plain}
\newtheorem{lem}{\protect\lemmaname}[section]

\usepackage{booktabs}
\usepackage{lscape}
\usepackage{setspace}
\usepackage{caption}
\usepackage{threeparttable}
\usepackage{longtable}
\usepackage{lmodern}
\usepackage[T1]{fontenc}
\usepackage{color}
\usepackage{pdflscape}
\usepackage{fancyhdr} 
\usepackage{tabularx}
\usepackage{subfig}

\newcommand{\sym}[1]{\rlap{#1}}

\usepackage{siunitx} 
\allowdisplaybreaks

\fancypagestyle{appendix}{%
  \fancyhf{}
  \fancyfoot[C]{Online Appendix \thepage}%
  
}

\@ifundefined{showcaptionsetup}{}{%
 \PassOptionsToPackage{caption=false}{subfig}}
\usepackage{subfig}
\makeatother

\usepackage{listings}
\providecommand{\assumptionname}{Assumption}
\providecommand{\lemmaname}{Lemma}
\providecommand{\propositionname}{Proposition}

\begin{document}
\title{Party On: The Labor Market Returns to Social Networks in Adolescence\thanks{We are grateful to seminar and conference participants at Northwestern
University, Santa Barbara University, USC, Peking University, SOCCAM,
the All-California Labor Economics conference, the North American
Winter Meeting of the Econometric Society, the Asian Meeting of the
Econometric Society in China, and the China Labor Economist Conference.
We are particularly grateful to Larry Katz for his very insightful
suggestions. This project was supported by the California Center for
Population Research at UCLA (CCPR), which receives core support (P2C-
HD041022) from the Eunice Kennedy Shriver National Institute of Child
Health and Human Development (NICHD). Research done prior to joining
Audible. Audible is a subsidiary of Amazon. All views contained in
this research are the author's own and do not represent the views
of Amazon or its affiliates. All errors are our own.}}
\author{Adriana Lleras-Muney, Matthew Miller, Shuyang Sheng, and Veronica
Sovero\thanks{Adriana Lleras-Muney, Department of Economics, UCLA, email: alleras@econ.ucla.edu;
Matthew Miller, Audible Inc, email: mattmm@audible.com; Shuyang Sheng,
Department of Economics, UCLA, email: ssheng@econ.ucla.edu; Veronica
Sovero, Department of Economics, UC Riverside, email: vsovero@ucr.edu.}}
\maketitle
\begin{abstract}
\begin{singlespace}
\noindent We investigate the returns to adolescent friendships on
earnings in adulthood using data from the National Longitudinal Study
of Adolescent to Adult Health. Because both education and friendships
are jointly determined in adolescence, OLS estimates of their returns
are likely biased. We implement a novel procedure to obtain bounds
on the causal returns to friendships: we assume that the returns to
schooling range from 5 to 15\% (based on prior literature), and instrument
for friendships using similarity in age among peers. Having one more
friend in adolescence increases earnings between 7 and 14\%, substantially
more than OLS estimates would suggest. 
\end{singlespace}
\end{abstract}
\begin{singlespace}

\section{Introduction}
\end{singlespace}

An individual's social capital (the number and quality of their connections)
impacts many economic outcomes. Individuals' classroom and school
peers impact their education (\citealp{Sacerdote2001}; \citealp{Carrell2009}),
earnings (\citealp{carrell2018disruptivepeers}; \citealp{Michelman2021})
and health (\citealp{Carrell2011}). Whom one befriends from among
these peers also influences important lifetime outcomes: \citet{Chetty2022}
show that the number of high SES friendships and economic mobility
in a neighborhood are positively associated. However, there is little
evidence on whether the number of one's friends (not just the types
they are exposed to) causally affect labor market outcomes.

Using data from the National Longitudinal Study of Adolescent to Adult
Health (hereafter Add Health), which follows individuals from adolescence
into adulthood, we study the labor market returns to adolescent friendships.
During adolescence, individuals develop important cognitive and social
skills, have key educational outcomes determined and form friendships
\citep{heckmanmobility2014}. Friendships are associated with better
larbor market outcomes: Figure \ref{fig:returns_edu_friend} plots
log annual earnings at ages 24--34 against the number of friends
at ages 12--20, separately for males and females. The non-parametric
lines show a striking positive association between adolescent friendships
and young adult earnings for both sexes.\footnote{We detail later in the paper how we measure the number of friends.}
OLS estimates of this association may be biased downward however:
individuals with social skills may engage in activities excluded from
regressions, like drinking and partying, that help them make friends
but reduce their education and earnings. This ``taste for partying''
is unobserved and omitted from earnings regressions, and can cause
the return to friendships to be downward biased. Conversely, social
individuals may prefer productive social activities (working together),
which improve their networks and yield higher earnings, biasing the
return to friendships upward.\footnote{Social skills are associated with large and growing returns in the
labor market (\citealp{Weidmann2021}; \citealp{Deming2017}).}

To estimate the causal returns to friendships, we construct an instrument
for the number of friends that exploits the fact that homophily (similarity
in traits) predicts friendship formation (\citealp{Boucher2015};
\citealp{Jackson2014}). Our instrument is the average absolute difference
in age between a student and the peers in her school and grade, which
we refer to as age distance. Because our baseline regression controls
for own age and the average age of individuals in the school-grade,
as well as grade and school fixed effects, age distance will vary
across individuals because the distribution of ages varies across
schools and grades. For example, a 13.5 year old with two peers aged
12.5 and 14.5 has an age distance of 1 while mean age is 13.5. If
the same 13.5 year old were instead in a group with a 13 year old
and a 14 year old, mean age would be the same, but age distance would
be smaller (0.5), resulting in more friends. This variation is similar
to what is used in \citet{Bifulco2011classmates} and \citet{carrell2018disruptivepeers},
who leverage variation across cohorts within schools to estimate the
effects of peer characteristics on outcomes.\footnote{These papers typically investigate the direct (reduced form) relationship
between peer traits and outcomes. Although we leverage similar (though
not identical) variation, we use this variation in peer characteristics
as an instrument for friendships -- we do not study the direct effect
of peer similarity on outcomes.} In our setting, there is also variation in the instrument \textit{within}
school-grade because age distance depends both on the mean age in
their group and the student's own age. Our identification strategy
is most similar to \citet{Fletcher2020} who investigate the effect
of friendships on education outcomes using similarity in race/ethnicity
and socioeconomic status as instruments for friendships. We investigate
the effect of friendships on earnings instead.

Education is also endogenous in these regressions because educational
and social investments are determined jointly in adolescence. Studies
with well-identified causal returns to education make use of large
data sets and typically exploit state-level variation in compulsory
schooling or the cost of schooling, such as distance to school (\citealp{Mountjoy2022})
or school openings (see \citet{card2001estimating} for a review and
\citet{Oreopolous2020} and \citet{Psacharopoulos2018} for more recent
examples). In our data, these instruments are weak predictors of education.

We propose a novel approach to estimating the returns to friendships
that does not require instruments for education. We rely on the findings
from the previous literature and assume that the causal returns to
a year of schooling range from 5 to 15\% (\citealp{Oreopolous2020};
\citealp{Psacharopoulos2018}). Under this assumption, and making
use of the homophily-based instrument for friendships, we derive bounds
for the returns to friendships.\footnote{We are grateful to Larry Katz for suggesting this approach.}

We find that the returns to having one more friend during adolescence
range between 7 and 14\%, similar to the returns to one more year
of schooling. Our identifying assumption is that age distance determines
earnings only through its effects on education and friendships, conditional
on own characteristics, mean peer characteristics, and grade and school
fixed effects. The results are robust to a number of checks, including
addressing the potential concern that parents sort into schools or
grades based on relative age. These instrumented returns to friendships
are larger than OLS estimates. Computations suggest that measurement
error can explain the discrepancy, consistent with \citet{Griffith2021}.
The result is also consistent with downward omitted variable bias:
the data show that preferences for activities like drinking are associated
with more friendships but lower GPA and earnings.

We contribute to the well-established literature examining peer effects
among adolescents and young adults (see the review by \citealp{Jeon2015}).\footnote{Identifying how the behaviors and characteristics of peers affect
individuals was first discussed by \citet{Manski1993}. A few early
papers (\citealp{BRAMOULLE2009identification}) attempted to use friendship
networks to instrument for peer outcomes and overcome the joint determination
of outcomes (the reflection problem). These papers assume that a friendship
network is exogenous or endogenous through unobserved group heterogeneity.} For example, having peers with good academic outcomes improves one's
academic outcomes (\citealp{Carrell2008,Carrell2009}; \citealp{SACERDOTE2011peer};
\citealp{Bifulco2011classmates}; \citealp{Denning2021}). Peer effects
in adolescence can also carry into adulthood. \citet{carrell2018disruptivepeers}
find that having disruptive peers in the classroom during adolescence
has deleterious effects on individuals' labor market outcomes as adults,
partly because disruptive peers lower test scores.

Our paper suggests a new mechanism by which peer characteristics might
operate: friendship formation. A few papers investigate the effects
of friendship networks on educational outcomes using data on friendship
nominations (\citealp{Babcock2008}; \citealp{LavySand2019}; \citealp{Fletcher2020}).
However, friendships may also matter for labor market outcomes, separately
from their effects on education attainment. Previous work documents
that the size and connectedness of an individual's social network
in adulthood can improve wages and job match quality (\citealp{granovetter1973strength};
\citealp{montgomery1991social}; \citealp{calvo2004effects}; \citealp{Cappellari2015};
\citealp{Dustmann2015}). In addition to the number of friends, the
types of friends one is associated with may also impact labor market
outcomes (\citealp{Chetty2022}). Although our results on this are
only suggestive (we do not have powerful instruments for different
types of friendships), we find that almost all types of friendships
have returns in the labor market, including friends from low SES backgrounds
and with weak connections. This suggests that the overall number of
friendships -- not just their type -- is a relevant factor in determining
earnings.

There are no papers we are aware of that estimate the \textit{causal}
returns to the number of friendships on labor market outcomes --
this is the main contribution of this paper. The closest paper to
ours is by \citet{conti2013popularity}, who estimate returns to high
school friendships on wages using data from the Wisconsin Longitudinal
Survey. Their estimation strategy corrects for non-classical measurement
error in the number of friendships due to under-sampling, but they
do not account for endogeneity in both education and friendships which
this paper addresses. 

Our second contribution is methodological. A strand of econometric
literature considers models with \textit{endogenous} networks (\citealp{goldsmith2013social};
\citealp{Hsieh2015network}; \citealp{Badev2021}; \citealp{johnsson2021peer};
\citealp{Auerbach2022}; \citealp{Griffith2022}; \citealp{Sheng2023}).
But these papers do not consider the joint determination of education
and friendships. As a result, they do not address how the endogeneity
in education complicates the identification of causal labor market
returns to friendships.

Our identification approach is motivated by the econometric literature
on partial identification, which is widely used as a remedy for identification
failure in various applications such as interval data \citep{Manski2002},
missing data \citep{Manski2003}, auctions \citep{Haile2003}, and
games with multiple equilibria \citep{Tamer2003}. We exploit the
idea of partial identification to resolve a new challenge where we
lack a good instrument for one of the endogenous regressors, but the
coefficient of that regressor is not of primary interest and can be
bounded. To our knowledge, this is the first paper that proposes a
partial identification approach to avoiding instrumenting for a secondary
endogenous regressor.

\section{Conceptual Framework: How are Friendships Formed in Adolescence?}

We start by summarizing a basic model of education and friendship
formation (Appendix \ref{app:Model}) and describe its implications
for the empirical analysis.\footnote{All figures and tables designated with a letter (e.g., \textquotedblleft A\textquotedblright ,
``B'') are shown in the Online Appendix.}

During adolescence, individuals decide how to allocate their time
between studying, socializing, and leisure. Studying increases educational
attainment, and socializing increases one's number of friends, both
of which have positive returns in the labor market. In deciding how
to allocate their time, individuals consider the returns to each activity,
which depend on their innate intelligence and social skills (Proposition
\ref{prop:endow}).

Because time spent investing in education and friends is determined
at the same time, both education and friendships are potentially endogenous
in a Mincer earnings equation and thus OLS estimates of their returns
may be biased. However, the sign of the bias in friendship returns
is not clear ex-ante (Proposition \ref{prop:ols}). On the one hand,
social skills may determine the number of friendships and have an
independent effect on earnings, causing an upwards bias in the returns
to friendships. On the other hand, partying and drinking may increase
friendships but negatively affect skill accumulation and labor market
outcomes, causing a downward bias in the returns to friendships.

To account for the endogeneity of friendships, we will exploit the
fact that an individual's accumulated social capital also depends
on the traits and decisions of their school-grade peers because the
production of friendships requires coordination with others -- you
cannot ``party alone''. Conditional on time spent socializing, individuals
that are more similar to each other are more likely to become friends.\footnote{In equilibrium this may not hold (see Proposition \ref{prop:homophily}).}

\section{Data}

We use the restricted-use National Longitudinal Study of Adolescent
to Adult Health (Add Health). The in-school sample is a complete census
of all students enrolled in a given school during the 1994--95 school
year. The in-school sample data include basic demographics as well
as friendship nominations. A random sample of the students interviewed
in school was selected for in-home interviews in Wave 1 during the
1994--95 school year (ages 12--20 years) and tracked over in subsequent
survey waves. In Wave 4 (which was conducted in 2008--09), respondents
were age 24 to 34, on average 29 years old. For this in-home sample,
we observe measures of endowments, investments, and cognitive and
social outcomes.

Our estimation sample is constructed from the in-home sample, but
we use information from the in-school sample to construct friendship
and homophily measures. We include 10,605 individuals with complete
data for gender, age, and race. Summary statistics for the estimation
sample are presented in Table \ref{tab:sum_stats}.

\textbf{Outcomes.} The main outcome of interest is total earnings
from wages or salary in the last year. If a respondent replied \textquotedblleft do
not know\textquotedblright{} to the earnings question, they were prompted
with twelve categories of earnings. We use the midpoint of the selected
range for these respondents (approximately 2\% of the sample). We
drop individuals who reported zero earnings (\textasciitilde 6\%),
which means they were unemployed the entire year. Individuals in our
sample made roughly \$38,000 in the previous year.

\textbf{Intermediate outcomes: friendships and education.} Our primary
measure of the number of friendships is a person's \emph{grade in-degree}.
For a given student, grade in-degree is the number of people \textit{within
the same school and grade} who nominate them as one of their friends
in Wave 1.\footnote{Individuals could list up to five nominations of each gender. }
In-degree has been widely used in the social network literature as
an objective measure of an individual's number of friendships because
it does not rely on self-reporting \citep{conti2013popularity}. Importantly,
we observe all nominations sent to students in the in-home sample
because the network data is derived from the in-school sample. Unlike
Add Health's measure of in-degree, we exclude nominations from students
in different grades: the majority of nominations (77\%) occur within
the same grade (on average, school in-degree is 4.4, whereas grade
in-degree is only 3.4 (Table \ref{tab:sum_stats})).

Education is measured by years of schooling. On average, individuals
in our data obtain almost 15 years of schooling. We also observe GPA,
a measure of how much students learned in school.

\textbf{Endowments}. We use self-reported extroversion, collected
in Wave 2, as the main measure of the social endowment of individuals.
Extroversion is one of the ``Big Five'' psychological traits. About
65\% of individuals report being extroverted.\footnote{The survey question is ``You are shy?'', and the the choices are
``strongly agree / agree / neither agree nor disagree / disagree
/ strongly disagree''. Individuals choosing last three categories
are defined as extrovert. Due to the survey design this measure is
missing for 26\% of individuals in the data. To maximize sample size,
we impute this measure and include a dummy for whether it is missing.} Consistent with existing evidence (\citealp{lenton2014personality}),
extroverts have larger earnings, and perhaps not surprisingly, more
friends (Columns 2 and 3 of Table \ref{tab:sum_stats}). Most interestingly,
they also have more years of schooling.

The Add Health Picture Vocabulary Test (AHPVT) score is our main measure
for cognitive endowments. This test, administered in Wave 1, is an
abbreviated version of the widely used Peabody Picture Vocabulary
test and measures verbal ability. While it is not an overall measure
of intellectual ability, it has a high correlation with other intelligence
tests (\citealp{Hodapp1999}; \citealp{Dunn2007}). For simplicity,
we refer to it as IQ. As expected, individuals with above median IQs
have greater earnings and education (more years of education and higher
GPA). Perhaps surprisingly, they also have more friends (Columns 4
and 5 of Table \ref{tab:sum_stats}).

\section{Empirical Strategy}

We now turn our attention to estimating the labor market returns to
friendships. We follow the previous literature and estimate the following
earnings equation (conditional on employment):
\[
Y_{i}=r_{e}E_{i}+r_{f}F_{i}+\beta'X_{i}+\gamma'\bar{X}_{gs}+\alpha_{g}+\lambda_{s}+\epsilon_{i},
\]
where the outcome of interest is the log annual earnings $Y_{i}$
in Wave 4 for a given individual $i$ (observed in grade $g$ and
school $s$ during Wave 1), $E_{i}$ stands for years of schooling,
and $F_{i}$ is a measure of $i$'s number of friends (such as in-degree).
We control for $i's$ characteristics $X_{i}$ (age in Wave 1, sex,
race, IQ, and extroversion) and mean characteristics in $i's$ school-grade
$\bar{X}_{gs}$ (mean age, fraction female, fraction white, mean IQ
and fraction extrovert). We control for grade fixed effects ($\alpha_{g}$)
to account for differences across grades within a school. We include
school fixed effects ($\lambda_{s}$) to control for unobserved school-level
characteristics that could be correlated with labor market outcomes
and potentially sort individuals into schools. Standard errors are
clustered at the school level.

The object of interest is the coefficient $r_{f}$, measuring the
returns to having one more friend. Proposition \ref{prop:ols} shows
that OLS estimates of $r_{f}$ are biased because education and friendships
are jointly determined. We take an instrumental variable approach
to overcome the endogeneity issue, which will also address classical
measurement error in number of friends.

We use homophily measures as instruments for friendships, following
the evidence that individuals that resemble each other are more likely
to become friends \citep{Jackson2008}. Unfortunately, the instruments
for education used in the literature (quarter of birth, distance to
school) are weak in our data. Instead, we propose a novel approach
to estimating the returns to friendships that does not require an
instrument for education.\footnote{The education literature shows that relative age within a classroom
also affects educational attainment (e.g., \citealp{black2011too}).
However, we experimented using homophily measures for education as
well as for in-degree and found that homophily instruments fail the
weak IV tests if we use them to predict both education and friendships.}

\subsection{\label{sec:bounds}Bounding the returns to friendships}

There is a substantial literature estimating the (causal) returns
to schooling in the United States using various approaches. While
estimates differ across studies and populations, causal estimates
of $r_{e}$ typically lie between 5 and 15\% \citep{card2001estimating,Psacharopoulos2018,Oreopolous2020}.
Instead of estimating $r_{e}$, we assume that $r_{e}$, while unknown,
lies in this range. Then by instrumenting for friendships only, we
derive upper and lower bounds for the (causal) returns to friendships.

Let $\theta=(r_{f},\beta',\gamma',\alpha',\lambda')'$ denote the
vector of parameters, where $\alpha=(\alpha_{g},\forall g)'$ is the
vector of grade fixed effects, and $\lambda=(\lambda_{s},\forall s)'$
is the vector of school fixed effects. Let $\tilde{X}_{i}$ denote
the vector of friendships $F_{i}$ and other covariates (individual
characteristics $X_{i}$, school-grade mean characteristics $\bar{X}_{gs}$,
and grade and school dummies), and $Z_{i}$ the vector that consists
of the instruments for $F_{i}$ (homophily measures) and the covariates.
Ideal instruments for friendships satisfy two conditions: (i) they
predict friendships $F_{i}$ ($\mathbb{E}[Z_{i}\tilde{X}'_{i}]$ has
full column rank); and (ii) they are excluded from the earnings equation
($\mathbb{E}[Z_{i}\epsilon_{i}]=0$). Instruments for friendships
can also predict education. The exclusion restriction is satisfied
if the instruments do not affect earnings except through friendships
and education. This exclusion restriction implies the moment condition
\begin{equation}
\mathbb{E}[Z_{i}(Y_{i}-r_{e}E_{i}-\theta'\tilde{X}_{i})]=0.\label{eq:moment}
\end{equation}
Suppose that $W$ is a positive definite weighting matrix. If the
education return $r_{e}$ were known, the parameter $\theta$ would
satisfy $\theta=\mathbb{E}[Q_{i}(Y_{i}-r_{e}E_{i})]$, where $Q_{i}$
denotes the vector $(\mathbb{E}[\tilde{X}_{i}Z'_{i}]W\mathbb{E}[Z_{i}\tilde{X}'_{i}])^{-1}\mathbb{E}[\tilde{X}_{i}Z'_{i}]WZ_{i}$.
Because the true value of $r_{e}$ lies between $r_{e}^{l}=0.05$
and $r_{e}^{u}=0.15$, the parameter $\theta$ is bounded between
$\mathbb{E}[Q_{i}(Y_{i}-r_{e}^{l}E_{i})]$ and $\mathbb{E}[Q_{i}(Y_{i}-r_{e}^{u}E_{i})]$.

In practice, we can estimate the bounds by setting $r_{e}$ at $r_{e}^{l}$
and $r_{e}^{u}$ and estimating a GMM regression of $Y_{i}-r_{e}E_{i}$
on friendships $F_{i}$ and other covariates, using $Z_{i}$ as the
instrument. For example, if we set $r_{e}=r_{e}^{l}$ and regard $Y_{i}-r_{e}^{l}E_{i}$
as the dependent variable, then the GMM estimator yields an estimator
for the bound $\mathbb{E}[Q_{i}(Y_{i}-r_{e}^{l}E_{i})]$. The bound
$\mathbb{E}[Q_{i}(Y_{i}-r_{e}^{u}E_{i})]$ can be estimated similarly.\footnote{To ensure that the upper and lower bounds are estimated using the
same weighting matrix, we use the GMM estimator that assumes homogeneous
errors, that is, we set $W=\mathbb{E}[Z_{i}Z'_{i}]^{-1}$. This estimator
is equivalent to 2SLS.} For each component of $\theta$, we then obtain a consistent estimator
for the upper (lower) bound by taking the maximum (minimum) of the
two estimates of the component.\footnote{\label{fn:bound}Whether the upper or lower bound of a component of
$\theta$ is achieved at $r_{e}^{l}$ or $r_{e}^{u}$ depends on the
sign of the corresponding component of $\mathbb{E}[Q_{i}E_{i}]$.
For example, let $Q_{i,1}$ denote the first component of $Q_{i}$.
Then $r_{f}$ has the upper bound $\mathbb{E}[Q_{i,1}(Y_{i}-r_{e}^{l}E_{i})]$
if $\mathbb{E}[Q_{i,1}E_{i}]\geq0$ and $\mathbb{E}[Q_{i,1}(Y_{i}-r_{e}^{u}E_{i})]$
if $\mathbb{E}[Q_{i,1}E_{i}]<0$. The lower bound of $r_{f}$ can
be derived by swapping $r_{e}^{l}$ and $r_{e}^{u}$. In practice,
the components of $\mathbb{E}[Q_{i}E_{i}]$ can be recovered by regressing
education $E_{i}$ on friendships $F_{i}$ and other covariates, using
$Z_{i}$ as the instrument. The sign of the coefficient on friendships
determines whether the upper bound of $r_{f}$ is achieved when $r_{e}$
is high or low. A positive friendship coefficient implies that the
upper (lower) bound of $r_{f}$ is achieved at low (high) $r_{e}$.}

\textbf{Inference.} We follow \citet{Imbens2004} to construct a confidence
interval for any true value of $\theta$ that lies between the bounds.
These confidence intervals are asymptotically valid regardless of
whether $\theta$ is point identified or partially identified. If
a component of $\mathbb{E}[Q_{i}E_{i}]$ is $0$, the upper and lower
bounds for the corresponding component of $\theta$ coincide, and
this component of $\theta$ is point identified. In general, the components
of $\mathbb{E}[Q_{i}E_{i}]$ are not equal to $0$, the bounds do
not coincide, and $\theta$ is only partially identified.

We can assess whether $r_{f}$ is point identified or not by regressing
education $E_{i}$ on friendships $F_{i}$ and other covariates, using
$Z_{i}$ as the instrument, as noted in footnote \ref{fn:bound}.
If the coefficient of friendships is not significantly different from
zero, then we cannot reject the null that $r_{f}$ is point identified
-- the exact returns to education are not important because $F_{i}$
and $E_{i}$ are not correlated. In this case, we would not need to
instrument for education in order to obtain an unbiased estimate of
$r_{f}$. In our data, the coefficient on friendships is significant
from zero -- we can reject the null that $r_{f}$ is point identified.\footnote{This also ensures that the estimated upper and lower bounds are jointly
asymptotically normal, as required by \citet[Assumption 1(i)]{Imbens2004}.}

To be specific about the confidence intervals, denote the upper and
lower bounds of $\theta$ by $\theta_{u}$ and $\theta_{l}$. Let
$\hat{\theta}_{u}$ and $\hat{\theta}_{l}$ be the estimators for
$\theta_{u}$ and $\theta_{l}$ and $se(\hat{\theta}_{u})$ and $se(\hat{\theta}_{l})$
the standard errors of the estimators.\footnote{In Section \ref{sec:age_dist} we also consider an instrument constructed
using estimates from a pairwise regression. In general, standard errors
in two-step estimators should account for the presence of first-step
estimators. However, we are in a special case where the moment condition
in equation (\ref{eq:moment}) has a zero derivative with respect
to the first-step parameters (coefficients in the pairwise regression).
Therefore, the first-step estimation has no impact on the standard
errors in the second step and standard calculation of standard errors
is valid \citep[p.2179]{NEWEY1994}.} \citet[Lemma 4]{Imbens2004} proposed a $(1-\alpha)$ confidence
interval for the true value of $\theta$ that takes the form of $[\hat{\theta}_{l}-c\cdot se(\hat{\theta}_{l}),\hat{\theta}_{u}+c\cdot se(\hat{\theta}_{u})]$,
where the critical value $c$ satisfies $\Phi(c+\frac{\hat{\theta}_{u}-\hat{\theta}_{l}}{\max\{se(\hat{\theta}_{u}),se(\hat{\theta}_{l})\}})-\Phi(-c)=1-\alpha$,
with $\Phi(\cdot)$ being the cdf of the standard normal distribution.\footnote{The confidence intervals proposed by \citet{Imbens2004} require a
superefficient estimator of the length of an identified set (Assumption
1(iii) in their paper). Nevertheless, \citet[Lemma 3]{Stoye2009}
provides a simple sufficient condition for the superefficiency to
hold: the estimated upper and lower bounds are almost surely ordered.
Our estimated bounds are ordered because we take the maximum/minimum
of the two estimates. Therefore, by \citet[Proposition 1]{Stoye2009}
the confidence intervals in \citet{Imbens2004} are appropriate.} In practice, each component of $\theta$ has a critical value $c$,
and we need to solve for it numerically. Using the critical values
for each component of $\theta$, we can construct confidence intervals
for each component of $\theta$.\footnote{The STATA code for the bound estimates and confidence intervals can
be found in Appendix \ref{app:Code}.}

\textbf{Advantage of bounding.} Provided that the true return to education
is between 5 and 15\%, our approach provides consistent bound estimates
and valid confidence intervals for the returns to friendships. This
approach is preferable to a simple calibration that assumes the return
to education is known and equal to a particular value (e.g., $r_{e}=10\%$),
which yields a biased estimator if the true return to education is
different from the presumed value and provides no confidence intervals.

\subsection{\label{sec:age_dist}Using age distance as an instrument}

\citet{McPherson2001} report the results of various studies documenting
that in many settings (including schools) individuals of the same
age are much more likely to be friends. Based on this evidence, we
compute age distance for each pair of individuals as the absolute
difference between their ages to use as an instrument.

The distance between individuals $i$ and $j$ is defined at the pair
level -- that is, between two students of the same school and grade.
To construct an instrumental variable for in-degree at the individual
level, we average the pairwise distances over all $j$ that are in
$i$'s school-grade. We call this measure ``age distance''. In Appendix
\ref{app:pairwise IV} we consider another aggregating approach. We
run a Probit regression of friendships among pairs of students on
the pairwise age distance and basic controls. We then use the predicted
in-degree, calculated by the sum of the predicted friendships, as
an instrument for in-degree.

To be valid our instrument must operate only through friendships and
education (the two endogenous variables). Although age distance also
potentially affects educational attainment (\citealp{black2011too}),
the returns to friendships are still (partially) identified because
we estimate the returns to friendships after subtracting off the impact
of education from log earnings. The instrument satisfies the exclusion
restriction so long as it is uncorrelated with the remaining unobservables.

\subsection{\label{sec:id_var}Identifying variation and identifying assumptions}

We use a similar set of controls and identifying variation as \citet{Murphy2020}
who study the effects of academic rank on academic outcomes, and \citet{Cicala2017}
who show that where an individual ranks within a given social distribution
determines their choice of friends, behaviors, and outcomes. Given
that we control for cohort-mean age and own age, this leaves variation
in age distance among students with the same age and cohort-mean age
for identification. 

Table \ref{tab:variation_dist} illustrates this variation for a simple
example. In the first cohort, students are aged 13, 13.5, and 14 (mean
age 13.5). The second cohort has three students aged 13.2, 13.5, and
13.8 (same mean age of 13.5). In both cohorts, there is a student
aged 13.5. However, the age distance for this student is smaller in
cohort 2 (0.3 years) compared to cohort 1 (0.5 years). In the absence
of coordination effects, our model predicts that the student in cohort
2 will have more friends than the student in cohort 1 because they
are closer in age to their peers. The greater variance in the distribution
of ages in cohort 1 compared to cohort 2 results in a greater age
distance in cohort 1 (0.67) than cohort 2 (0.4).

Following \citet{Bifulco2011classmates}, we document in Table \ref{tab:dist_resid}
that there is sufficient variation of this kind in our data after
accounting for the basic set of controls. The variation in age distance,
measured by the standard deviation (s.d. 0.435), is about halved by
the inclusion of individual age (s.d. 0.2). But 45\% of the original
variation (s.d. 0.2) remains after including the full set of controls,
regardless of whether we control for grade FE or mean age in the cohort.
This residual variation in age distance is significant, and larger
than the residual variation in \citet{Bifulco2011classmates}.

A key identifying assumption is that conditional on our basic controls,
age distance does not predict earnings except through its effects
on education and friendships (exclusion restriction). Because the
identifying variation in age distance is generated from differences
in the age distribution across cohorts, we must assume that individuals
do not sort into schools and grades based on the variance in age --
that is, parents do not care or know about age distance and do not
sort using this criterion. 

Our identification assumption is not violated if parents want to place
their child in a cohort where the child is older relative to their
classmates (the practice sometimes referred to as red-shirting). In
our previous example, both cohorts have the same mean age, which means
that parents would be indifferent between the two cohorts because
the child\textquoteright s age relative to the cohort mean will be
identical in both cohorts. However, age distance will be on average
smaller in cohort 2. Alternatively, parents may care about their child\textquoteright s
\textit{age rank} within a cohort. In our example, a 13.7 year old
child would still be indifferent between the two cohorts, because
they would be the second oldest in both cohorts. But this student
would be closer in age to their peers in cohort 2, and we expect they
would have more friends.

\subsubsection{\label{sec:exclusion}Empirical Support for the Exclusion Restriction}

While we cannot test the exclusion restriction directly, we conduct
a series of placebo tests using observable pre-determined characteristics
(parental education, parental and own nativity, religion, birthweight
and breastfeeding, height, disabilities, etc.): age distance should
not predict these variables conditional on our basic controls. We
select the variables using two criteria. First, they are mostly determined
early in life, before adolescent friendships are formed. Second, the
prior literature and our data suggest they determine earnings.\footnote{We check that these variables predict earnings by regressing log earnings
with and without the predetermined variables, conditional on the basic
covariates (Columns 1 and 2 of Table \ref{tab:iv_test_joint}). We
reject the null that these variables do not jointly predict earnings
(p-value < 0.001). The $R^{2}$ increases from 0.059 to 0.066 (a 12\%
increase) with the addition of these variables, confirming that they
predict earnings.} Table \ref{tab:iv_test} shows that age distance does not predict
any of the 14 variables we consider, conditional on basic controls:
the coefficients on age distance are statistically insignificant in
all regressions. Moreover, if we regress age distance on all of these
variables (Column 3 of Table \ref{tab:iv_test_joint}), we cannot
reject the null that they do not predict age distance, conditional
on the basic controls (the p-value of the joint F-test is 0.56).

These placebo tests show that many important pre-determined characteristics
that predict earnings are not statistically associated with age distance,
providing support for the exclusion restriction.

\subsection{First stage results: the effect of age distance on friendships}

Table \ref{tab:first_stage} documents that age distance has a negative
and statistically significant effect on in-degree. If the average
age distance between a student and the students in their school-grade
increases by one year, the student has one less friend (Column 1).
Increasing the age distance by one standard deviation (0.44) lowers
the number of friends by about 0.44 friends, a 13\% decline relative
to the mean (3.4). The F-statistic is about 94, well above the standard
thresholds required to rule out weak instruments and large enough
for standard t-statistics in the second stage to be valid \citep{Lee2022}.
These results are similar if we estimate the first stage using the
pairwise data where an observation is a potential link between two
students in the same school-grade (22 million potential links). Using
a Probit probability model, we show that pairwise age distance is
a strong predictor of friendships between two individuals, even after
controlling for the age of the nominated individual in the pair and
the mean age in the cohort (Column 1 of Table \ref{tab:pair_dist}).\footnote{We do not control for the age of the individual who nominates the
friendship because otherwise there is no variation in the pairwise
distance.}

\subsection{Monotonicity}

In settings with heterogeneous treatment effects, IV estimates are
interpretable only under monotonicity: the endogenous variable (number
of friends) must be weakly monotonic in the instrument (age distance)
for all individuals. From a theoretical standpoint, increasing an
individual's homophily level has ambiguous predictions on socializing
and the number of friends because of coordination effects (Proposition
\ref{prop:homophily}). For example, suppose that an adolescent is
placed in a cohort with adolescents that are older instead of being
of the same age. If older individuals socialize more than younger
individuals, then the adolescent may socialize more and accumulate
more friends in the older cohort because it is more productive to
socialize in this group, despite the larger age difference.

Although monotonicity remains an untestable assumption, we empirically
investigate whether it appears to be violated. Figure \ref{fig:first_stage}
shows a bin-scatter plot of friendships and age distance, in the raw
data (left plot) and with basic controls (right plot). The slope is
negative. A non-parametric plot (Figure \ref{fig:first_stage_nonpara})
further confirms that the relationship is weakly decreasing: age distance
does not increase the number of friendships.\footnote{There is a small portion where this is not true but this occurs for
very high values of age distance that are uncommon.} Column 2 of Table \ref{tab:first_stage} shows that age distance
squared has a negative but insignificant effect on friendships. In
the pairwise first stage, age distance squared is significant, so
the relationship is not exactly linear (Column 2 of Table \ref{tab:pair_dist}).
However, the coefficients still imply a monotonic relationship: these
negative coefficients imply that the curve is strictly decreasing
and concave.

We also investigate whether the effect of age distance is symmetric.
Column 3 of Table \ref{tab:first_stage} shows that it is not: it
is more detrimental --- from the point of view of making friends
--- to be young among older peers than it is to be older among young
ones (this is also true in the pairwise estimation, see Column 3 of
Table \ref{tab:pair_dist}). However, age distance is still negatively
correlated to friendships regardless of whether peers are older or
younger.

\citet{Angrist1995} discuss what is needed for the monotonicity assumption
to be met in the case of a continuous treatment. A testable implication
of monotonicity is that the CDFs of the treatment (in-degree) at different
levels of the instrument (age distance) should not cross. A visual
representation of this test is given in Figure \ref{fig:cdf_df}.
We plot the difference in the CDF of in-degree as age distance increases
by one unit.\footnote{Let $X$ denote in-degree and $Z$ age distance. We estimate the CDF
difference $\mathbb{E}[1\{X\leq x\}|Z=z']-\mathbb{E}[1\{X\leq x\}|Z=z]$
for all $x$ and $z'\geq z$ by regressing the indicator variable
$1\{X\leq x\}$ on $Z$. The coefficient of $Z$ provides an estimate
of the CDF difference by one unit increase in $Z$ \citep{Rose2021}.} The CDF differences are non-negative throughout the support, whether
we control for the basic covariates or not. A formal test of this
no-crossing condition is provided by \citet{Barrett2003}. We split
the data evenly into high and low values of age distance. The null
hypothesis is that the distribution of in-degree with high age distance
is first-order stochastically dominated by the distribution of in-degree
with low age distance.\footnote{Let $N_{H}$ and $N_{L}$ denote the number of individuals with age
distance higher and lower than the median in the sample. Let $F_{H}(x)$
and $F_{L}(x)$ denote the CDFs of in-degree for the subgroups with
high and low age distance. The null and alternative hypotheses are
$H_{0}:F_{H}(x)\geq F_{L}(x)$ for all $x$ and $H_{1}:F_{H}(x)<F_{L}(x)$
for some $x$. \citet{Barrett2003} propose the test statistic $\hat{S}=\left(\frac{N_{H}N_{L}}{N_{H}+N_{L}}\right)^{1/2}\sup_{x}(\hat{F}_{L}(x)-\hat{F}_{H}(x))$,
where $\hat{F}_{H}(\cdot)$ and $\hat{F}_{L}(\cdot)$ are estimators
of $F_{H}(\cdot)$ and $F_{L}(\cdot)$. They suggest that the p-value
can be computed by $\exp(-2(\hat{S})^{2})$.} The p-values for the raw and residual distributions are 0.927 and
0.993 respectively, suggesting that we cannot reject the null, further
supporting the monotonicity assumption.

\section{IV Results: The Causal Returns to Friendships}

We now turn our attention to estimating the causal returns to friendships.

\subsection{Reduced form results: homophily and earnings in adulthood}

Figure \ref{fig:earn_iv} shows the correlation between log earnings
and age distance in a bin-scatter plot. Greater age distance is associated
with lower earnings in the raw data (left plot), and with basic controls
(right plot). Conditional on basic controls, individuals in cohorts
with more dissimilar peers in terms of age have lower earnings as
adults: increasing age distance by one standard deviation (0.44) lowers
earnings by 7.5\% (Column 1 of Table \ref{tab:reduced_form}). This
relationship is statistically significant.

\subsection{The causal returns to friendships: IV results}

Table \ref{tab:second_stage_main} reports the estimated bounds on
the returns to adolescent friendships. The OLS estimates for in-degree
are reported in Column 1 for reference, where the return to one more
friend is 0.025. If in-degree is treated as endogenous but education
is not, the estimate for in-degree increases to 0.12 (Column 2). If
we take a calibration approach, where we assume the return to education
is 10\%, the returns to in-degree remain at 0.11 (Column 3).\footnote{This estimate is obtained by running a regression on in-degree, controlling
for covariates and instrumenting for in-degree, where the dependent
variable is given by log earnings minus 10\% times years of schooling.}

Our main specification, which allows for the returns to education
to vary anywhere from 5 to 15\%, bounds the returns to friendships
from 0.093 to 0.137 (Column 4). The confidence interval (CI) for the
returns does not include zero -- these estimates are also statistically
significant. The upper bound corresponds to the lower return to schooling,
and vice versa. This is because the correlation between in-degree
and schooling is positive (Figure \ref{fig:edu_friend} and Table
\ref{tab:cov_edu=000026indeg}): if we regress education on in-degree,
controlling for covariates and instrumenting for in-degree, the coefficient
on in-degree is 0.44 and statistically significant. The significant
coefficient also implies that the returns to friendships are not point
identified.

\textbf{Robustness.} If we use a Probit probability model to predict
links using the pairwise data and then use the predicted in-degree
as an instrument, the bounds are similar and range from 0.065 to 0.096
(Column 5), although the estimates are not significant.

The results are robust to using alternative measures of friendships
(Table \ref{tab:second_stage_measures}), including reciprocated friendships,
which suggests that we are not simply capturing the returns to popularity.

We also check if the results are robust to the inclusion of additional
controls (Table \ref{tab:second_stage_robust}). Column 1 reproduces
our preferred specification for reference and shows bounds of 0.093--0.137.
Our estimates are similar (and the CI does not include 0) if we control
for age rank (Column 2), predetermined individual-level controls (the
ones we used for the placebo tests in Section \ref{sec:exclusion},
Column 3), variables capturing current SES including parental income
and living with parents (Column 4), or their respective cohort-level
means (Column 5). The last two regressions are potentially problematic
because income (or living with parents) is endogenous (it is not truly
predetermined) but it is nevertheless reassuring that the results
are similar.

Table \ref{tab:second_stage_bully} investigates another potential
violation of the exclusion restriction. Perhaps in settings where
age distance is larger, there is greater bullying. Because bullying
can affect mental health \citep{Arseneault2010} which in turn affects
earnings, our instrument could affect earnings through this additional
channel. We test this by controlling for measures of social cohesion
in adolescence. These controls do not individually or jointly affect
the estimated bounds.

\textbf{Magnitudes}. The magnitude of the returns to friendships ranges
from 0.07 to 0.14 across specifications. Since recent estimates of
the returns to schooling are on the larger side (more than 10\%),
the more likely bound for the returns to friendships is the lower
bound, which hovers around 0.07--0.09. Thus an increase of one standard
deviation in the number of friends (3.2) increases earnings by 22\%
-- a large and economically significant return. By comparison, a
one standard deviation increase in years of schooling (2.1) would
increase earnings by 21\% so the effect size is similar.\footnote{The implied elasticity of earnings with respect to friendships is
0.24, whereas the elasticity of earnings with respect to education
is 1.5, assuming the return to education is 10\%.} While the return to in-degree seems large, it represents the returns
to having a friend during adolescence -- we are not measuring the
returns to one year of friendship but the returns to having a friend.
These friendships likely last many years, potentially into adulthood
(Section \ref{sec:discussion}).

Our bounds are larger than the point estimates in \citet{conti2013popularity}
of around 2\%. In addition to methodological differences (they do
not account for the endogeneity of education and friendships), they
study men who were seniors in high school in Wisconsin in 1957, whereas
the Add Health data is nationally representative and surveys students
in middle or high school in 1994--95, at a time when the returns
to education and other individual traits in the labor market are much
larger. They measure labor market outcomes 35 years later whereas
we observe them 15 years later. The returns to friendships in adolescence
may attenuate over time, yielding larger returns when respondents
are surveyed earlier in their careers.

\subsection{Why are OLS estimates downward biased?}

Our preferred bounds do not include the OLS estimate of 0.025.\footnote{However, the confidence interval for the returns to friendships {[}0.008,
0.224{]} includes the confidence interval for the OLS estimator {[}0.018,
0.031{]}, so we cannot reject the null that the IV and OLS estimates
are the same.} We explore two reasons why the estimated bounds are larger than the
OLS estimate: omitted variable bias and measurement error.\footnote{IV estimates might also exceed OLS estimates if treatment effects
are heterogeneous. OLS estimates in Table \ref{tab:subsample} suggest
that while there is some heterogeneity in the returns to friendships
across subpopulations, it is too small to explain the discrepancy
between OLS and IV estimates.}

\textbf{Omitted variable bias.} Omitted variable biases in our model
are in general ambiguous (Proposition \ref{prop:ols}). The sign of
the OLS bias is determined by the correlation between in-degree and
the error term $\epsilon_{i}$.\footnote{Strictly speaking, the bias is given by the covariance between in-degree
and $\epsilon_{i}$, multiplied by the inverse of a matrix $X'X$,
where $X$ consists of residualized education and in-degree. In our
sample this inverse matrix has positive diagonal elements and relatively
small off-diagonal elements (Table \ref{tab:cov_edu=000026indeg}).
Therefore, the sign of the bias is mostly determined by the correlation
between in-degree and $\epsilon_{i}$.} A downward (upward) bias arises if in-degree is positively correlated
with an unobserved factor that negatively (positively) affects earnings.

The data suggest possible omitted factors that could explain our findings:
partying and drinking. Alcohol consumption among adolescents is largely
motivated by its capacity to facilitate social interactions (\citealp{feldman1999alcohol};
\citealp{kuntsche2005young}), but it may have detrimental impacts
on other outcomes. We find that drinking alcohol increases friendships
but lowers GPA and the odds of working in cognitively demanding jobs
(Table \ref{tab:production_functions}).\footnote{We use data from the O{*}NET to construct indices of the extent to
which occupations require social or cognitive skills and match these
scores to a person's occupation. See Appendix \ref{app:ONET} for
details.} Drinking also increases depression and lowers self-reported health
in adulthood. However, time spent with friends is associated with
more friends and better health, lowers the odds of working in jobs
with high cognitive demands, but it does not lower GPA. Thus \textit{certain}
social behaviors, like drinking, increase friendships but lower cognitive
productivity and lower health, both of which likely affect wages.

Additional results suggest that OLS is downward biased. Table \ref{tab:ols_build}
shows that the addition of covariates in OLS regressions of earnings
on grade in-degree increases the estimated returns to friendships.
We control for education only (Column 1), and then we progressively
control endowments (IQ and extroversion - Column 2), personal demographics
(age, gender, race - Column 3), mean characteristics of one's peers
(Column 4), grade fixed effects (Column 5) and school fixed effects
(Column 6). The model with all the controls yields the largest coefficient.
This evidence suggests that individual characteristics and school
fixed effects capture preferences and environments that cause OLS
to be downward biased.

\textbf{Measurement error.} Classical measurement error in friendships
would attenuate the OLS coefficient. The direction and magnitude of
the OLS bias depend on the ratio of the covariance between in-degree
and the true measure, divided by the variance of in-degree: if the
ratio is smaller than one then OLS is attenuated \citep{Bound2001}.
Suppose that high-engagement friendships are the ``true'' number
of friends and grade in-degree is a poor proxy. What would the OLS
bias be in this case? In our sample, the covariance between (residualized)
grade in-degree and (residualized) high-engagement degree is 2.28.
The (residualized) variance of grade in-degree is 8.67 (Table \ref{tab:cov_network_measures}).
Their ratio is about 0.26, which suggests that a consistent estimate
could be 4 times larger than OLS, roughly 0.09, similar to our IV
estimates.

\textbf{Limitations}. Measurement error may also come from the fact
that individuals in the survey are only allowed to list up to five
friends of each gender, generating censoring -- a non-classical form
of measurement error. In this case our IV might be inconsistent --
an IV estimate is consistent only if the instrument is uncorrelated
with measurement error \citep{Bound2001}.

\section{\label{sec:discussion}Discussion: Which Friends Matter and How do
Friendships Affect Earnings?}

We end with an informal discussion of two important issues. First,
do all friendships yield positive returns in the labor market? Second,
what are the mechanisms by which adolescent friendships increase earnings?
While we cannot fully answer these questions in a causal manner, we
discuss some suggestive evidence.

Figure \ref{fig:friend_types} plots the non-parametric relationship
between in-degree and earnings, for different measures of in-degree
suggested by prior studies: same vs. opposite gender (\citealp{McDougall2007};
\citealp{Hall2011}), high vs. low engagement \citep{Gee2017strongties},
high vs. low SES (\citealp{LavySand2019}; \citealp{Fletcher2020};
\citealp{Chetty2022}), and disruptive vs. non-disruptive peers \citep{carrell2018disruptivepeers}.
The friendships that are expected to have higher returns indeed appear
to have higher returns: it pays off more to be friends with people
of the same gender, to have more close friends, to have more friends
who are not disruptive or come from higher SES families. However,
all friendships (except for disruptive ones) appear to have positive
returns, albeit smaller ones. Table \ref{tab:ols_types} confirms
these descriptive patterns using OLS. Unfortunately, our instruments
are not powerful enough to produce precise estimates of different
types of friendships so further work is needed in this area.\footnote{Although all the first stages are strong (Table \ref{tab:first_stage_types}),
we do not have sufficient variation to \textit{separately} identify
the effects of, e.g., male and female friendships.}

Why do friendships in adolescence matter for adult earnings? We do
not have estimates of the causal effect of education on potential
mediators, so we cannot estimate the ideal IV bounds. Instead, we
summarize the suggestive evidence from OLS regressions. Friendships
formed in adolescence are associated with higher likelihood of working
(Column 1 of Table \ref{tab:social_mechanism}), consistent with the
previous literature that friends help workers find jobs (\citealp{Dustmann2015}).
Adolescent friendships are associated with a lower likelihood of working
in repetitive occupations (Column 2) and a greater likelihood of working
in supervisory roles (though this relationship is not statistically
significant, Columns 3). Individuals with more friends are more likely
to be employed in occupations that require greater social skills,
which on average pay larger wages (Column 4). Individuals with more
friends appear to have greater social and management skills and broader
networks as adults. They have more friends in adulthood (Column 5)
and are more likely to be married (Column 6). Controlling for extroversion
in adolescence, those with more adolescent friends are more likely
to be extroverted in adulthood (Column 7). Individuals with more friends
also have greater GPAs (Columns 1 of Table \ref{tab:cognitive_mechanism}),
are less likely to get in trouble while they are in school (Columns
2 and 3), and ultimately are more likely to work in jobs that require
higher cognitive skills (Column 4). They are also less depressed and
healthier (Columns 5 and 6). Thus the returns to friendships in the
labor force do not operate uniquely through social skills and networks
-- friendships also help in the formation of other cognitive and
non-cognitive skills that are rewarded in the labor market.

\section{Conclusion}

We show that individuals that have more friends in adolescence have
higher earnings in adulthood, partly because they are employed in
higher paying occupations that require higher social and cognitive
skills, and partly because they turn into more socially connected
adults that are also more extroverted.

Our results confirm recent findings emphasizing the importance of
adolescence as a formative period for socio-emotional outcomes \citep{Jackson2020}
and particularly for friendships \citep{Denworth2020}. Our findings
also suggest that when students are more similar to each other they
are more likely to form friendships. Therefore, how students are allocated
in groups in schools has important long-term consequences. Re-structuring
classrooms to be more homogeneous in age would appear to benefit all
students involved, from the point of view of friendships and adult
earnings. An important direction for future research is to investigate
which other contexts and student compositions are most conducive to
the creation of friendships.

The importance of socialization during schooling years has further
implications for higher education policy. Many colleges and universities
face criticism for investing in infrastructure for non-academic, recreational
facilities that are believed to contribute to increasing tuition \citep{Jacob2018college}.
These investments are more reasonable in the presence of high returns
to socialization. Our results help rationalize why, for instance,
Greek fraternities and sororities persist, despite the fact that they
can detract from strictly academic endeavors. While partying may decrease
education, particularly if it is associated with drinking, it does
not necessarily decrease earnings. Overall if schools can promote
productive social activities (like working together), it might be
possible to improve both students' social connections and their educational
attainment.

Like recent work (\citealp{Xiang2019}; \citealp{Heckman2012}), our
findings suggest that paying excessive attention to traditional education
measures like test scores might lower the long-term outcomes and well-being
of individuals because they reduce investments in other important
forms of human capital. They also suggest that educational interventions
that are becoming common, such as remote learning, might deter from
social capital formation and result in lost lifetime earnings. In
contrast, interventions that target soft skills might have large returns,
partly through their effects on network formation. Indeed an emerging
literature shows that interventions targeting non-cognitive skills
among young adults can have large returns (\citealp{Katz2022}, \citealp{Heller2017}).
Our paper complements this research by documenting the importance
of having friends, separately from social skills.

There are many unanswered questions that remain. For example, we don't
know much about which environments best foster friendships while maintaining
academic performance. We provide some evidence that not all friendships
matter equally in the labor market but our evidence is only suggestive.
Finally, we only track the impact of adolescent friendships -- not
whether earlier friendships matter, how social networks evolve from
adolescence onward and how adult networks affect labor markets. These
important questions are left to future research.

\bibliographystyle{jpe}
\bibliography{ms.bbl}
\newpage{}
\begin{figure}
\centering\caption{\label{fig:returns_edu_friend}Log Earnings in Adulthood and Friendships
in Adolescence}

\includegraphics{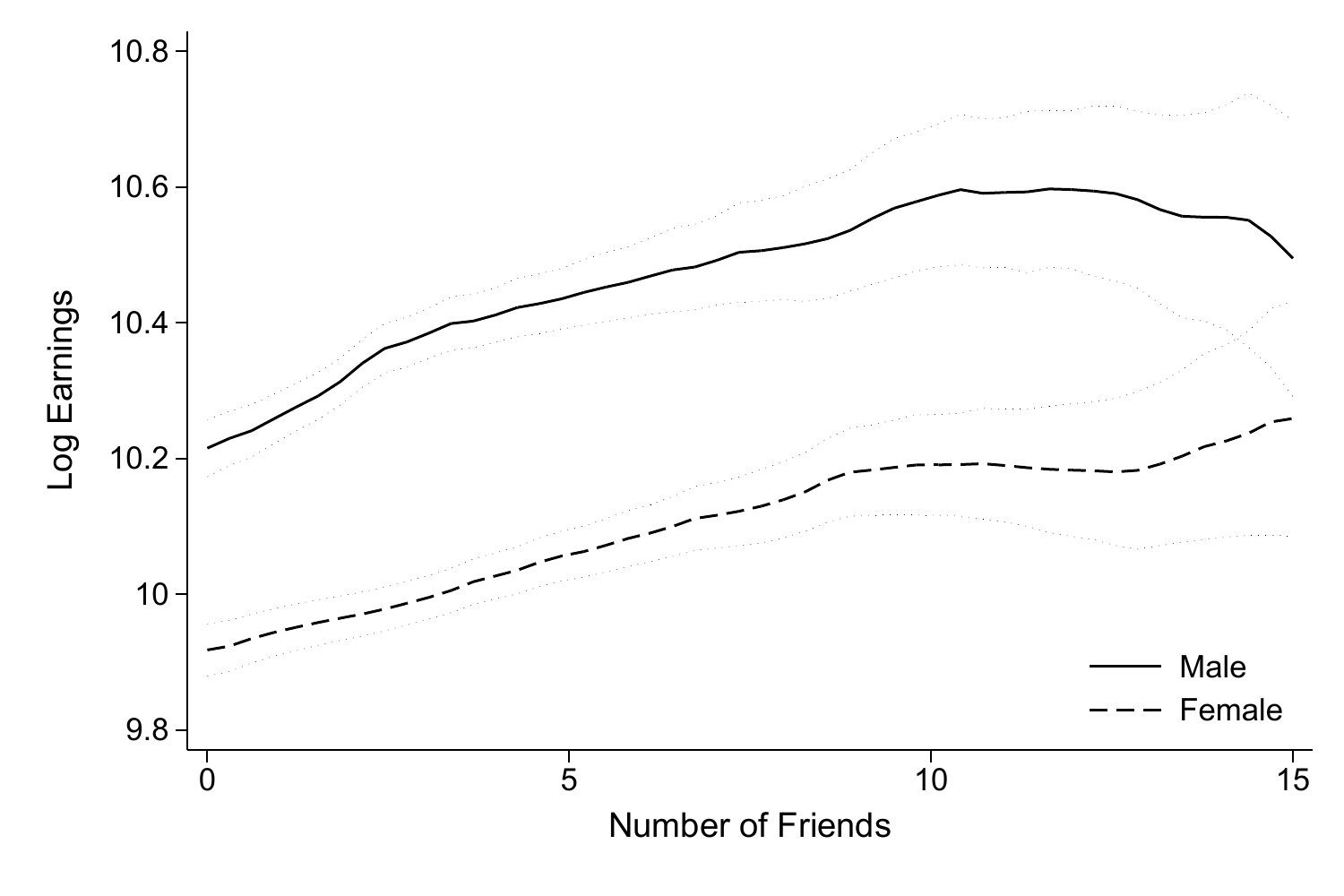}\caption*{Note: Add Health restricted-use data. Local polynomial nonparametric plots and 95\% confidence intervals of log earnings on number of friendships in adolescence measured by grade in-degree (see data section for explanations on how grade in-degree is defined). The series separate males and females with non-zero earnings.}
\end{figure}
\newpage{}
\begin{figure}
\centering

\caption{Homophily, Friendships, and Earnings: First Stage and Reduced Form}

\subfloat[\label{fig:first_stage}Friendships and Homophily in Age]{\includegraphics[scale=0.8]{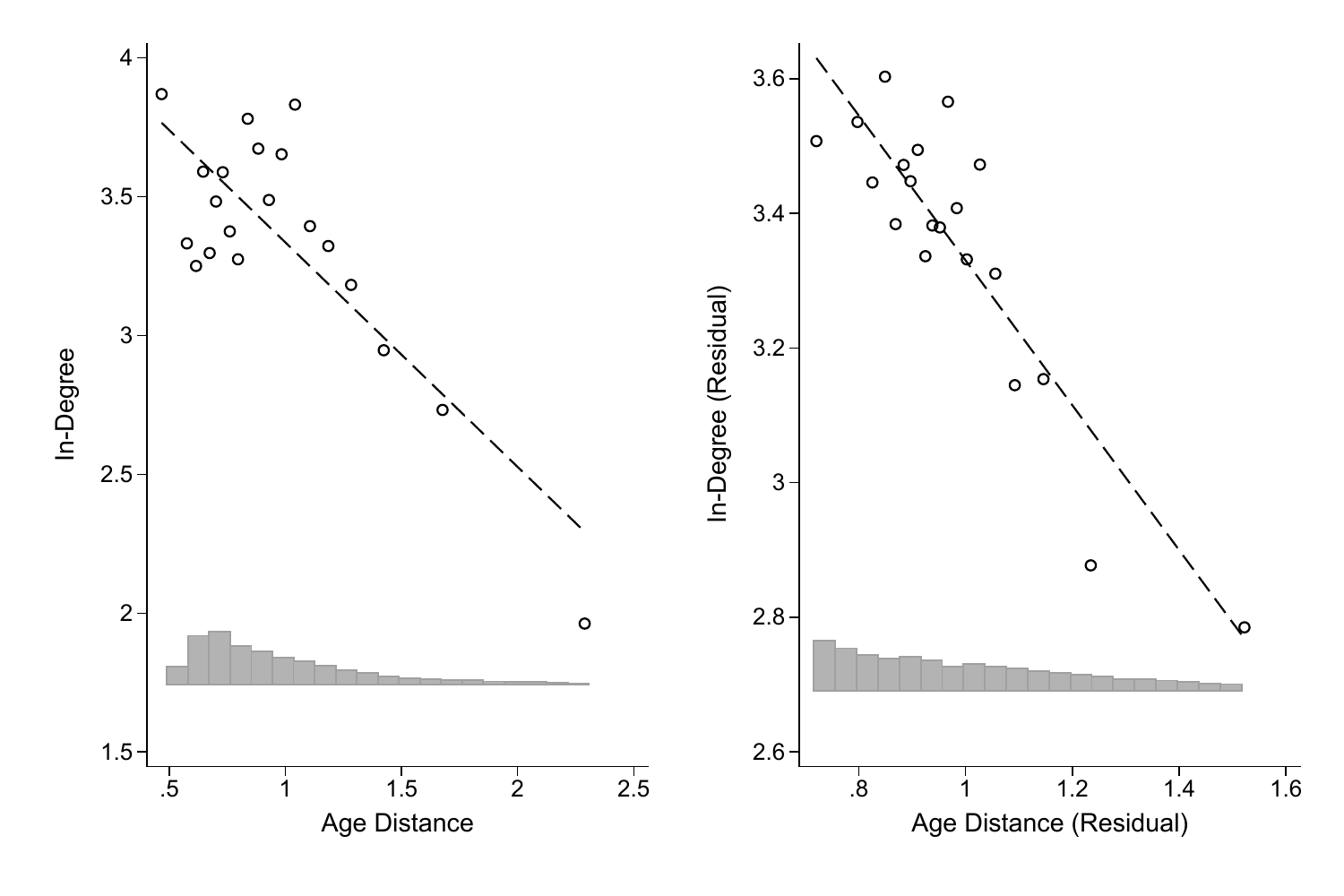}}

\subfloat[\label{fig:earn_iv}Log Earnings and Homophily in Age]{

\includegraphics[scale=0.8]{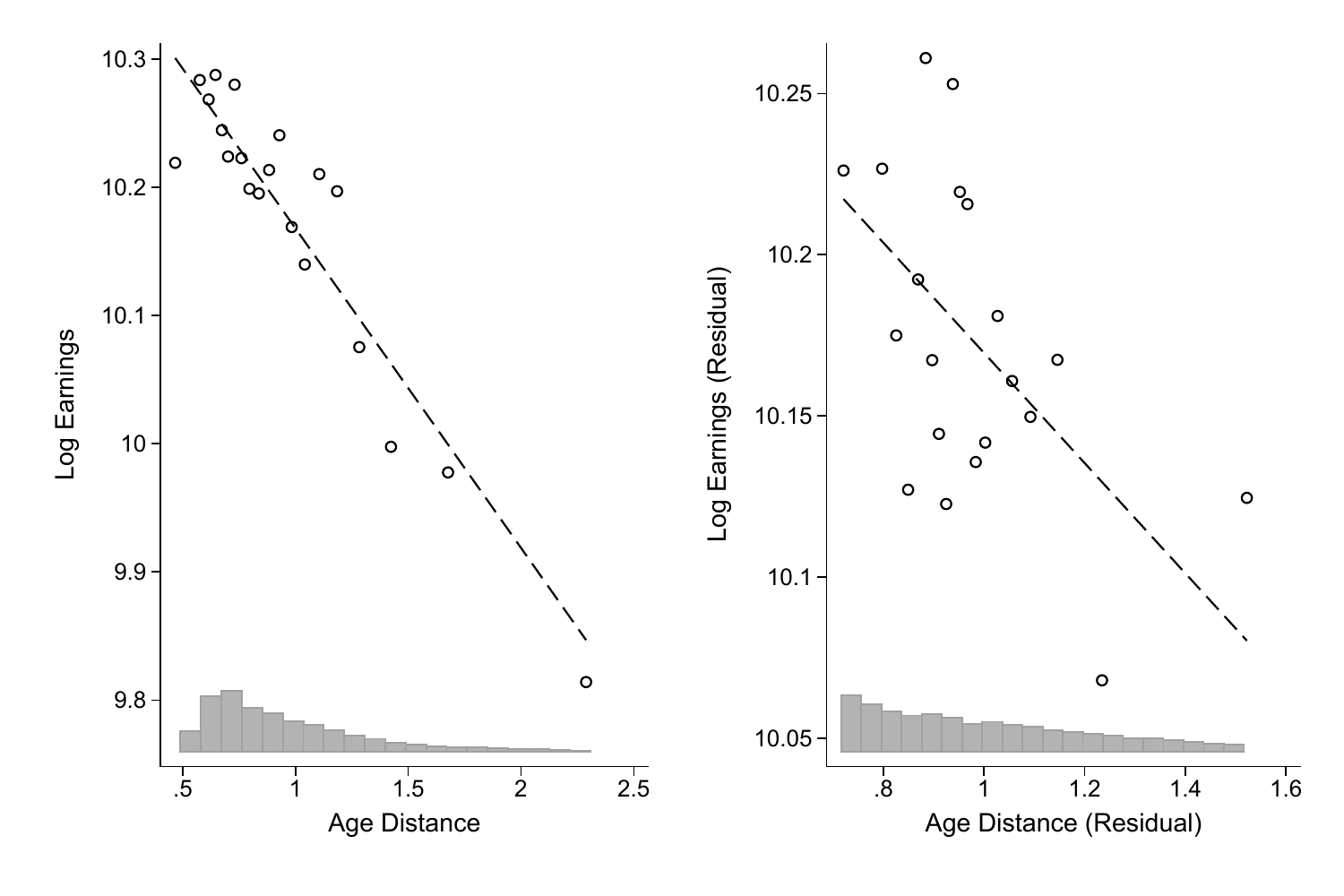}}

\caption*{Note: Add Health restricted-use data. Bin-scatter plots of grade in-degree (panel a) and log earnings (panel b) on age distance with the histogram of age distance (gray bars). In both panels, the left plot uses the original variables, and the right plot uses their residuals after removing individual characteristics (age, IQ, and indicators for whether the student is extrovert, female, and white), cohort-level characteristics (mean age, mean IQ, fraction extrovert, fraction female, and fraction white), grade fixed effects, and  school fixed effects.}
\end{figure}
\newpage{}
\begin{table}[t]
\centering \begin{threeparttable}\caption{\label{tab:sum_stats}Summary Statistics}
{\footnotesize{}\begin{tabular}{l*{5}{c}} \toprule                     &\multicolumn{1}{c}{(1)}&\multicolumn{1}{c}{(2)}&\multicolumn{1}{c}{(3)}&\multicolumn{1}{c}{(4)}&\multicolumn{1}{c}{(5)}\\                     &All Students&   Extrovert&         Shy&       Hi IQ&      Low IQ\\ \midrule Annual Earnings (W4)&     38073.2&     37753.8&     35312.9&     41988.6&     33578.1\\                     &   (46674.4)&   (43283.4)&   (38818.7)&   (45419.6)&   (45555.9)\\ \addlinespace School In-Degree (W1)&       4.359&       4.733&       3.928&       4.682&       3.998\\                     &     (3.729)&     (3.891)&     (3.386)&     (3.896)&     (3.507)\\ \addlinespace Grade In-Degree (W1)&       3.352&       3.688&       3.069&       3.645&       3.024\\                     &     (3.153)&     (3.347)&     (2.847)&     (3.271)&     (2.975)\\ \addlinespace Grade Out-Degree (W1)&       3.285&       3.484&       3.250&       3.638&       2.883\\                     &     (2.688)&     (2.747)&     (2.620)&     (2.694)&     (2.618)\\ \addlinespace Grade Reciprocated Degree (W1)&       1.388&       1.491&       1.311&       1.597&       1.147\\                     &     (1.564)&     (1.616)&     (1.469)&     (1.640)&     (1.429)\\ \addlinespace Grade Network Size (W1)&       5.248&       5.681&       5.009&       5.686&       4.760\\                     &     (3.843)&     (4.025)&     (3.555)&     (3.887)&     (3.722)\\ \addlinespace Years of Schooling (W4)&       14.76&       14.82&       14.73&       15.37&       14.08\\                     &     (2.125)&     (2.122)&     (2.110)&     (1.930)&     (2.126)\\ \addlinespace GPA (W2)            &       2.839&       2.850&       2.818&       2.997&       2.649\\                     &     (0.746)&     (0.739)&     (0.758)&     (0.735)&     (0.714)\\ \addlinespace IQ (W1)             &       101.5&       102.1&       100.6&       112.4&       89.30\\                     &     (13.95)&     (13.50)&     (14.59)&     (7.447)&     (9.408)\\ \addlinespace Extrovert (W2)      &       0.650&           1&           0&       0.663&       0.632\\                     &     (0.412)&         (0)&         (0)&     (0.407)&     (0.418)\\ \addlinespace Very Hard Study Effort (W1)&       0.381&       0.383&       0.419&       0.346&       0.419\\                     &     (0.486)&     (0.486)&     (0.494)&     (0.476)&     (0.494)\\ \addlinespace Some Study Effort (W1)&       0.511&       0.515&       0.483&       0.536&       0.484\\                     &     (0.500)&     (0.500)&     (0.500)&     (0.499)&     (0.500)\\ \addlinespace Frequently Hang with Friends (W1)&       0.387&       0.401&       0.344&       0.376&       0.399\\                     &     (0.487)&     (0.490)&     (0.475)&     (0.485)&     (0.490)\\ \addlinespace Sometimes Hang with Friends (W1)&       0.523&       0.514&       0.553&       0.541&       0.505\\                     &     (0.499)&     (0.500)&     (0.497)&     (0.498)&     (0.500)\\ \addlinespace Frequently Drink (W1)&       0.163&       0.158&       0.121&       0.174&       0.152\\                     &     (0.369)&     (0.365)&     (0.326)&     (0.379)&     (0.359)\\ \addlinespace Sometimes Drink (W1)&       0.299&       0.306&       0.265&       0.316&       0.283\\                     &     (0.458)&     (0.461)&     (0.441)&     (0.465)&     (0.450)\\ \addlinespace Age (W1)            &       16.06&       15.70&       15.83&       16.04&       16.05\\                     &     (1.670)&     (1.534)&     (1.546)&     (1.619)&     (1.716)\\ \addlinespace Age Distance (W1)   &       0.980&       0.964&       0.984&       0.903&       1.061\\                     &     (0.435)&     (0.421)&     (0.438)&     (0.341)&     (0.497)\\ \midrule Observations        &      10,605&       5,124&       2,765&       5,384&       4,741\\ \bottomrule \end{tabular} }{\footnotesize\par}

\begin{tablenotes}[flushleft]\footnotesize \item 
Note: Add Health restricted-use data. Standard deviations in parentheses. W1 stands for Wave I, etc. Due to missing values in IQ and Extrovert, the numbers of observations in Columns 2 and 3 do not sum up to that in Column 1, and the same for Columns 4 and 5.
\end{tablenotes}\end{threeparttable}
\end{table}
\newpage{}
\begin{table}
\centering

\caption{\label{tab:first_stage}What Predicts Friendships in Adolescence?
First-Stage Results}
\begin{threeparttable}{ \def\sym#1{\ifmmode^{#1}\else\(^{#1}\)\fi} \begin{tabular}{l*{3}{c}} \toprule                     &\multicolumn{1}{c}{(1)}&\multicolumn{1}{c}{(2)}&\multicolumn{1}{c}{(3)}\\                     &\multicolumn{1}{c}{In-Degree}&\multicolumn{1}{c}{In-Degree}&\multicolumn{1}{c}{In-Degree}\\ \midrule Age Distance        &     -1.0760\sym{***}&     -0.7023\sym{**} &                     \\                     &    (0.1111)         &    (0.2767)         &                     \\ Age Distance Squared&                     &     -0.1120         &                     \\                     &                     &    (0.0721)         &                     \\ Age Distance to Older Peers&                     &                     &     -2.0300\sym{***}\\                     &                     &                     &    (0.2822)         \\ Age Distance to Young Peers&                     &                     &     -0.2571         \\                     &                     &                     &    (0.3130)         \\ Age                 &      0.2895\sym{***}&      0.2534\sym{***}&     -0.5636\sym{*}  \\                     &    (0.0841)         &    (0.0813)         &    (0.3046)         \\ Mean Age            &     -0.3538\sym{*}  &     -0.3669\sym{*}  &      0.0784         \\                     &    (0.1982)         &    (0.2010)         &    (0.2314)         \\ \midrule F-stat of Homophily Measures&      93.846         &      49.464         &      72.139         \\ P-value             &       0.000         &       0.000         &       0.000         \\ $ R^{2}$            &       0.038         &       0.038         &       0.039         \\ Observations        &      10,605         &      10,605         &      10,605         \\ \bottomrule \end{tabular} }

\begin{tablenotes}[flushleft]\footnotesize \item 
Note: Add Health restricted-use data. In-degree refers to the number of friendships nominated by other students in the same school-grade. Age distance refers to the average age distance to other students in the same school-grade. Age distance to older (younger) peers refers to the average age distance to other students in the same school-grade who are older (younger) than the respondent. Additional covariates include individual characteristics (IQ and indicators for whether the student is extrovert, female, and white), cohort-level characteristics (mean IQ, fraction extrovert, fraction female, and fraction white), grade fixed effects, and  school fixed effects. Standard errors in parentheses are clustered at the school level.
\item * p < 0.10, ** p < 0.05, *** p < 0.01.
\end{tablenotes}\end{threeparttable}
\end{table}
\newpage\begin{landscape}
\begin{table}[t]
\centering \begin{threeparttable}\caption{\label{tab:second_stage_main}Labor Market Returns to Friendships:
Main Results}
\begin{tabular}{lccccc}
\toprule 
 & (1) & (2) & (3) & (4) & (5)\tabularnewline
 & OLS & IV & IV & IV & IV\tabularnewline
\midrule 
In-Degree & 0.0245 & 0.1242 & 0.1146 & {[}0.0926, 0.1367{]} & {[}0.0650, 0.0956{]}\tabularnewline
95\% CI & {[}0.0177, 0.0312{]} & {[}0.0130, 0.2354{]} & {[}0.0155, 0.2138{]} & {[}0.0075, 0.2242{]} & {[}-0.0308, 0.1936{]}\tabularnewline
\midrule
First Stage for In-Degree &  & Aggregate & Aggregate & Aggregate & Pairwise Probit\tabularnewline
Education Endogenous & N & N & Y & Y & Y\tabularnewline
Education Returns & 10.30\% & 7.83\% & 10\% & 5-15\% & 5-15\%\tabularnewline
Observations & 10,605 & 10,605 & 10,605 & 10,605 & 10,605\tabularnewline
\bottomrule
\end{tabular}

\begin{tablenotes}[flushleft]\footnotesize \item 
Note: Add Health restricted-use data. The dependent variable is the log annual earnings. In-degree refers to the number of friendships nominated by other students in the same school-grade. Column 1 presents OLS estimates for the returns to in-degree. Column 2 presents estimates where in-degree is endogenous, but education is not (estimated returns to education are reported in the second to last row). Column 3 presents estimates where in-degree is endogenous and the returns to education is assumed to be 10\%. Columns 4-5 present estimates where in-degree is endogenous and the returns to education are bounded between 5 and 15\%. All specifications include individual characteristics (age, IQ, and indicators for whether the student is extrovert, female, and white), cohort-level characteristics (mean age, mean IQ, fraction extrovert, fraction female, and fraction white), grade fixed effects, and school fixed effects. The confidence intervals are constructed based on standard errors clustered at the school level.
\end{tablenotes} \end{threeparttable}
\end{table}
\newpage{}
\begin{table}[t]
\centering \begin{threeparttable}\caption{\label{tab:second_stage_robust}Labor Market Returns to Friendships:
Robustness Checks}
\begin{tabular}{lccccc}
\toprule 
 & (1) & (2) & (3) & (4) & (5)\tabularnewline
 & IV & IV & IV & IV & IV\tabularnewline
\midrule 
In-Degree & {[}0.0926, 0.1367{]} & {[}0.0908, 0.1347{]} & {[}0.0952, 0.1356{]} & {[}0.0951, 0.1355{]} & {[}0.0694, 0.1157{]}\tabularnewline
95\% CI & {[}0.0075, 0.2242{]} & {[}0.0068, 0.2212{]} & {[}0.0002, 0.2333{]} & {[}0.0030, 0.2303{]} & {[}0.0033, 0.1841{]}\tabularnewline
\midrule
Additional Individual Controls & N & Age Rank & Pre Controls & SES Controls & N\tabularnewline
Additional Cohort Mean Controls & N & N & N & N & SES Controls\tabularnewline
Observations & 10,605 & 10,605 & 10,605 & 10,605 & 10,605\tabularnewline
\bottomrule
\end{tabular}

\begin{tablenotes}[flushleft]\footnotesize \item  
Note: Add Health restricted-use data. The dependent variable is the log annual earnings. In-degree refers to the number of friendships nominated by other students in the same school-grade. All specifications include individual characteristics (age, IQ, and indicators for whether the student is extrovert, female, and white), cohort-level characteristics (mean age, mean IQ, fraction extrovert, fraction female, and fraction white), grade fixed effects, and school fixed effects. Column 2 controls for age rank, which refers to the ranking of the respondent's age in the school-grade. Column 3 controls for additional individual-level predetermined covariates, including number of siblings, mother's age at the student's birth, birth weight, height, and indicators for whether the mother was born in US, the father was born in US, the parents are Catholic, Baptist, the student was born in US, was breastfed, is mentally retarded, and has diabilibity. Column 4 controls for additional individual-level SES controls, including family income, mother's years of schooling, father's years of schooling, and indicators for whether the mother lives in the household and whether the father lives in the household. Column 5 controls for cohort-level means of the additional SES controls in Column 4. The instrument is age distance. The confidence intervals are constructed based on standard errors clustered at the school level. 
\end{tablenotes} \end{threeparttable}
\end{table}
\end{landscape}

\pagestyle{appendix}\newpage{}

\appendix
\begin{center}
	{\LARGE{}Online Appendix to}\\[1cm]{\LARGE{}Party On: The Labor Market
		Returns to Social}\\[0.25cm]{\LARGE{}Networks In Adolescence}\\[1cm]
	{\Large{}Adriana Lleras-Muney, Matthew Miller, Shuyang Sheng, and}\\[0.25cm]
	{\Large{}Veronica Sovero }\\[1cm]\setcounter{page}{1} 
	\numberwithin{figure}{section} \numberwithin{table}{section}
	\numberwithin{equation}{section} \numberwithin{assumption}{section}
	\par\end{center}

\section{\label{app:Additional-Tabs-Figs}Additional Figures and Tables}

\begin{figure}[H]
	\centering
	
	\caption{\label{fig:first_stage_nonpara}Friendships and Homophily in Age:
		Nonparametric Plots}
	\includegraphics[scale=0.99]{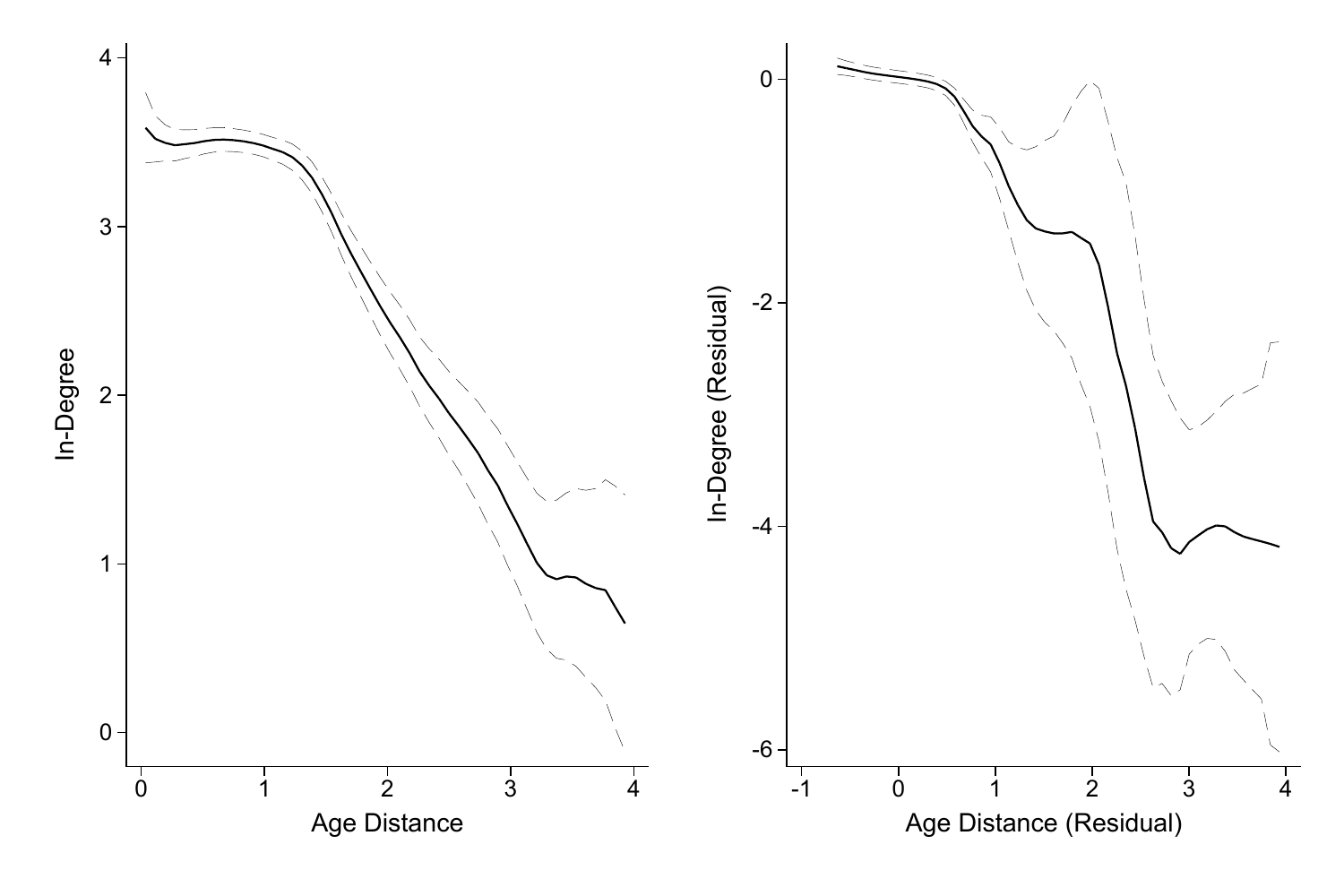}
	
	\caption*{Note: Add Health restricted-use data. Local polynomial nonparametric plots and 95\% confidence intervals of grade in-degree on age distance. The left plot uses the original variables. The right plot uses their residuals after removing individual characteristics (age, IQ, and indicators for whether the student is extrovert, female, and white), cohort-level characteristics (mean age, mean IQ, fraction extrovert, fraction female, and fraction white), grade fixed effects, and  school fixed effects.}
\end{figure}
\newpage{}
\begin{figure}
	\centering\caption{\label{fig:cdf_df}Difference in Friendship CDF by Age Distance}
	\includegraphics{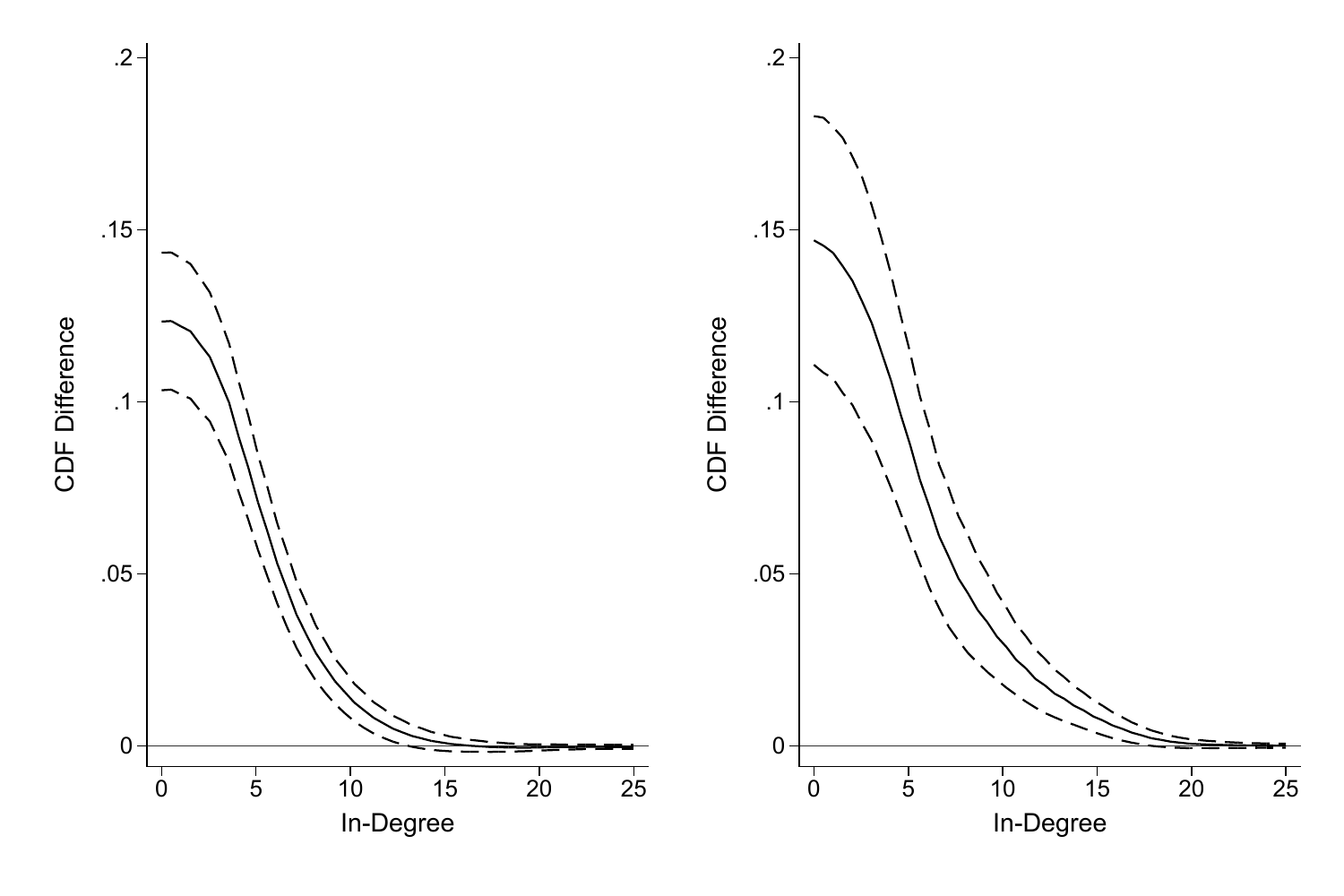}
	
	\caption*{Note: Add Health restricted-use data. Estimated difference in the CDF of grade in-degree as age distance increases by one unit. Dashed lines are 95\% confidence intervals. To be specific about the estimation procedure, let $X$ represent grade in-degree and $Z$ age distance. We estimate the CDF difference $\mathbb{E}[1\{X\leq x\}|Z=z']-\mathbb{E}[1\{X\leq x\}|Z=z]$ for all $x$ and $z'\geq z$ by regressing the indicator variable $1\{X\leq x\}$ on $Z$. The coefficient of $Z$ provides an estimate of the CDF difference by one unit increase in $Z$. The left plot does not include any covariates. The right plot controls for individual characteristics (age, IQ, and indicators for whether the student is extrovert, female, and white), cohort-level characteristics (mean age, mean IQ, fraction extrovert, fraction female, and fraction white), grade fixed effects, and  school fixed effects.}
\end{figure}
\newpage{}
\begin{figure}
	\centering\caption{\label{fig:edu_friend}Education and Friendships}
	\includegraphics{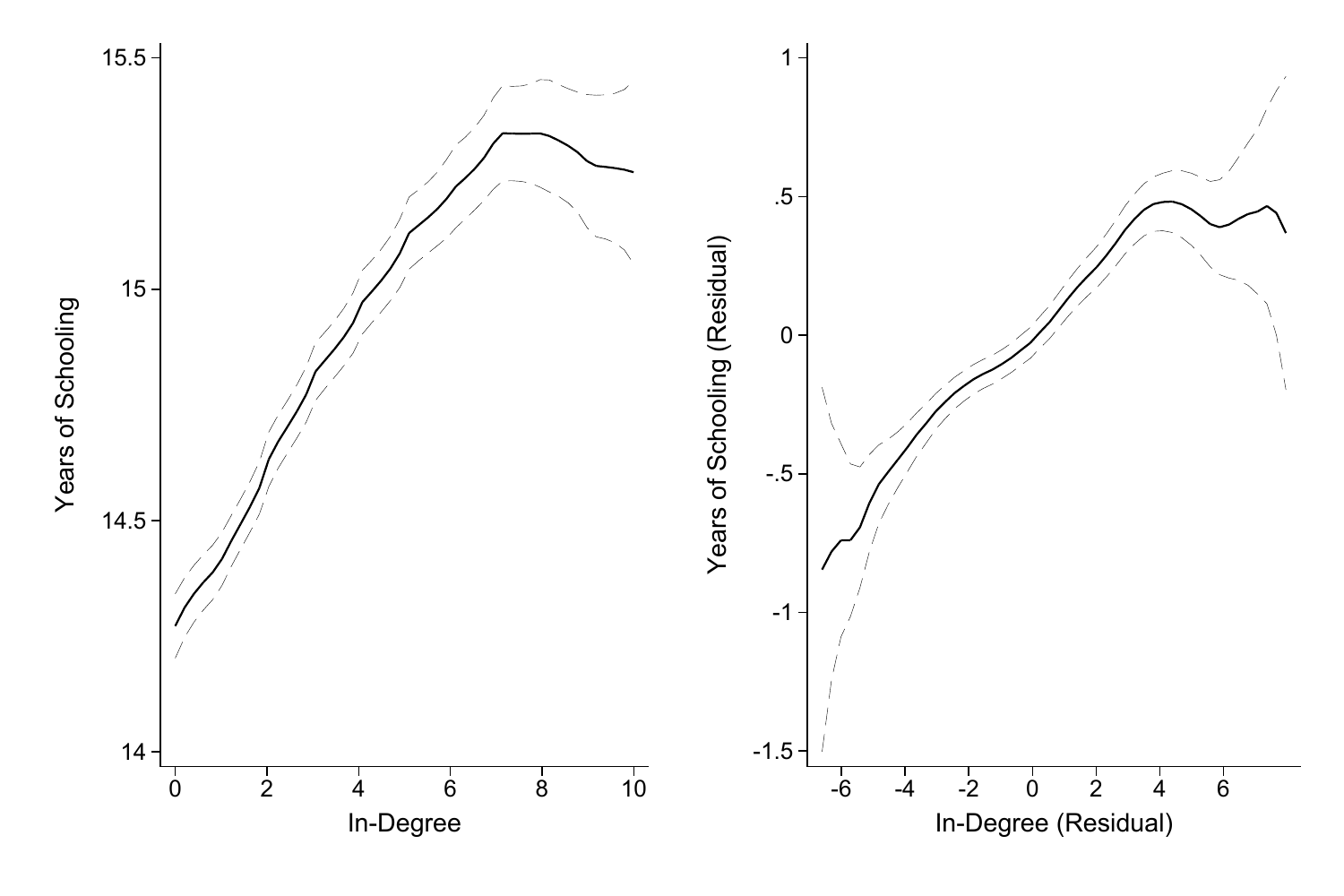}\caption*{Note: Add Health restricted-use data. Local polynomial nonparametric plots and 95\% confidence intervals of years of schooling on grade in-degree. The left plot uses the original variables. The right plot uses their residuals after removing individual characteristics (age, IQ, and indicators for whether the student is extrovert, female and white), cohort-level characteristics (mean age, mean IQ, fraction extrovert, fraction female, and fraction white), grade fixed effects, and school fixed effects.}
\end{figure}

\newpage\begin{landscape}
	\begin{figure}
		\centering\caption{\label{fig:friend_types}Log Earnings and Friendship Measures}
		
		\includegraphics[scale=0.7]{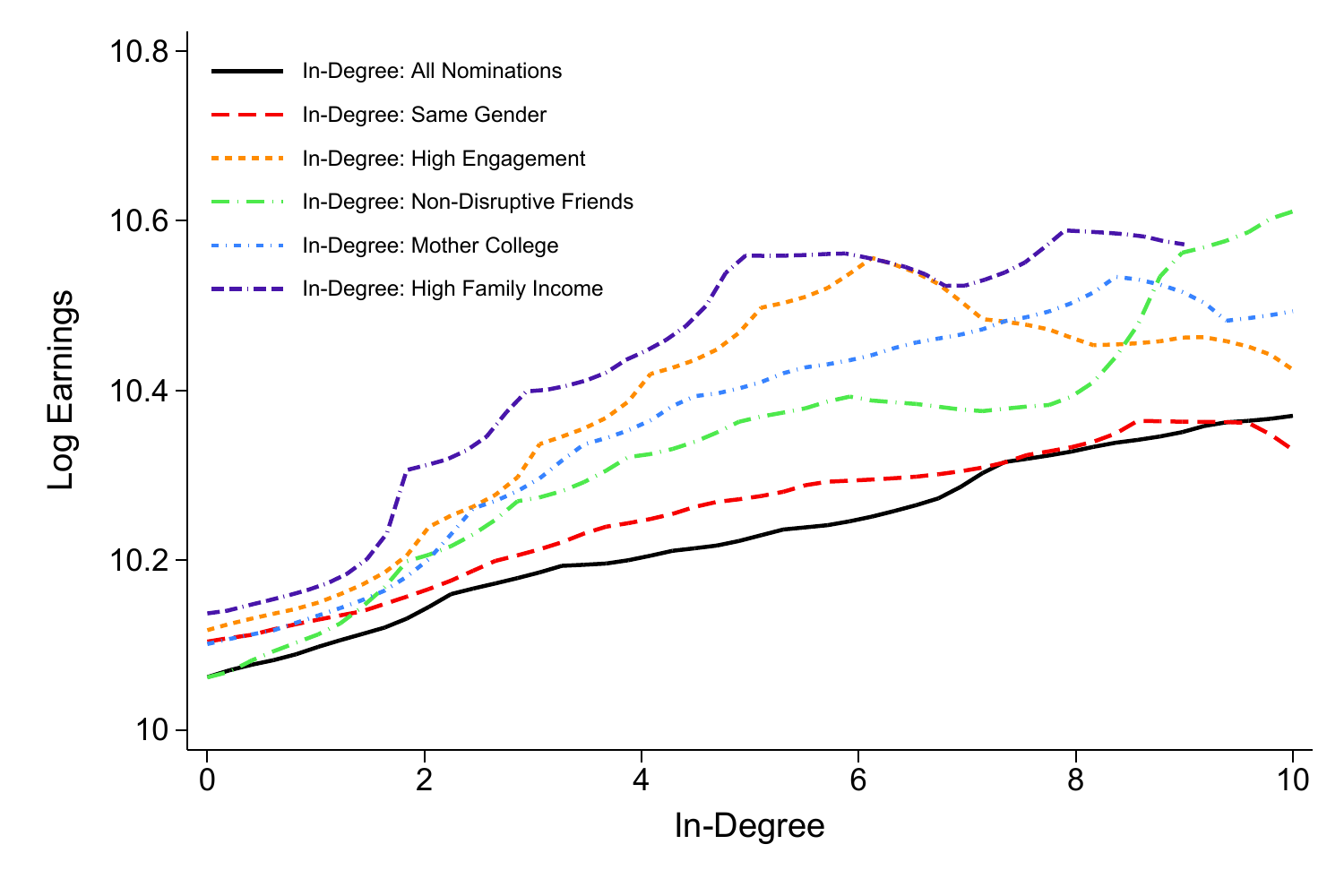}\includegraphics[scale=0.7]{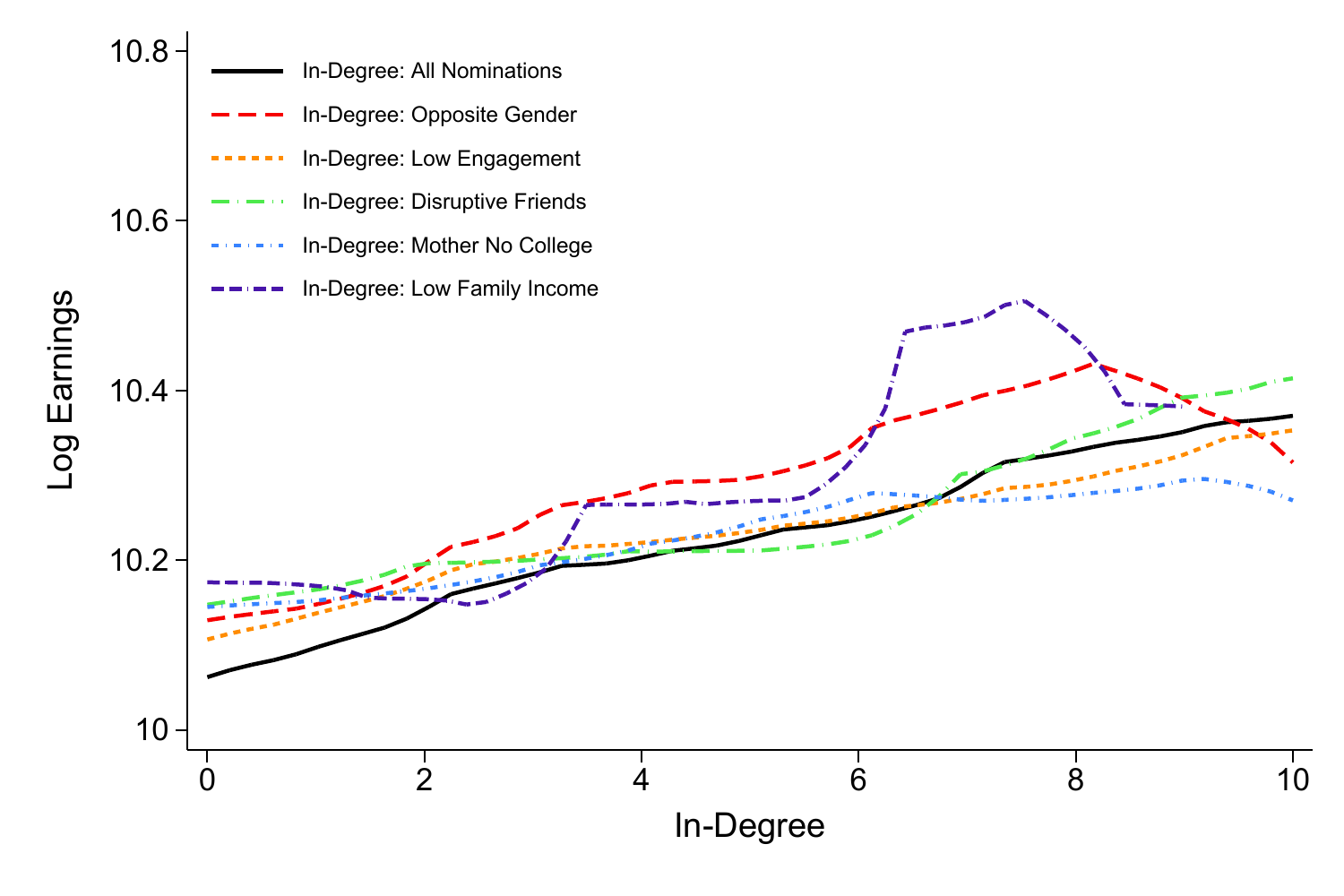}
		
		\caption*{Note: Add Health restricted-use data. Local polynomial nonparametric plots of log earnings on various network measures. In both the left and right figures, the solid line includes all the friendships nominated by students in the same school-grade. The red long-dashed lines restrict to friendships nominated by students in the same (opposite) gender. The orange short-dashed lines restrict to high(low)-engagement nominations (nominations that the nominator reports doing at least three (at most two) out of the five listed activities with the nominated friend). The green long-dash-dotted lines restrict to friendships nominated by (non-)disruptive students.  The blue short-dash-dotted lines restrict to friendships nominated by students with mothers who are (not) college educated.  The purple double-dash-dotted lines restrict to friendships nominated by students from high(low)-income families (above (below) median).}
	\end{figure}
	\newpage{}
	\begin{table}
		\centering \begin{threeparttable}\caption{\label{tab:variation_dist}Variation in Age Distance: An Illustrating
				Example}
			
			\begin{tabular}{cccccc}
				\toprule 
				Cohort & Person & Age & Cohort-Mean Age & Age Distance & Cohort-Mean Age Distance\tabularnewline
				\midrule 
				Cohort 1 & 1 & 14 & 13.5 & 0.75 & 0.67\tabularnewline
				& 2 & 13.5 & 13.5 & 0.5 & 0.67\tabularnewline
				& 3 & 13 & 13.5 & 0.75 & 0.67\tabularnewline
				&  &  &  &  & \tabularnewline
				Cohort 2 & 1 & 13.8 & 13.5 & 0.45 & 0.4\tabularnewline
				& 2 & 13.5 & 13.5 & 0.3 & 0.4\tabularnewline
				& 3 & 13.2 & 13.5 & 0.45 & 0.4\tabularnewline
				\bottomrule
			\end{tabular}
			
			\begin{tablenotes}[flushleft]\footnotesize \item 
				Note: We compute the average distance of a person to all their peers in their cohort excluding themselves. For example, for person 1 in cohort 1, their distance to person 2 is 0.5 and to person 3 is 1, so the average distance is 0.75.
		\end{tablenotes}\end{threeparttable}
	\end{table}
\end{landscape}\newpage{}
\begin{table}
	\centering \begin{threeparttable}\caption{\label{tab:dist_resid}Sample Variation in Age Distance}
		
		{ \def\sym#1{\ifmmode^{#1}\else\(^{#1}\)\fi} \begin{tabular}{l*{1}{c}} \toprule                                 &          SD         \\ \midrule Age Distance                    &       0.435         \\                                 &                     \\ Residual After School FE        &       0.400         \\                                 &                     \\ Residual After School FE + Mean Age&       0.400         \\                                 &                     \\ Residual After School FE + Mean Age + Age&       0.205         \\                                 &                     \\ Residual After School FE + Full Controls&       0.202         \\                                 &                     \\ Residual After School FE + Grade FE&       0.399         \\                                 &                     \\ Residual After School FE + Grade FE + Age&       0.214         \\                                 &                     \\ Residual After School FE + Grade FE + Full Controls&       0.200         \\                                 &                     \\ \midrule Observations                    &      10,605         \\ \bottomrule \end{tabular} }
		
		\begin{tablenotes}[flushleft]\footnotesize \item 
			Note: Add Health restricted-use data. Standard deviations in the residuals of age distance after removing a variety of controls. Mean Age refers to the average age of all the students in a school-grade. Full Controls include individual characteristics (age, IQ, and indicators for whether the student is extrovert, female, and white) and cohort-level characteristics (mean age, mean IQ, fraction extrovert, fraction female, and fraction white).
	\end{tablenotes}\end{threeparttable}
\end{table}
\newpage{}
\begin{table}
	\centering\caption{\label{tab:iv_test}Testing for Instrument Validity: Individual Tests}
	\begin{threeparttable}
		
		{ \def\sym#1{\ifmmode^{#1}\else\(^{#1}\)\fi} \begin{tabular}{l*{1}{c}} \toprule                     &Age Distance         \\ \midrule \textit{The Dependent Variable is}&                     \\ Mother's Years of Schooling&     -0.0642         \\                     &    (0.1437)         \\ Father's Years of Schooling&     -0.0175         \\                     &    (0.1009)         \\ Mother Born in US   &     -0.0119         \\                     &    (0.0260)         \\ Father Born in US   &     -0.0224         \\                     &    (0.0138)         \\ Catholic            &      0.0114         \\                     &    (0.0176)         \\ Baptist             &     -0.0256         \\                     &    (0.0178)         \\ Mother's Age at Respondent's Birth&     -0.1676         \\                     &    (0.2476)         \\ Number of Siblings  &      0.0999         \\                     &    (0.0961)         \\ Born in US          &     -0.0041         \\                     &    (0.0233)         \\ Birth Weight        &      0.0513         \\                     &    (0.0555)         \\ Breastfed           &      0.0221         \\                     &    (0.0198)         \\ Height              &     -0.1614         \\                     &    (0.1616)         \\ Mentally Retarded   &     -0.0000         \\                     &    (0.0019)         \\ Disability          &      0.0005         \\                     &    (0.0029)         \\ \midrule Observations        &      10,605         \\ \bottomrule \end{tabular} }
		
		\begin{tablenotes}[flushleft]\footnotesize \item
			Note: Add Health restricted-use data. Each row presents an OLS regression of the specified dependent variable on age distance, controlling for individual characteristics (age, IQ, and indicators for whether the student is extrovert, female and white), cohort-level characteristics (mean age, mean IQ, fraction extrovert, fraction female, and fraction white), grade fixed effects, and school fixed effects. Standard errors in parentheses are clustered at the school level.
			\item * p < 0.10, ** p < 0.05, *** p < 0.01.
	\end{tablenotes} \end{threeparttable}
\end{table}

\newpage{}
\begin{table}
	\centering \begin{threeparttable}\caption{\label{tab:iv_test_joint}Testing for Instrument Validity: Joint Tests}
		{ \def\sym#1{\ifmmode^{#1}\else\(^{#1}\)\fi} \begin{tabular}{l*{3}{c}} \toprule                     &\multicolumn{1}{c}{(1)}&\multicolumn{1}{c}{(2)}&\multicolumn{1}{c}{(3)}\\                     &\multicolumn{1}{c}{Log Earnings}&\multicolumn{1}{c}{Log Earnings}&\multicolumn{1}{c}{Age Distance}\\ \midrule Mother's Years of Schooling&                     &      0.0143\sym{***}&     -0.0004         \\                     &                     &    (0.0046)         &    (0.0013)         \\ Father's Years of Schooling&                     &      0.0138\sym{**} &      0.0002         \\                     &                     &    (0.0060)         &    (0.0011)         \\ Mother Born in US   &                     &     -0.0518         &      0.0006         \\                     &                     &    (0.0330)         &    (0.0115)         \\ Father Born in US   &                     &     -0.0772\sym{*}  &     -0.0096         \\                     &                     &    (0.0424)         &    (0.0068)         \\ Catholic            &                     &      0.0645\sym{***}&      0.0007         \\                     &                     &    (0.0216)         &    (0.0053)         \\ Baptist             &                     &     -0.0171         &     -0.0072         \\                     &                     &    (0.0280)         &    (0.0059)         \\ Mother's Age at Respondent's Birth&                     &     -0.0000         &     -0.0004         \\                     &                     &    (0.0020)         &    (0.0004)         \\ Number of Siblings  &                     &      0.0128\sym{*}  &      0.0026         \\                     &                     &    (0.0068)         &    (0.0022)         \\ Born in US          &                     &     -0.1577\sym{***}&      0.0033         \\                     &                     &    (0.0410)         &    (0.0148)         \\ Birth Weight        &                     &      0.0193\sym{**} &      0.0017         \\                     &                     &    (0.0079)         &    (0.0015)         \\ Breastfed           &                     &     -0.0292         &      0.0046         \\                     &                     &    (0.0224)         &    (0.0043)         \\ Height              &                     &      0.0030         &     -0.0007         \\                     &                     &    (0.0029)         &    (0.0006)         \\ Mentally Retarded   &                     &     -0.4339         &     -0.0094         \\                     &                     &    (0.3464)         &    (0.0327)         \\ Disability          &                     &     -0.1970         &      0.0019         \\                     &                     &    (0.1370)         &    (0.0192)         \\ \midrule F-stat of Predetermined Characteristics&                     &       6.790         &       0.904         \\ P-value             &                     &       0.000         &       0.556         \\ $ R^{2}$            &       0.059         &       0.066         &       0.753         \\ Observations        &      10,605         &      10,605         &      10,605         \\ \bottomrule \end{tabular} }
		
		\begin{tablenotes}[flushleft]\footnotesize \item
			Note: Add Health restricted-use data. The dependent variable in Columns 1 and 2 is the log annual earnings. The dependent variable in Column 3 is age distance. All specifications include individual characteristics (age, IQ, and indicators for whether the student is extrovert, female and white), cohort-level characteristics (mean age, mean IQ, fraction extrovert, fraction female, and fraction white), grade fixed effects, and school fixed effects. Standard errors in parentheses are clustered at the school level.
			\item * p < 0.10, ** p < 0.05, *** p < 0.01.
	\end{tablenotes} \end{threeparttable}
\end{table}
\newpage{}
\begin{table}
	\centering \begin{threeparttable}
		
		\caption{\label{tab:pair_dist}First-Stage Estimates for Pairwise Friendship
			Nominations}
		{ \def\sym#1{\ifmmode^{#1}\else\(^{#1}\)\fi} \begin{tabular}{l*{3}{c}} \toprule                     &\multicolumn{1}{c}{(1)}         &\multicolumn{1}{c}{(2)}         &\multicolumn{1}{c}{(3)}         \\                     &      Probit         &      Probit         &      Probit         \\ \midrule Pairwise Age Distance&     -0.0699\sym{***}&     -0.0343\sym{***}&                     \\                     &    (0.0031)         &    (0.0051)         &                     \\ Pairwise Age Distance Squared&                     &     -0.0176\sym{***}&                     \\                     &                     &    (0.0022)         &                     \\ Pairwise Age Distance: Older Sender&                     &                     &     -0.1017\sym{***}\\                     &                     &                     &    (0.0046)         \\ Pairwise Age Distance: Younger Sender&                     &                     &     -0.0160\sym{***}\\                     &                     &                     &    (0.0042)         \\ Receiver Age        &     -0.0417\sym{***}&     -0.0424\sym{***}&     -0.0844\sym{***}\\                     &    (0.0035)         &    (0.0036)         &    (0.0059)         \\ Mean Age            &      0.0636         &      0.0656         &      0.0994         \\                     &    (0.0658)         &    (0.0657)         &    (0.0652)         \\ \midrule Mean of Dep. Var.   &       0.011         &       0.011         &       0.011         \\ LR-stat of Homophily Measures&     516.798         &     594.731         &     531.939         \\ P-value             &       0.000         &       0.000         &       0.000         \\ Observations        &    21,918,404         &    21,918,404         &    21,918,404         \\ \bottomrule \end{tabular} }
		
		\begin{tablenotes}[flushleft]\footnotesize \item 
			Note: Add Health restricted-use data. The estimation sample includes all the pairwise combinations of students in the in-home survey within the same school-grade. The dependent variable is an indicator for whether student $i$ nominates student $j$ as a friend. Pairwise age distance is constructed by taking the absolute value of the age distance between student $i$ (sender) and student $j$ (receiver). Pairwise age distance to older (younger) sender is equal to pairwise age distance if the sender is older (younger) than the receiver and  0 otherwise. Additional covariates include receiver's characteristics (IQ and indicators for whether the student is extrovert, female and white), cohort-level characteristics (mean IQ, fraction extrovert, fraction female, and fraction white), grade fixed effects, and school fixed effects. Standard errors in parentheses are clustered at the school level.
			\item * p < 0.10, ** p < 0.05, *** p < 0.01.
	\end{tablenotes}\end{threeparttable}
\end{table}
\newpage{}
\begin{table}
	\centering \begin{threeparttable}\caption{\label{tab:reduced_form}Log Earnings on Age Distance: Reduced-Form
			Results}
		{ \def\sym#1{\ifmmode^{#1}\else\(^{#1}\)\fi} \begin{tabular}{l*{3}{c}} \toprule                     &\multicolumn{1}{c}{Log Earnings}&\multicolumn{1}{c}{Log Earnings}&\multicolumn{1}{c}{Log Earnings}\\ \midrule Age Distance        &     -0.1708\sym{***}&     -0.2908\sym{***}&                     \\                     &    (0.0571)         &    (0.1042)         &                     \\ Age Distance Squared&                     &      0.0359         &                     \\                     &                     &    (0.0315)         &                     \\ Age Distance to Older Peers&                     &                     &     -0.4134\sym{**} \\                     &                     &                     &    (0.1823)         \\ Age Distance to Younger Peers&                     &                     &      0.0373         \\                     &                     &                     &    (0.1131)         \\ Age                 &     -0.0474         &     -0.0358         &     -0.2643\sym{**} \\                     &    (0.0376)         &    (0.0373)         &    (0.1216)         \\ Mean Age            &      0.0644         &      0.0687         &      0.1743         \\                     &    (0.1143)         &    (0.1165)         &    (0.1495)         \\ \midrule F-stat of Homophily Measures&       8.942         &       6.428         &       4.931         \\ P-value             &       0.003         &       0.002         &       0.009         \\ $ R^{2}$            &       0.054         &       0.054         &       0.054         \\ Observations        &      10,605         &      10,605         &      10,605         \\ \bottomrule \end{tabular} }
		
		\begin{tablenotes}[flushleft]\footnotesize \item 
			Note: Add Health restricted-use data. Age distance refers to the average age distance to other students in the same school-grade. Age distance to older (younger) peers refers to the average age distance to other students in the same school-grade who are older (younger) than the respondent. Additional covariates include individual characteristics (IQ and indicators for whether the student is extrovert, female, and white), cohort-level characteristics (mean IQ, fraction extrovert, fraction female, and fraction white), grade fixed effects, and school fixed effects. Standard errors in parentheses are clustered at the school level.
			\item * p < 0.10, ** p < 0.05, *** p < 0.01.
	\end{tablenotes}\end{threeparttable}
\end{table}
\newpage\begin{landscape}
	\begin{table}
		\centering \begin{threeparttable}\caption{\label{tab:cov_edu=000026indeg}Covariance of Education and Friendships}
			\begin{tabular}{lccclcc}
				\toprule 
				\multicolumn{3}{l}{Panel A: Covariance Matrix} & \qquad{} & \multicolumn{3}{l}{Panel B: Inverse of Covariance Matrix}\tabularnewline
				\midrule 
				& Years of Schooling & In-Degree &  &  & Years of Schooling & In-Degree\tabularnewline
				\midrule 
				Years of Schooling & 3.3532 &  &  & Years of Schooling & 0.3055 & \tabularnewline
				In-Degree & 0.8296 & 8.6717 &  & In-Degree & -0.0292 & 0.1181\tabularnewline
				\bottomrule
			\end{tabular}
			
			\begin{tablenotes}[flushleft]\footnotesize \item 
				Note: In-degree refers to the number of friendships nominated by other students in the same school-grade. Both years of schooling and in-degree are the residuals after controlling for individual characteristics (age, IQ, and indicators for whether the student is extrovert, female, and white), cohort-level characteristics (mean age, mean IQ, fraction extrovert, fraction female, and fraction white), grade fixed effects, and school fixed effects.
		\end{tablenotes}\end{threeparttable}
	\end{table}
	\newpage{}
	\begin{table}[t]
		\centering \begin{threeparttable}\caption{\label{tab:second_stage_measures}Labor Market Returns to Friendships:
				Various Friendship Measures}
			\begin{tabular}{lcccc}
				\toprule 
				& (1) & (2) & (3) & (4)\tabularnewline
				& In-Degree & Out-Degree & Reciprocated Degree & Network Size\tabularnewline
				\midrule 
				OLS & 0.0245 & 0.0146 & 0.0392 & 0.0183\tabularnewline
				95\% CI & {[}0.0177, 0.0312{]} & {[}0.0064, 0.0227{]} & {[}0.0261, 0.0522{]} & {[}0.0123, 0.0243{]}\tabularnewline
				\midrule
				IV & {[}0.0926, 0.1367{]} & {[}0.1246, 0.1840{]} & {[}0.2007, 0.2964{]} & {[}0.0722, 0.1067{]}\tabularnewline
				95\% CI & {[}0.0075, 0.2242{]} & {[}0.0114, 0.3022{]} & {[}0.0149, 0.4898{]} & {[}0.0067, 0.1742{]}\tabularnewline
				\midrule
				Observations & 10,605 & 10,605 & 10,605 & 10,605\tabularnewline
				\bottomrule
			\end{tabular}
			
			\begin{tablenotes}[flushleft]\footnotesize \item 
				Note: Add Health restricted-use data. The dependent variable is the log annual earnings. All friendship measures restrict to nominations within the same school-grade. In-degree refers to the number of friendships nominated by other students. Out-degree refers to the number of friendships that the respondent nominates. Reciprocated degree refers to the number of friendships that both nominate. Network size refers to the number of friendships that either one nominates. All specifications include individual characteristics (age, IQ, and indicators for whether the student is extrovert, female, and white), cohort-level characteristics (mean age, mean IQ, fraction extrovert, fraction female, and fraction white), grade fixed effects, and school fixed effects. The instrument is age distance. The confidence intervals are constructed based on standard errors clustered at the school level.
		\end{tablenotes} \end{threeparttable}
	\end{table}
	\newpage{}
	\begin{table}[t]
		\centering \begin{threeparttable}\caption{\label{tab:second_stage_bully}Robustness Check: Bullying}
			
			\begin{tabular}{lccccc}
				\toprule 
				& (1) & (2) & (3) & (4) & (5)\tabularnewline
				& IV & IV & IV & IV & IV\tabularnewline
				\midrule 
				In-Degree & {[}0.0926, 0.1367{]} & {[}0.0913, 0.1347{]} & {[}0.0923, 0.1356{]} & {[}0.0930, 0.1363{]} & {[}0.0923, 0.1351{]}\tabularnewline
				95\% CI & {[}0.0075, 0.2242{]} & {[}0.0051, 0.2235{]} & {[}0.0036, 0.2268{]} & {[}0.0063, 0.2255{]} & {[}0.0029, 0.2270{]}\tabularnewline
				\midrule
				Additional Individual Controls & N & Get Along w Others & Part of School & Safe in School & All 3 Measures\tabularnewline
				Observations & 10,605 & 10,605 & 10,605 & 10,605 & 10,605\tabularnewline
				\bottomrule
			\end{tabular}
			
			\begin{tablenotes}[flushleft]\footnotesize \item 
				Note: Add Health restricted-use data. The dependent variable is the log annual earnings. In-degree refers to the number of friendships nominated by other students in the same school-grade. All specifications include individual characteristics (age, IQ, and indicators for whether the student is extrovert, female, and white), cohort-level characteristics (mean age, mean IQ, fraction extrovert, fraction female, and fraction white), grade fixed effects, and school fixed effects. Column 2 controls for an indicator for whether the student gets along with other students. Column 3 controls for an indicator for whether the student feels like he/she is part of the school.  Column 4 controls for an indicator for whether the student feels safe in the school.  Column 5 controls for all the three indicators in Columns 2-4. The confidence intervals are constructed based on standard errors clustered at the school level.
		\end{tablenotes} \end{threeparttable}
	\end{table}
	\newpage{}
	\begin{table}
		\centering
		
		\caption{\label{tab:subsample}Labor Market Returns to Friendships in Various
			Subsamples: OLS Estimates}
		\begin{threeparttable}{\small{}{ \def\sym#1{\ifmmode^{#1}\else\(^{#1}\)\fi} \begin{tabular}{l*{8}{c}} \toprule                     &\multicolumn{1}{c}{(1)}&\multicolumn{1}{c}{(2)}&\multicolumn{1}{c}{(3)}&\multicolumn{1}{c}{(4)}&\multicolumn{1}{c}{(5)}&\multicolumn{1}{c}{(6)}&\multicolumn{1}{c}{(7)}&\multicolumn{1}{c}{(8)}\\                     &\multicolumn{1}{c}{Log Earnings}&\multicolumn{1}{c}{Log Earnings}&\multicolumn{1}{c}{Log Earnings}&\multicolumn{1}{c}{Log Earnings}&\multicolumn{1}{c}{Log Earnings}&\multicolumn{1}{c}{Log Earnings}&\multicolumn{1}{c}{Log Earnings}&\multicolumn{1}{c}{Log Earnings}\\ \midrule Years of Schooling  &      0.0758\sym{***}&      0.1301\sym{***}&      0.0924\sym{***}&      0.1063\sym{***}&      0.1104\sym{***}&      0.1050\sym{***}&      0.1027\sym{***}&      0.0859\sym{***}\\                     &    (0.0080)         &    (0.0077)         &    (0.0083)         &    (0.0102)         &    (0.0117)         &    (0.0081)         &    (0.0060)         &    (0.0163)         \\ In-Degree           &      0.0240\sym{***}&      0.0239\sym{***}&      0.0221\sym{***}&      0.0255\sym{***}&      0.0191\sym{***}&      0.0263\sym{***}&      0.0237\sym{***}&      0.0292\sym{**} \\                     &    (0.0053)         &    (0.0041)         &    (0.0045)         &    (0.0063)         &    (0.0058)         &    (0.0041)         &    (0.0035)         &    (0.0128)         \\ \midrule Subsample           &        Male         &      Female         & High Family         &  Low Family         &      Mother         &      Mother         &      Native         &   Immigrant         \\                     &                     &                     &      Income         &      Income         &     College         &  No College         &                     &                     \\ $ R^{2}$            &       0.063         &       0.096         &       0.080         &       0.102         &       0.082         &       0.111         &       0.101         &       0.097         \\ Observations        &       5,059         &       5,546         &       4,342         &       3,768         &       3,933         &       5,347         &       9,774         &         831         \\ \bottomrule \end{tabular} }}{\small\par}
			
			\begin{tablenotes}[flushleft]\scriptsize  \item 
				Note: Add Health restricted-use data. In-degree refers to the number of friendships nominated by other students in the same school-grade. Additional covariates include individual characteristics (age, IQ, and indicators for whether the student is extrovert, female and white), cohort-level characteristics (mean age, mean IQ, fraction extrovert, fraction female, and fraction white), grade fixed effects, and school fixed effects. Standard errors in parentheses are clustered at the school level.
				\item * p < 0.10, ** p < 0.05, *** p < 0.01.
		\end{tablenotes}\end{threeparttable}
	\end{table}
	\newpage{}
	\begin{table}
		\centering 
		
		\caption{\label{tab:production_functions}The Effect of Social Activities on
			Friendships, and Cognitive and Health Outcomes}
		\begin{threeparttable}{ \def\sym#1{\ifmmode^{#1}\else\(^{#1}\)\fi} \begin{tabular}{l*{5}{c}} \toprule                     &\multicolumn{1}{c}{(1)}&\multicolumn{1}{c}{(2)}&\multicolumn{1}{c}{(3)}&\multicolumn{1}{c}{(4)}&\multicolumn{1}{c}{(5)}\\                     &\multicolumn{1}{c}{In-Degree}&\multicolumn{1}{c}{GPA}&\multicolumn{1}{c}{Job Cognitive Score}&\multicolumn{1}{c}{Depression}&\multicolumn{1}{c}{Healthy}\\ \midrule Years of Schooling  &                     &                     &      0.0620\sym{***}&     -0.0119\sym{***}&      0.0408\sym{***}\\                     &                     &                     &    (0.0031)         &    (0.0020)         &    (0.0026)         \\ Very Hard Study Effort&      0.3352\sym{***}&      0.3671\sym{***}&      0.0983\sym{***}&     -0.0353\sym{***}&      0.1079\sym{***}\\                     &    (0.1121)         &    (0.0295)         &    (0.0181)         &    (0.0104)         &    (0.0186)         \\ Some Study Effort   &      0.5040\sym{***}&      0.1646\sym{***}&      0.0735\sym{***}&     -0.0255\sym{**} &      0.0632\sym{***}\\                     &    (0.1039)         &    (0.0270)         &    (0.0156)         &    (0.0113)         &    (0.0181)         \\ Frequently Hang with Friends&      0.9035\sym{***}&     -0.0084         &     -0.0330\sym{*}  &     -0.0215\sym{*}  &      0.0727\sym{***}\\                     &    (0.1178)         &    (0.0354)         &    (0.0186)         &    (0.0112)         &    (0.0184)         \\ Sometimes Hang with Friends&      0.6947\sym{***}&      0.0405         &     -0.0169         &     -0.0284\sym{**} &      0.0578\sym{***}\\                     &    (0.1077)         &    (0.0318)         &    (0.0209)         &    (0.0122)         &    (0.0170)         \\ Frequently Drink    &      0.4353\sym{***}&     -0.1996\sym{***}&     -0.0224\sym{*}  &      0.0257\sym{**} &     -0.0259\sym{*}  \\                     &    (0.1199)         &    (0.0281)         &    (0.0134)         &    (0.0104)         &    (0.0151)         \\ Sometimes Drink     &      0.3519\sym{***}&     -0.1050\sym{***}&     -0.0158         &      0.0212\sym{***}&     -0.0243\sym{**} \\                     &    (0.0767)         &    (0.0221)         &    (0.0119)         &    (0.0072)         &    (0.0118)         \\ IQ                  &      0.0148\sym{***}&      0.0122\sym{***}&      0.0017\sym{***}&      0.0014\sym{***}&      0.0003         \\                     &    (0.0024)         &    (0.0010)         &    (0.0005)         &    (0.0003)         &    (0.0004)         \\ Extrovert           &      0.3960\sym{***}&      0.0075         &      0.0283\sym{*}  &     -0.0123         &      0.0251\sym{**} \\                     &    (0.0777)         &    (0.0185)         &    (0.0145)         &    (0.0074)         &    (0.0121)         \\ \midrule Mean of Dep. Var.   &       3.385         &       2.851         &       2.668         &       0.145         &       0.590         \\ $ R^{2}$            &       0.046         &       0.127         &       0.123         &       0.039         &       0.043         \\ Observations        &      10,205         &       7,008         &       5,754         &      10,204         &      10,205         \\ \bottomrule \end{tabular} } 
			
			\begin{tablenotes}[flushleft]\footnotesize \item
				Note: Add Health restricted-use data. The dependent variable in Column 2 is the GPA of the respondent in Wave 2. The dependent variable in Columns 3 is the index of cognitive skills required by the respondent's occupation that is constructed from O*NET data. The dependent variable in Column 4 is an indicator for whether the respondent has ever been diagnosed with depression (Wave 4). The dependent variable in Column 5 is an indicator for whether the respondent's self-reported health status in Wave 4 is excellent or very good. Additional covariates include individual characteristics (age, indicators for whether the student is female and white), cohort-level characteristics (mean age, mean IQ, fraction extrovert, fraction female, and fraction white), grade fixed effects, and school fixed effects. Standard errors in parentheses are clustered at the school level.
				\item * p < 0.10, ** p < 0.05, *** p < 0.01.
		\end{tablenotes}\end{threeparttable}
	\end{table}
	\newpage{}
	\begin{table}
		\centering \begin{threeparttable}\caption{\label{tab:ols_build}Labor Market Returns to Friendships: OLS Estimates}
			
			{ \def\sym#1{\ifmmode^{#1}\else\(^{#1}\)\fi} \begin{tabular}{l*{6}{c}} \toprule                     &\multicolumn{1}{c}{(1)}&\multicolumn{1}{c}{(2)}&\multicolumn{1}{c}{(3)}&\multicolumn{1}{c}{(4)}&\multicolumn{1}{c}{(5)}&\multicolumn{1}{c}{(6)}\\                     &\multicolumn{1}{c}{Log Earnings}&\multicolumn{1}{c}{Log Earnings}&\multicolumn{1}{c}{Log Earnings}&\multicolumn{1}{c}{Log Earnings}&\multicolumn{1}{c}{Log Earnings}&\multicolumn{1}{c}{Log Earnings}\\ \midrule Years of Schooling  &      0.1175\sym{***}&      0.1064\sym{***}&      0.1219\sym{***}&      0.1125\sym{***}&      0.1099\sym{***}&      0.1030\sym{***}\\                     &    (0.0068)         &    (0.0059)         &    (0.0060)         &    (0.0059)         &    (0.0060)         &    (0.0058)         \\ In-Degree           &      0.0153\sym{***}&      0.0145\sym{***}&      0.0211\sym{***}&      0.0219\sym{***}&      0.0229\sym{***}&      0.0245\sym{***}\\                     &    (0.0034)         &    (0.0033)         &    (0.0034)         &    (0.0035)         &    (0.0035)         &    (0.0034)         \\ \midrule Endowments          &           N         &           Y         &           Y         &           Y         &           Y         &           Y         \\ Individual Covariates&           N         &           N         &           Y         &           Y         &           Y         &           Y         \\ Cohort Means        &           N         &           N         &           N         &           Y         &           Y         &           Y         \\ Grade FE            &           N         &           N         &           N         &           N         &           Y         &           Y         \\ School FE           &           N         &           N         &           N         &           N         &           N         &           Y         \\ $ R^{2}$            &       0.065         &       0.070         &       0.117         &       0.125         &       0.127         &       0.098         \\ Observations        &      10,605         &      10,605         &      10,605         &      10,605         &      10,605         &      10,605         \\ \bottomrule \end{tabular} }
			
			\begin{tablenotes}[flushleft]\footnotesize \item
				Note: Add Health restricted-use data. In-degree refers to the number of friendships nominated by other students in the same school-grade. Endowments include IQ and an indicator for whether the student is extrovert. Individual Covariates include age and indicators for whether the student is female and white. Cohort Means include mean age, mean IQ, fraction extrovert, fraction female, and fraction white. Standard errors in parentheses are clustered at the school level.
				\item * p < 0.10, ** p < 0.05, *** p < 0.01.
		\end{tablenotes}\end{threeparttable}
	\end{table}
	\newpage{}
	\begin{table}
		\centering  \begin{threeparttable}\caption{\label{tab:cov_network_measures}Covariance Matrix of Network Measures}
			
			\begin{tabular}{lccc}
				\hline 
				& School In-Degree & Grade In-Degree & Grade In-Degree\tabularnewline
				&  &  & (High Engagement)\tabularnewline
				\hline 
				School In-Degree & 12.250 &  & \tabularnewline
				Grade In-Degree & 9.507 & 8.672 & \tabularnewline
				Grade In-Degree (High Engagement) & 2.499 & 2.281 & 1.363\tabularnewline
				\hline 
			\end{tabular}\begin{tablenotes}[flushleft]\footnotesize \item 
				Note: Add Health restricted-use data. The variance-covariance matrix is calculated using the residuals of the variables after removing individual characteristics (IQ, age, and indicators for whether the student is extrovert, female, and white), cohort-level characteristics (mean IQ, mean age, fraction extrovert, fraction female, and fraction white), grade fixed effects, and school fixed effects.
		\end{tablenotes}\end{threeparttable}
	\end{table}
	\newpage{}
	\begin{table}
		\centering \begin{threeparttable}
			
			\caption{\label{tab:ols_types}Labor Market Returns to Friendships of Various
				Types: OLS Estimates}
			\footnotesize{ \def\sym#1{\ifmmode^{#1}\else\(^{#1}\)\fi} \begin{tabular}{l*{6}{c}} \toprule                     &\multicolumn{1}{c}{(1)}&\multicolumn{1}{c}{(2)}&\multicolumn{1}{c}{(3)}&\multicolumn{1}{c}{(4)}&\multicolumn{1}{c}{(5)}&\multicolumn{1}{c}{(6)}\\                     &\multicolumn{1}{c}{Log Earnings}&\multicolumn{1}{c}{Log Earnings}&\multicolumn{1}{c}{Log Earnings}&\multicolumn{1}{c}{Log Earnings}&\multicolumn{1}{c}{Log Earnings}&\multicolumn{1}{c}{Log Earnings}\\ \midrule Years of Schooling  &      0.1030\sym{***}&      0.1027\sym{***}&      0.1029\sym{***}&      0.1013\sym{***}&      0.1025\sym{***}&      0.1029\sym{***}\\                     &    (0.0058)         &    (0.0058)         &    (0.0058)         &    (0.0059)         &    (0.0060)         &    (0.0058)         \\ In-Degree: All Nominations&      0.0245\sym{***}&                     &                     &                     &                     &                     \\                     &    (0.0034)         &                     &                     &                     &                     &                     \\ In-Degree: Same Gender&                     &      0.0344\sym{***}&                     &                     &                     &                     \\                     &                     &    (0.0048)         &                     &                     &                     &                     \\ In-Degree: Opposite Gender&                     &      0.0149\sym{***}&                     &                     &                     &                     \\                     &                     &    (0.0053)         &                     &                     &                     &                     \\ In-Degree: High Engagement&                     &                     &      0.0401\sym{***}&                     &                     &                     \\                     &                     &                     &    (0.0078)         &                     &                     &                     \\ In-Degree: Low Engagement&                     &                     &      0.0189\sym{***}&                     &                     &                     \\                     &                     &                     &    (0.0047)         &                     &                     &                     \\ In-Degree: Non-Disruptive Friends&                     &                     &                     &      0.0373\sym{***}&                     &                     \\                     &                     &                     &                     &    (0.0049)         &                     &                     \\ In-Degree: Disruptive Friends&                     &                     &                     &      0.0101         &                     &                     \\                     &                     &                     &                     &    (0.0064)         &                     &                     \\ In-Degree: Mother College&                     &                     &                     &                     &      0.0290\sym{***}&                     \\                     &                     &                     &                     &                     &    (0.0072)         &                     \\ In-Degree: Mother No College&                     &                     &                     &                     &      0.0260\sym{***}&                     \\                     &                     &                     &                     &                     &    (0.0045)         &                     \\ In-Degree: Missing Mother Education&                     &                     &                     &                     &     -0.0018         &                     \\                     &                     &                     &                     &                     &    (0.0160)         &                     \\ In-Degree: High Family Income&                     &                     &                     &                     &                     &      0.0474\sym{***}\\                     &                     &                     &                     &                     &                     &    (0.0093)         \\ In-Degree: Low Family Income&                     &                     &                     &                     &                     &      0.0342\sym{*}  \\                     &                     &                     &                     &                     &                     &    (0.0174)         \\ In-Degree: Missing Family Income&                     &                     &                     &                     &                     &      0.0186\sym{***}\\                     &                     &                     &                     &                     &                     &    (0.0044)         \\ \midrule $ R^{2}$            &       0.098         &       0.098         &       0.098         &       0.099         &       0.098         &       0.098         \\ Observations        &      10,605         &      10,605         &      10,605         &      10,605         &      10,605         &      10,605         \\ \bottomrule \end{tabular} }
			
			\begin{tablenotes}[flushleft]\scriptsize  \item 
				Note: Add Health restricted-use data. Column 2  separates friendships nominated by students in the same (opposite) gender. Column 3 separates nominations of high (low) engagement  (nominations that the nominator reports doing three or more out of the five listed activities with the nominated friend). Column 4 separates friendships nominated by (non-)disruptive students. Column 5 separates friendships nominated by students with (without) college-educated mothers.  Column 6 separates friendships nominated by students with family income above (below) median. All specifications include individual characteristics (age, IQ, and indicators for whether the student is extrovert, female, and white), cohort-level characteristics (mean age, mean IQ, fraction extrovert, fraction female, and fraction white), grade fixed effects, and school fixed effects. Standard errors in parentheses are clustered at the school level.
				\item * p < 0.10, ** p < 0.05, *** p < 0.01. 
		\end{tablenotes}\end{threeparttable}
	\end{table}
	\newpage{}
	\begin{table}
		\centering \begin{threeparttable}
			
			\caption{\label{tab:first_stage_types}First-Stage Results for Friendships
				of Various Types}
			{ \def\sym#1{\ifmmode^{#1}\else\(^{#1}\)\fi} \begin{tabular}{l*{5}{c}} \toprule                     &\multicolumn{1}{c}{(1)}         &\multicolumn{1}{c}{(2)}         &\multicolumn{1}{c}{(3)}         &\multicolumn{1}{c}{(4)}         &\multicolumn{1}{c}{(5)}         \\                     &   In-Degree         &   In-Degree         &   In-Degree         &   In-Degree         &   In-Degree         \\ \midrule Age Distance        &     -0.6239\sym{***}&     -0.3580\sym{***}&     -0.6443\sym{***}&     -0.4980\sym{***}&     -0.1941\sym{***}\\                     &    (0.0707)         &    (0.0420)         &    (0.0665)         &    (0.0631)         &    (0.0353)         \\ Age                 &      0.1294\sym{**} &      0.1255\sym{***}&      0.1880\sym{***}&      0.1348\sym{***}&      0.0615\sym{**} \\                     &    (0.0560)         &    (0.0304)         &    (0.0501)         &    (0.0400)         &    (0.0235)         \\ Mean Age            &     -0.0928         &     -0.1111         &     -0.3320\sym{**} &     -0.1198         &     -0.1679\sym{**} \\                     &    (0.1167)         &    (0.0685)         &    (0.1561)         &    (0.1209)         &    (0.0691)         \\ \midrule Who Nominates?      &        Same         &        High         &        Non-         &      Mother         & High Family         \\                     &      Gender         &  Engagement         &  Disruptive         &     College         &      Income         \\ Mean of Dep. Var.   &       2.103         &       0.906         &       1.800         &       1.250         &       0.464         \\ F-stat of Instrument&      77.905         &      72.556         &      93.841         &      62.221         &      30.284         \\ P-value             &       0.000         &       0.000         &       0.000         &       0.000         &       0.000         \\ $ R^{2}$            &       0.043         &       0.030         &       0.039         &       0.036         &       0.014         \\ Observations        &      10,605         &      10,605         &      10,605         &      10,605         &      10,605         \\ \bottomrule \end{tabular} }
			
			\begin{tablenotes}[flushleft]\footnotesize \item 
				Note: Add Health restricted-use data. The dependent variable in Column 1 is the number of friendships nominated by students in the same gender. The dependent variable in Column 2 is the number of high-engagement nominations (nominations that the nominator reports doing three out of the five listed activities with the nominated friend). The dependent variable in Column 3 is the number of friendships nominated by non-disruptive students. The dependent variable in Column 4 is the number of friendships nominated by students with college-educated mothers.  The dependent variable in Column 5 is the number of friendships nominated by students with family income above median. Additional covariates include individual characteristics (IQ and indicators for whether the student is extrovert, female, and white), cohort-level characteristics (mean IQ, fraction extrovert, fraction female, and fraction white), grade fixed effects, and school fixed effects. Standard errors in parentheses are clustered at the school level.
				\item * p < 0.10, ** p < 0.05, *** p < 0.01.
		\end{tablenotes}\end{threeparttable}
	\end{table}
	\newpage{}
	\begin{table}
		\centering\caption{\label{tab:social_mechanism}Why Do Adolescent Friendships Matter
			for Labor Market Outcomes? Social Mechanisms}
		\begin{threeparttable}{ \def\sym#1{\ifmmode^{#1}\else\(^{#1}\)\fi} \begin{tabular}{l*{7}{c}} \toprule                     &\multicolumn{1}{c}{(1)}&\multicolumn{1}{c}{(2)}&\multicolumn{1}{c}{(3)}&\multicolumn{1}{c}{(4)}&\multicolumn{1}{c}{(5)}&\multicolumn{1}{c}{(6)}&\multicolumn{1}{c}{(7)}\\                     &\multicolumn{1}{c}{Work}&\multicolumn{1}{c}{Repetitive}&\multicolumn{1}{c}{Supervisory}&\multicolumn{1}{c}{Job Social Score}&\multicolumn{1}{c}{Friends (W4)}&\multicolumn{1}{c}{Marry}&\multicolumn{1}{c}{Extrovert (W4)}\\ \midrule In-Degree (W1)      &      0.0049\sym{***}&     -0.0091\sym{***}&      0.0025         &      0.0124\sym{***}&      0.0682\sym{***}&      0.0034\sym{**} &      0.0092\sym{***}\\                     &    (0.0013)         &    (0.0014)         &    (0.0015)         &    (0.0025)         &    (0.0093)         &    (0.0017)         &    (0.0015)         \\ Years of Schooling  &      0.0235\sym{***}&     -0.0338\sym{***}&      0.0050\sym{*}  &      0.0662\sym{***}&      0.1841\sym{***}&      0.0046\sym{*}  &      0.0069\sym{***}\\                     &    (0.0023)         &    (0.0024)         &    (0.0030)         &    (0.0041)         &    (0.0173)         &    (0.0026)         &    (0.0025)         \\ IQ                  &     -0.0002         &     -0.0016\sym{***}&     -0.0003         &      0.0006         &      0.0051\sym{**} &     -0.0008\sym{*}  &     -0.0003         \\                     &    (0.0004)         &    (0.0004)         &    (0.0004)         &    (0.0006)         &    (0.0025)         &    (0.0004)         &    (0.0005)         \\ Extrovert (W2)      &     -0.0070         &     -0.0032         &      0.0135         &      0.0427\sym{**} &      0.1255\sym{**} &      0.0309\sym{***}&      0.1694\sym{***}\\                     &    (0.0084)         &    (0.0107)         &    (0.0097)         &    (0.0180)         &    (0.0613)         &    (0.0116)         &    (0.0121)         \\ \midrule Mean of Dep. Var.   &       0.809         &       0.314         &       0.362         &       2.890         &       4.726         &       0.491         &       0.364         \\ $ R^{2}$            &       0.027         &       0.044         &       0.012         &       0.085         &       0.047         &       0.035         &       0.034         \\ Observations        &      11,590         &      11,442         &      11,442         &       5,982         &      10,469         &      10,599         &      10,599         \\ \bottomrule \end{tabular} } 
			
			\begin{tablenotes}[flushleft]\footnotesize \item
				Note: Add Health restricted-use data. W1 stands for Wave I, etc. OLS estimates are reported. The dependent variable in Column 1 is an indicator for whether the respondent is currently working for at least ten hours a week. The dependent variable in Column 2 is an indicator for whether the respondent's job tasks are repetitive. The dependent variable in Column 3 is an indicator for whether the respondent has a supervisory role at their current or previous job. The dependent variable in Columns 4 is the index of social skills required by the respondent's occupation that is constructed from O*NET data. The dependent variable in Column 5 is the number of friends reported in Wave 4. The dependent variable in Column 6 is an indicator for whether the respondent has ever been married (Wave 4). The dependent variable in Column 7 is an indicator for whether the respondent's extroversion index in Wave 4 is 15 and above. In-degree refers to the number of friendships nominated by students in the same school-grade. Additional covariates include individual characteristics (age, and indicators for whether the student is female and white), cohort-level characteristics (mean age, mean IQ, fraction extrovert, fraction female, and fraction white), grade fixed effects, and school fixed effects. Standard errors in parentheses are clustered at the school level.
				\item * p < 0.10, ** p < 0.05, *** p < 0.01.
		\end{tablenotes} \end{threeparttable}
	\end{table}
	\newpage{}
	\begin{table}
		\centering \caption{\label{tab:cognitive_mechanism}Why Do Adolescent Friendships Matter
			for Labor Market Outcomes? Cognitive and Other Mechanisms}
		\begin{threeparttable}{ \def\sym#1{\ifmmode^{#1}\else\(^{#1}\)\fi} \begin{tabular}{l*{6}{c}} \toprule                     &\multicolumn{1}{c}{(1)}&\multicolumn{1}{c}{(2)}&\multicolumn{1}{c}{(3)}&\multicolumn{1}{c}{(4)}&\multicolumn{1}{c}{(5)}&\multicolumn{1}{c}{(6)}\\                     &\multicolumn{1}{c}{GPA}&\multicolumn{1}{c}{Skip School}&\multicolumn{1}{c}{Suspension}&\multicolumn{1}{c}{Job Cognitive Score}&\multicolumn{1}{c}{Depression}&\multicolumn{1}{c}{Healthy}\\ \midrule In-Degree           &      0.0273\sym{***}&     -0.0040\sym{***}&     -0.0105\sym{***}&      0.0091\sym{***}&     -0.0034\sym{***}&      0.0100\sym{***}\\                     &    (0.0028)         &    (0.0014)         &    (0.0017)         &    (0.0017)         &    (0.0010)         &    (0.0016)         \\ Years of Schooling  &                     &                     &                     &      0.0615\sym{***}&     -0.0119\sym{***}&      0.0408\sym{***}\\                     &                     &                     &                     &    (0.0030)         &    (0.0020)         &    (0.0027)         \\ IQ                  &      0.0114\sym{***}&     -0.0013\sym{***}&     -0.0027\sym{***}&      0.0015\sym{***}&      0.0014\sym{***}&      0.0001         \\                     &    (0.0010)         &    (0.0004)         &    (0.0004)         &    (0.0004)         &    (0.0003)         &    (0.0004)         \\ Extrovert           &     -0.0286\sym{*}  &      0.0509\sym{***}&      0.0515\sym{***}&      0.0192         &     -0.0093         &      0.0140         \\                     &    (0.0169)         &    (0.0087)         &    (0.0101)         &    (0.0148)         &    (0.0072)         &    (0.0120)         \\ \midrule Mean of Dep. Var.   &       2.839         &       0.274         &       0.244         &       2.663         &       0.144         &       0.587         \\ $ R^{2}$            &       0.097         &       0.050         &       0.071         &       0.122         &       0.036         &       0.041         \\ Observations        &       7,303         &      10,468         &      10,594         &       5,982         &      10,604         &      10,605         \\ \bottomrule \end{tabular} }
			
			\begin{tablenotes}[flushleft]\footnotesize \item 
				Note: Add Health restricted-use data. OLS estimates are reported. The dependent variable in Column 1 is the GPA of the respondent in Wave 2. The dependent variable in Column 2 is an indicator for whether the respondent has skipped school for at least one full day without an excuse during the school year (Wave 1). The dependent varaible in Column 3 is an indicator for whether the respondent has ever received an out-of-school suspension from school (Wave 1). The dependent variable in Columns 4 is the index of cognitive skills required by the respondent's occupation that is constructed from O*NET data. The dependent variable in Column 5 is an indicator for whether the respondent has ever been diagnosed with depression (Wave 4). The dependent variable in Column 6 is an indicator for whether the respondent's self-reported health status in Wave 4 is excellent or very good. In-degree refers to the number of students in the same school-grade who nominate the respondent as a friend in Wave 1. Additional covariates include individual characteristics (age, and indicators for whether the student is female and white), cohort-level characteristics (mean age, mean IQ, fraction extrovert, fraction female, and fraction white), grade fixed effects, and school fixed effects. Standard errors in parentheses are clustered at the school level.
				\item * p < 0.10, ** p < 0.05, *** p < 0.01.
		\end{tablenotes} \end{threeparttable}
	\end{table}
\end{landscape}

\newpage{}

\section{\label{app:Model}A Model of Friendship Formation, Education and
	Earnings}

We now present a simple model of how individuals decide to invest
in educational and social connections and what ultimately determines
their education, friendships and earnings. In the model, each individual
$i$ decides how to allocate their time (which we normalize to one)
between hours of studying ($H_{i}$), socializing ($S_{i}$), and
leisure ($L_{i}=1-H_{i}-S_{i}$). Studying increases education (denoted
by $E_{i}$), whereas socializing increases the number of friends
(denoted by $F_{i}$). Both $E_{i}$ and $F_{i}$ raise earnings.
Any remaining time is devoted to pure leisure ($L_{i}$), which represents
time spent on solo activities such as watching television or sleeping
and has no labor market returns.

Following \citet{card1999causal}, we assume the individual's utility
is determined by consumption, which is determined by log earnings
($Y$)\footnote{We assume that all earnings are consumed. We abstract from saving
	and borrowing considerations in this paper because we observe earnings
	in the data only once.} and leisure ($U$) (which is split into social and non-social leisure
time), and can be written as
\begin{equation}
	Y(E_{i},F_{i},X_{i},\epsilon_{i})+U(S_{i},L_{i},X_{i},\upsilon_{i},\omega_{i}).\label{eq:utility}
\end{equation}
The log earnings $Y(E_{i},F_{i},X_{i},\epsilon_{i})\equiv Y_{i}(E_{i},F_{i})$
depends on education ($E_{i}$), friendships ($F_{i}$), observed
characteristics ($X_{i}$) such as gender, IQ or sociability, and
unobserved characteristics that affect earnings ($\epsilon_{i}$)
such as competence, motivation, etc. The utility from socializing
and other forms of leisure $U(S_{i},L_{i},X_{i},\upsilon_{i})\equiv U_{i}(S_{i},L_{i})$
depends on the time one spends socializing ($S_{i}$) and doing solo
leisure activities ($L_{i}$), one's type of endowment ($X_{i}$),
and one's unobserved preferences for socializing ($\upsilon_{i}$)
and for leisure ($\omega_{i}$).\footnote{Given that we have two choice variables $S_{i}$ and $L_{i}$, we
	consider errors in both. If we specified the utility function fully,
	it would have two terms (one for the utility of socializing and one
	for the utility of leisure) each of which might contain some unobservable
	determinants.} The cost (or disutility) of studying ($C(H_{i},X_{i},\upsilon_{i},\omega_{i}$))
is netted from $U(S_{i},L_{i},X_{i},\upsilon_{i},\omega_{i})$, so
this utility term can be positive or negative.

Although the model is static, observations and decisions are realized
with a certain timing. At the beginning of the schooling period, each
student $i$ observes their characteristics and the characteristics
of their peers $X=(X_{1},\ldots,X_{n})$, as well as everyone's utility
preferences $\upsilon=(\upsilon_{1},\dots,\upsilon_{n})$ and $\omega=(\omega_{1},\dots,\omega_{n})$.
Then they decide how much time to spend on studying $H_{i}$ and socializing
$S_{i}$. A student's education $E_{i}$ and friendships $F_{i}$
are realized at the end of the schooling period. After the schooling
period, students enter the labor market, $\epsilon_{i}$ is realized,
and they receive log earnings $Y(E_{i},F_{i},X_{i},\epsilon_{i})$.

Given $X$, $\upsilon$ and $\omega$, an individual's expected utility
from the choices $H_{i}$ and $S_{i}$ is given by
\begin{equation}
	\mathbb{E}[Y(E_{i},F_{i},X_{i},\epsilon_{i})|X,\upsilon,\omega]+U(S_{i},1-H_{i}-S{}_{i},X_{i},\upsilon_{i},\omega_{i}).\label{eq:EU}
\end{equation}
We assume that the expected value of log earnings takes the form $\mathbb{E}[Y(E_{i},F_{i},X_{i},\epsilon_{i})|X,\upsilon,\omega]=Y(\mathbb{E}[E_{i}|X,\upsilon,\omega],\mathbb{E}[F_{i}|X,\upsilon,\omega],X_{i},\mathbb{E}[\epsilon_{i}|X,\upsilon,\omega])$,
that is, the expected log earnings is a function of $i$'s expected
education, expected friendships, observable characteristics, and expected
labor market shocks. This is satisfied by the standard specification
in the literature that log earnings is a linear function of education,
friendships and other traits
\begin{equation}
	Y(E_{i},F_{i},X_{i},\epsilon_{i})=r_{e}E_{i}+r_{f}F_{i}+\beta'X_{i}+\epsilon_{i},\label{eq:earning}
\end{equation}
where $\text{\ensuremath{r_{e}}}$ represents the returns to education,
$r_{f}$ represents the returns to friendships, and $\beta$ captures
the effects of the observed characteristics of $i$.\footnote{More general specifications can also be considered. For example, we
	can add an interaction between education and networks, that is, $Y(E_{i},F_{i},X_{i},\epsilon_{i})=r_{e}E_{i}+r_{f}F_{i}+r_{e\cdot f}E_{i}\times F_{i}+\beta'X_{i}+\epsilon_{i}$,
	provided that the production of $E_{i}$ and $F_{i}$ are conditionally
	independent given $X$, $\upsilon$, and $\omega$, which is satisfied
	given the production functions specified below.}

We assume that $X_{i}$ is independent of $\epsilon_{i}$, $\upsilon_{i}$
and $\omega_{i}$ for all $i$, but we allow $\epsilon_{i}$ to be
correlated with $\upsilon_{i}$ and $\omega_{i}$, potentially making
education $E_{i}$ and friendships $F_{i}$ in this equation endogenous
(correlated with the error term). We also assume that conditional
on $(X,\upsilon,\omega)$, $\epsilon_{i}$ does not depend on time
spent studying $H_{i}$ and socializing $S_{i}$. In other words,
the amounts of time spent studying and socializing while growing up
have no direct effect on adult earnings: they affect earnings only
through their effects on education and friendships.\footnote{More explicitly, we are assuming that socializing, a form of leisure,
	increases one's utility but it does not directly affect earnings,
	except through its effect on one's network and education, and similarly
	for studying. The exclusion of $H_{i}$ and $S_{i}$ from $Y_{i}$
	will turn out to be a crucial assumption later on for our IVs to be
	valid.}

\subsection*{Production of education and friendships.}

Education $E_{i}$ and friendships $F_{i}$ depend on initial endowments
and on the time individuals allocate to each activity.

\textbf{Friendship formation.} We assume $i$ becomes a friend of
$j$ if they spend time together (they socialize), and they like each
other (they derive nonnegative utility from the friendship). If individuals
$i$ and $j$ spend $S_{i}$ and $S_{j}$ amounts of time socializing,
then $i$ and $j$ become friends ($F_{ij}=1$) following
\begin{equation}
	F_{ij}=\{g(S_{i},S_{j},X_{i},X_{j})+\eta_{ij}\geq0\},\label{eq:link}
\end{equation}
where $g(S_{i},S_{j},X_{i},X_{j})$ represents the deterministic latent
utility of becoming friends and $\eta_{ij}$ represents an unobservable
preference shock in friendship formation that is independent of $\ensuremath{X}$,
$\upsilon$, and $\omega$. For example, individuals $i$ and $j$
might become friends because they were together during a particularly
good or bad event. The likelihood that $i$ and $j$ become friends,
$\Pr(F_{ij}=1|S_{i},S_{j},X_{i},X_{j})\equiv p_{ij}(S_{i},S_{j})$,
depends on how much time they spend socializing $S_{i}$ and $S_{j}$.\footnote{This model is similar in spirit to the dyadic network formation models
	with individual fixed effects \citep{Graham2017}, where the time
	spent socializing $S_{i}$ and $S_{j}$ act as the individual fixed
	effects.} We assume that both\textit{ $i$} and $j$ have to spend time socializing
together ($S_{i}>0$ and $S_{j}>0$) for them to have a non-zero probability
of becoming friends. Thus the production of friendships requires coordination
with others -- you cannot ``party alone.'' We assume that if people
study together they are socializing part of the time, and studying
part of the time.\footnote{We assume that even when people study together that time can be split
	into socializing time and studying time. This is of course a simplification.
	The results are qualitatively similar if we write a more realistic
	model where there are four activities possible including ``studying
	together''. However, this model is more complicated and less intuitive
	so we present the simpler version here.}

Conditional on socializing $S_{i}$ and $S_{j}$, the likelihood of
becoming friends also depends on the individual and shared characteristics
of \textit{$i$} and\textit{ $j$} ($X_{i}$ and $X_{j}$). This feature
captures the empirical finding that individuals tend to form friendships
with other individuals with whom they share similar characteristics,
often referred to as homophily. Individual $i$'s total number of
friends $F_{i}$ is given by $F_{i}=\sum_{j\neq i}F_{ij},$where $F_{ij}$
is the indicator for whether $i$ and $j$ are friends.

\textbf{Education.} The production of education is standard. If individual
$i$ spends a certain amount of time $H_{i}$ studying, their education
is given by
\begin{equation}
	E_{i}=a(H_{i},X_{i},X_{-i})+\xi_{i},\label{eq:edu}
\end{equation}
where $a(H_{i},X_{i},X_{-i})\equiv a_{i}(H_{i})$ represents the deterministic
educational output, and $\xi_{i}$ is an unobservable shock that is
assumed to be independent of $X$, $\upsilon$, and $\omega$. For
example, individuals could suffer unexpected health shocks (such as
getting the flu) that affect their ability to attend school or study.
Thus in addition to studying, the production of education depends
on the observed characteristics of $i$ ($X_{i}$) (which include
their cognitive skills, social skills, and other characteristics,
most importantly the school they attend), and on the characteristics
of their peers $X_{-i}=(X_{j},j\neq i)$. This specification allows
for (exogenous) peer effects in education, namely, the possibility
that individuals learn faster (or slower) depending on the characteristics
of their classroom peers. Note however that even in the absence of
peers, individuals can obtain education.

\subsection*{How do individuals decide how much to study and socialize? Best response
	functions and equilibrium.}

Given information $(X,\upsilon,\omega)$, individual $i$ chooses
$H_{i}$ and $S_{i}$ in order to maximize their \textit{expected}
utility in (\ref{eq:EU}). Recall that the expected value of log earnings
is a function of $i$'s expected education, expected number of friends,
observable characteristics, and expected labor market shocks. For
individual $i,$ the expected number of friends depends on the time
she spends socializing and the amount of time others spend socializing.
In particular, the expected number of friends is of the form $\mathbb{E}[F_{i}|X,\upsilon,\omega]=\sum_{j\neq i}p_{ij}(S_{i},S_{j})$.
The expected education depends on the time $i$ spent studying and
takes the form $\mathbb{E}[E_{i}|X,\upsilon,\omega]=a_{i}(H_{i})$.
Let $\partial Y_{i}/\partial E_{i}$ and $\partial Y_{i}/\partial F_{i}$
represent the derivatives of $Y_{i}(E_{i},F_{i})$ with respective
to $E_{i}$ and $F_{i}$ evaluated at $E_{i}=a_{i}(H_{i})$ and $F_{i}=\sum_{j\neq i}p_{ij}(S_{i},S_{j})$.

Given the actions $S_{j}$ of other individuals $j\neq i$, the optimal
amounts of studying $H_{i}$ and socializing $S_{i}$ satisfy the
first-order conditions
\begin{eqnarray}
	\frac{\partial U_{i}}{\partial L_{i}}(S_{i},1-H_{i}-S_{i}) & = & \frac{\partial Y_{i}}{\partial E_{i}}\frac{\partial a_{i}}{\partial H_{i}}(H_{i})\label{eq:foc.study}\\
	\frac{\partial U_{i}}{\partial L_{i}}(S_{i},1-H_{i}-S_{i}) & = & \frac{\partial U_{i}}{\partial S_{i}}(S_{i},1-H_{i}-S_{i})+\frac{\partial Y_{i}}{\partial F_{i}}\sum_{j\neq i}\frac{\partial p_{ij}}{\partial S_{i}}(S_{i},S_{j}).\label{eq:foc.social}
\end{eqnarray}
The optimal amount of studying equates the marginal utility with the
marginal cost of studying. The marginal cost of studying is the utility
loss due to one less unit of leisure (the left-hand side of (\ref{eq:foc.study})).
The marginal benefit from studying (the right-hand side of (\ref{eq:foc.study}))
is the additional earnings from studying one more unit of time. Similarly,
the optimal amount of socializing equates the marginal utility with
the marginal cost of socializing. The marginal cost of socializing
is also the utility loss due to one less unit of leisure. The marginal
benefit of socializing (the right-hand side of (\ref{eq:foc.social}))
has two terms. The first term is the direct utility individuals derive
from socializing one more unit of time. The second term comes from
the additional earnings individuals get when they socialize and accumulate
friendships.

Assume that the marginal utilities from socializing and leisure are
both diminishing.\footnote{Specifically we assume that the marginal utility from socializing
	$\partial U_{i}/\partial S_{i}$ is decreasing in $S_{i}$ ($\partial^{2}U_{i}/\partial S_{i}^{2}<0$),
	and similarly for the marginal utility from leisure $\partial U_{i}/\partial L_{i}$
	($\partial^{2}U_{i}/\partial L_{i}^{2}<0$).} Monotonicity implies that their inverse functions exist. Let $I_{i}^{S}(\cdot,L_{i})$
and $I_{i}^{L}(\cdot,S_{i})$ denote the inverse of the marginal utility
functions $\partial U_{i}(\cdot,L_{i})/\partial S_{i}$ and $\partial U_{i}(S_{i},\cdot)/\partial L_{i}$
with respect to $S_{i}$ and $L_{i}$ respectively. Then, under the
standard Inada conditions,\footnote{We assume that $\lim_{S_{i}\rightarrow0}\frac{\partial U_{i}}{\partial S_{i}}(S_{i},L_{i})=\infty$,
	$\lim_{L_{i}\rightarrow0}\frac{\partial U_{i}}{\partial L_{i}}(S_{i},L_{i})=\infty$,
	and $\lim_{H_{i}\rightarrow0}\frac{\partial a_{i}}{\partial H_{i}}(H_{i})=\infty$.} we can derive the optimal $H_{i}$ and $S_{i}$ as an interior solution
to
\begin{eqnarray}
	S_{i} & = & I_{i}^{S}\left(\frac{\partial Y_{i}}{\partial E_{i}}\frac{\partial a_{i}}{\partial H_{i}}(H_{i})-\frac{\partial Y_{i}}{\partial F_{i}}\sum_{j\neq i}\frac{\partial p_{ij}}{\partial S_{i}}(S_{i},S_{j}),1-H_{i}-S_{i}\right)\label{eq:opt.social}\\
	H_{i} & = & 1-S_{i}-I_{i}^{L}\left(\frac{\partial Y_{i}}{\partial E_{i}}\frac{\partial a_{i}}{\partial H_{i}}(H_{i}),S_{i}\right).\label{eq:opt.study}
\end{eqnarray}
Because the right-hand sides of (\ref{eq:opt.social}) and (\ref{eq:opt.study})
are both continuous functions of $H_{i}$ and $S_{i}$, by Brouwer's
fixed-point theorem, there exists at least one solution to equations
(\ref{eq:opt.social})-(\ref{eq:opt.study}). In general, there may
be multiple solutions.\footnote{In the most general form of the model -- without imposing the Inada
	conditions, there can be multiple equilibria in the model. Intuitively,
	there can be situations where no one ever parties because no one else
	is partying, or where everyone parties all the time, because everyone
	else is partying. We can also obtain multiple interior solutions where
	individuals both study and socialize in partial amounts.}

The optimal amount of socializing depends on the decisions of others
to socialize; thus, equations (\ref{eq:opt.social})-(\ref{eq:opt.study})
correspond to the ``best response function'' of a given individual,
who takes others' actions as given. A Nash equilibrium is a profile
of actions $(H_{i}^{*},S_{i}^{*})$, $i=1,\dots,n$, that satisfies
(\ref{eq:opt.social})-(\ref{eq:opt.study}) for all $i=1,\dots,n$.
The best response of socializing in a Nash equilibrium is given by
\begin{equation}
	S_{i}^{*}=I_{i}^{S}\left(\frac{\partial Y_{i}^{*}}{\partial E_{i}}\frac{\partial a_{i}}{\partial H_{i}}(H_{i}^{*})-\frac{\partial Y_{i}^{*}}{\partial F_{i}}\sum_{j\neq i}\frac{\partial p_{ij}}{\partial S_{i}}(S_{i}^{*},S_{j}^{*}),1-H_{i}^{*}-S_{i}^{*}\right),\label{eq:brf}
\end{equation}
where $\partial Y_{i}^{*}/\partial E_{i}$ and $\partial Y_{i}^{*}/\partial F_{i}$
represent the derivatives of $Y_{i}(E_{i},F_{i})$ with respect to
$E_{i}$ and $F_{i}$ evaluated at $E_{i}=a_{i}(H_{i}^{*})$ and $F_{i}=\sum_{j\neq i}p_{ij}(S_{i}^{*},S_{j}^{*})$.
This equation states that the optimal amounts of socializing of students
in a school are jointly determined by the returns to schooling, the
returns to friendships, schooling inputs and productivity, homophily
with peers, and their initial endowments and preferences. Due to the
coordination effects, the social investment of a student is no longer
unilaterally determined by that student: it also depends on the investments
and endowments of the peers.

\subsection{Implications.}

We now investigate a few properties of the model that are useful for
the empirical analysis.
\begin{prop}
	\label{prop:no-social-return}Assume that the marginal product of
	studying is positive ($\partial\ensuremath{a_{i}/\partial H_{i}}>0$)
	and the labor market returns to education are positive ($\partial Y_{i}/\partial E_{i}>0$).
	If there are no labor market returns to friendships ($\partial Y_{i}/\partial F_{i}=0$),
	then in an OLS regression of log earnings on socializing, the coefficient
	on socializing will be negative when education and friendships are
	not controlled for.
\end{prop}
\begin{proof}
	See Appendix \ref{app:Proof.no-social-return}. If there are no returns
	to friendships ($\partial Y_{i}/\partial F_{i}=0$), then the optimal
	amount of socializing becomes
	\begin{equation}
		S_{i}^{*}=I_{i}^{S}\left(\frac{\partial Y_{i}^{*}}{\partial E_{i}}\frac{\partial a_{i}}{\partial H_{i}}(H_{i}^{*}),1-H_{i}^{*}-S_{i}^{*}\right).\label{eq:social.brf.nr}
	\end{equation}
	In this case, individuals still socialize ($S_{i}^{*}>0$), but socializing
	would have non-positive returns in the labor market because it is
	pure leisure.
\end{proof}
\begin{prop}
	\label{prop:ols}If we estimate the earnings equation by OLS, then
	the estimated returns to education and friendships will be biased.
	Without further assumptions, the directions of the biases in the OLS
	estimates of returns to education and friendships are ambiguous.
\end{prop}
\begin{proof}
	See Appendix \ref{app:Proof.ols}. In the model, both education and
	friendships can be correlated with the error term because tastes for
	leisure and for socializing are unobserved. If we estimate the earnings
	equation by OLS, then the estimated returns to education and friendships
	will be biased. This proposition states that ex-ante, it is not possible
	to assign the direction of the bias. The bias can be upward or downward,
	depending on, for example, whether the taste for socializing leans
	more toward productive activities like studying or working together
	or pure leisure like drinking or partying.
\end{proof}
\begin{prop}
	\label{prop:homophily}Suppose that Assumption \ref{ass:S_d} is satisfied.
	Increasing an individual's homophily level has ambiguous predictions
	on socializing and on the number of friends. But the level of homophily
	will in general affect the equilibrium level of socializing and the
	number of friends an individual has.
\end{prop}
\begin{proof}
	See Appendix \ref{app:Proof.homophily}. Without coordination effects,
	the model predicts that individuals that are more similar to each
	other are more likely to be friends. However, the effect of increasing
	homophily on the number of friends is theoretically ambiguous because
	of coordination effects. For example, suppose that a boy is placed
	in a cohort with only girls. Compare him to another boy placed in
	a cohort with only boys. If girls socialize more than boys, then the
	boy placed in the girls-only cohort may still want to socialize more
	because it is more productive to socialize in this group and he might
	make more friends as a result, despite the gender difference. Therefore,
	the effect of homophily on social outcomes (socializing and number
	of friends) is an empirical question. Nevertheless, the theory also
	predicts that unless the direct impact of social distance is exactly
	offset by the coordination effects, the effect of homophily on social
	outcomes will not be zero.
\end{proof}
\begin{prop}
	\label{prop:endow}Suppose that Assumptions \ref{ass:S_x}-\ref{ass:S_x.extro}
	are satisfied. The effect of an individual's cognitive and social
	endowments ($X_{i}$) on socializing and thus on friendships is theoretically
	ambiguous. The same is true for studying and education.
\end{prop}
\begin{proof}
	See Appendix \ref{app:Proof.endow}. One might expect that individuals
	reinforce their initial endowments -- social people socialize and
	make friends, smart people study -- but that is not necessarily the
	case. If we assume that individuals with large social skills are not
	better students (social endowments do not affect studying productivity),
	then individuals with social skills will spend more time socializing
	and less time studying. Consequently, they accumulate more friends
	and less education. But if social endowments affect studying productivity
	-- for example if people study together, then the predictions will
	not hold. Similarly, whether or not smart students socialize more
	depends on their socializing and studying productivity. If the marginal
	productivity of socializing is increasing in intelligence and large
	enough to induce a reduction in studying time, then high-IQ students
	can accumulate more friends.
\end{proof}

\subsection{\label{app:Proofs}Proofs}

\subsubsection{\label{app:Proof.no-social-return}Proposition \ref{prop:no-social-return}}
\begin{proof}[Proof of Proposition \ref{prop:no-social-return}]
	Recall that the expected log earnings is 
	\begin{eqnarray*}
		\mathbb{E}[Y_{i}|X,\upsilon,\omega] & = & Y(a_{i}(1-S_{i}-L_{i}),\sum_{j}p_{ij}(S_{i},S_{j}),X_{i},\mathbb{E}[\epsilon_{i}|X,\upsilon,\omega]).
	\end{eqnarray*}
	Suppose there are no social returns ($\partial\ensuremath{Y}_{i}/\partial\ensuremath{F}_{i}=0$).
	Under the assumption that the education returns are positive ($\partial Y_{i}/\partial E_{i}>0$)
	and the marginal product of studying is positive ($\partial\ensuremath{a_{i}(H_{i})}/\partial H_{i}>0$),
	the derivative of the expected log earnings with respect to $S_{i}$
	is negative
	\[
	\frac{\partial\mathbb{E}[Y_{i}|X,\upsilon,\omega]}{\partial S_{i}}=-\frac{\partial\ensuremath{Y}_{i}}{\partial E_{i}}\frac{\partial\ensuremath{a_{i}}}{\partial H_{i}}(1-S_{i}-L_{i})<0.
	\]
	In other words, when there are no returns to socializing then a student
	who socializes relatively more will study less, attain a lower level
	of education, and have lower expected log earnings.
\end{proof}

\subsubsection{\label{app:Proof.ols}Proposition \ref{prop:ols}}
\begin{proof}[Proof of Proposition \ref{prop:ols}]
	Write the OLS regression of log earnings on education $E_{i}$, friendships
	$F_{i}$, and other traits $X_{i}$ as 
	\begin{equation}
		Y_{i}=r'K_{i}+\beta'X_{i}+\epsilon_{i},\label{eq:app.earning}
	\end{equation}
	where $K_{i}=(E_{i},F_{i})'$ represents the vector of education and
	friendships and $r=(r_{e},r_{f})'$ represents the vector of returns
	to education and friendships. Assume that $\mathbb{E}(\epsilon_{i})=0$.
	The endogeneity of $K_{i}$ and exogeneity of $X_{i}$ mean that $\mathbb{E}[K_{i}\epsilon_{i}]\neq0$
	and $\mathbb{E}[X_{i}\epsilon_{i}]=0$.
	
	The OLS estimator of $r$ is given by
	\[
	\text{\ensuremath{\text{\ensuremath{\hat{r}=(\sum_{i=1}^{n}\tilde{K}_{i}K_{i}')^{-1}\sum_{i=1}^{n}\tilde{K}_{i}Y_{i}}}}},
	\]
	where $\tilde{K}_{i}=K_{i}-(\sum_{i=1}^{n}K_{i}X_{i}')(\sum_{i=1}^{n}X_{i}X_{i}')^{-1}X_{i}$
	is the residual of $K_{i}$ after partialing out $X_{i}$. Because
	$\sum_{i=1}^{n}\tilde{K}_{i}X_{i}'=0$ we can derive $\text{\ensuremath{\text{\ensuremath{\hat{r}-r}}}}=(\sum_{i=1}^{n}\tilde{K}_{i}K_{i}')^{-1}\sum_{i=1}^{n}\tilde{K}_{i}\epsilon_{i}$.
	By the law of large numbers and the exogeneity of $X_{i}$ ($\mathbb{E}[X_{i}\epsilon_{i}]=0$),
	we obtain
	\[
	\text{\ensuremath{\text{\ensuremath{\hat{r}-r}}}}\overset{p}{\rightarrow}(\mathbb{E}[K_{i}K_{i}']-\mathbb{E}[K_{i}X_{i}']\mathbb{E}[X_{i}X_{i}']^{-1}\mathbb{E}[X_{i}K_{i}'])^{-1}\mathbb{E}[K_{i}\epsilon_{i}]
	\]
	as $n\rightarrow\infty$. Because $\mathbb{E}[K_{i}\epsilon_{i}]\neq0$
	due to the endogeneity of $K_{i}$, the right-hand side is nonzero,
	leading to a bias in the OLS estimator $\hat{r}$.
	
	The bias in $\ensuremath{\hat{r}}$ depends on both the correlation
	between $K_{i}$ and $\epsilon_{i}$ ($\mathbb{E}[K_{i}\epsilon_{i}]$)
	and the inverse of the matrix
	\begin{equation}
		\mathbb{E}[K_{i}K_{i}']-\mathbb{E}[K_{i}X_{i}']\mathbb{E}[X_{i}X_{i}']^{-1}\mathbb{E}[X_{i}K_{i}'].\label{eq:app.Exx}
	\end{equation}
	Notice that because this matrix is positive definite,\footnote{This is because $\mathbb{E}[K_{i}K_{i}']-\mathbb{E}[K_{i}X_{i}']\mathbb{E}[X_{i}X_{i}']^{-1}\mathbb{E}[X_{i}K_{i}']=\mathbb{E}[(K_{i}-\mathbb{E}[K_{i}X_{i}']\mathbb{E}[X_{i}X_{i}']^{-1}X_{i})(K_{i}-\mathbb{E}[K_{i}X_{i}']\mathbb{E}[X_{i}X_{i}']^{-1}X_{i})']$.}
	so is its inverse. Unlike the case with a single endogenous variable,
	unless matrix (\ref{eq:app.Exx}) is diagonal, the OLS bias in $r_{f}$
	does not necessarily have the same sign as the correlation between
	friendships and the error term $\epsilon_{i}$, and the same for $r_{e}$.
	Without knowledge about matrix (\ref{eq:app.Exx}), even if the correlations
	between education, friendships and $\epsilon_{i}$ are known, the
	OLS biases in the returns to education and friendships are ambiguous.
\end{proof}

\paragraph{Discussion: OLS bias in our sample}

Observe that matrix (\ref{eq:app.Exx}) involves the observables $K_{i}$
and $X_{i}$ only. In a given dataset this matrix can be estimated.
In our Add Health data, we find that the off-diagonal elements in
the inverse of matrix (\ref{eq:app.Exx}) are relatively small compared
with the diagonal elements (Table \ref{tab:cov_edu=000026indeg}).\footnote{The off-diagonal elements in the inverse matrix are determined by
	the correlation between education and friendships. If conditional
	on other covariates, education and friendships are positively (negatively)
	correlated, then we expect the off-diagonal elements in the inverse
	of the matrix (\ref{eq:app.Exx}) to be negative (positive). In our
	data, we find that education and friendships are positively correlated,
	controlling for other covariates, and the off-diagonal elements in
	the inverse of the matrix (\ref{eq:app.Exx}) are indeed negative.} This suggests that the OLS biases in the returns to education and
friendships are mostly determined by their correlations with $\epsilon_{i}$.
In particular, a positive (negative) correlation between friendships
and \textbf{$\epsilon_{i}$ }will lead to an upward (downward) bias
in the returns to friendships.

The number of friends can be correlated with $\epsilon_{i}$ for a
number of reasons. For example, because the number of friends depends
on socializing, which in turn depends on the unobserved preference
for socializing $\upsilon_{i}$, if $\upsilon_{i}$ is correlated
with \textbf{$\epsilon_{i}$}, so is the number of friends.

To speculate on the direction of the correlation between friendships
and \textbf{$\epsilon_{i}$ }through $\upsilon_{i}$, we show in a
lemma that socializing is increasing in the unobserved preference
for socializing.
\begin{assumption}
	\label{ass:S_v}(i) The marginal utility from socializing is increasing
	in \textup{$\upsilon_{i}$} ($\partial^{2}U_{i}/\partial S_{i}\partial\upsilon_{i}>0$),
	that is, students with higher \textup{$\upsilon_{i}$} enjoy more
	utility from each unit of socializing. (ii) The marginal utility from
	leisure does not depend on socializing nor $\upsilon_{i}$ ($\partial^{2}U_{i}/\partial L_{i}\partial S_{i}=0$
	and $\partial^{2}U_{i}/\partial L_{i}\partial\upsilon_{i}=0$). (iii)
	The marginal product from socializing (or studying) is diminishing
	in socializing (or studying) (\textup{$\partial^{2}p_{ij}/\partial S_{i}^{2}<0$}
	and \textup{$\partial^{2}a_{i}/\partial H_{i}^{2}<0$}). (iv) The
	labor market returns to education and friendships are positive constants
	(\textup{$\partial Y_{i}/\partial E_{i}=r_{e}>0$} and $\partial Y_{i}/\partial F_{i}=r_{f}>0$)\textup{.}
\end{assumption}
\begin{lem}
	Under Assumption \ref{ass:S_v}, the optimal amount of socializing
	$S_{i}^{*}$ is increasing in \textup{$\upsilon_{i}$}.
\end{lem}
\begin{proof}
	Suppose that the change of $\upsilon_{i}$ does not trigger any change
	in the equilibrium so that $\frac{\partial S_{j}^{*}}{\partial\upsilon_{i}}=0$,
	for all $j\neq i$. Taking the derivative of both sides of (\ref{eq:foc.study})
	and (\ref{eq:foc.social}) evaluated at optimal $(H_{i}^{*},S_{i}^{*})$,
	$i=1,\dots,n$, with respect to $\upsilon_{i}$, we obtain
	\begin{eqnarray}
		\frac{\partial^{2}U_{i}}{\partial L_{i}^{2}}(S_{i}^{*},L_{i}^{*})\frac{\partial L_{i}^{*}}{\partial\upsilon_{i}} & = & r_{e}\frac{\partial^{2}a_{i}}{\partial H_{i}^{2}}(H_{i}^{*})\frac{\partial H_{i}^{*}}{\partial\upsilon_{i}}\label{eq:foc.study.dv}\\
		\frac{\partial^{2}U_{i}}{\partial L_{i}^{2}}(S_{i}^{*},L_{i}^{*})\frac{\partial L_{i}^{*}}{\partial\upsilon_{i}} & = & \frac{\partial^{2}U_{i}}{\partial S_{i}\partial\upsilon_{i}}(S_{i}^{*},L_{i}^{*})+\left(\frac{\partial^{2}U_{i}}{\partial S_{i}^{2}}(S_{i}^{*},L_{i}^{*})+r_{f}\sum_{j\neq i}\frac{\partial^{2}p_{ij}}{\partial S_{i}^{2}}(S_{i}^{*},S_{j}^{*})\right)\frac{\partial S_{i}^{*}}{\partial\upsilon_{i}},\label{eq:foc.social.dv}
	\end{eqnarray}
	where we have used $\partial^{2}U_{i}/\partial L_{i}\partial S_{i}=0$,
	$\partial^{2}U_{i}/\partial L_{i}\partial\upsilon_{i}=0$, \textit{$\partial Y_{i}/\partial E_{i}=r_{e}$},
	and $\partial Y_{i}/\partial F_{i}=r_{f}$ by Assumption \ref{ass:S_v}(ii)
	and (iv), and the fact that $a_{i}$ does not depend on $\upsilon_{i}$.
	Recall that $\partial^{2}U_{i}/\partial L_{i}^{2}<0$ and $\partial^{2}U_{i}/\partial S_{i}^{2}<0$.
	Under Assumption \ref{ass:S_v}(i) and (iii), we have $\partial^{2}U_{i}/\partial S_{i}\partial\upsilon_{i}>0$,
	$\partial^{2}p_{ij}/\partial S_{i}^{2}<0$, and $\partial^{2}a_{i}/\partial H_{i}^{2}<0$.
	Suppose that $\partial S_{i}^{*}/\partial\upsilon_{i}\leq0$. Equations
	(\ref{eq:foc.study.dv}) and (\ref{eq:foc.social.dv}) imply that
	$\partial L_{i}^{*}/\partial\upsilon_{i}<0$ and $\partial H_{i}^{*}/\partial\upsilon_{i}<0$.
	This contradicts the setup that $H_{i}^{*}+S_{i}^{*}+L_{i}^{*}=1$:
	the amounts of time spent on studying, socializing, and leisure cannot
	all decrease in $\upsilon_{i}$ because they sum up to $1$. We conclude
	that $\partial S_{i}^{*}/\partial\upsilon_{i}>0$.
\end{proof}
From the lemma, socializing is increasing in $\upsilon_{i}$. Because
the number of friends is increasing in the time spent socializing,
we therefore expect that the number of friends is positively (negatively)
correlated with $\epsilon_{i}$ if $\upsilon_{i}$ and $\epsilon_{i}$
are positively (negatively) correlated.

For example, unobserved communication skills in $\upsilon_{i}$ may
contribute to more friends and better labor market performance, resulting
in a positive correlation between friendships and $\epsilon_{i}$.
This will yield an upward bias in the social returns. On the other
hand, if unobserved tastes for socializing in $\upsilon_{i}$ are
in favor of leisure and counterproductive on the labor market, then
we expect a negative correlation between friendships and $\epsilon_{i}$,
yielding a downward bias in the social returns. Overall, the OLS bias
in the social returns can be ambiguous, depending on whether the positive
or negative correlation dominates.

\subsubsection{\label{app:Proof.homophily}Proposition \ref{prop:homophily}}

To prove Proposition \ref{prop:homophily}, let $d_{ij}=d(X_{i},X_{j})$
denote the social distance between students $i$ and $j$, and assume
that the linking probability $p_{ij}(S_{i},S_{j})$ depends on social
distance $d_{ij}$. We examine the impact of a change in social distance
$d_{ij}$, by changing $j$'s characteristics ($X_{j}$), on $i$'s
time spent socializing and studying and on later friendships and education.

To speculate on the impacts of social distance, we make additional
assumptions on the utility and production functions.
\begin{assumption}
	\label{ass:S_d}(i) The marginal product of socializing is decreasing
	in social distance (\textup{$\partial^{2}p_{ij}/\partial S_{i}\partial d_{ij}<0$}),
	that is, students with higher social distance to peers make fewer
	friends from each unit of socializing. (ii) The marginal utility from
	leisure does not depend on socializing ($\partial^{2}U_{i}/\partial L_{i}\partial S_{i}=0$).
	The marginal product of studying does not depend on social distance
	(\textup{$\partial^{2}a_{i}/\partial H_{i}\partial d_{ij}=0$}).\footnote{This assumption is imposed for simplicity and can be relaxed. For
		example, if we assume $\partial^{2}p_{ij}/\partial S_{i}\partial d_{ij}<\partial^{2}a_{i}/\partial H_{i}\partial d_{ij}<0$,
		we can derive similar results on $S_{i}^{*}$, though with more ambiguity.} (iii) The marginal product from socializing (or studying) is diminishing
	in socializing (or studying) (\textup{$\partial^{2}p_{ij}/\partial S_{i}^{2}<0$}
	and \textup{$\partial^{2}a_{i}/\partial H_{i}^{2}<0$}). The marginal
	product from socializing is increasing in peer's socializing ($\partial^{2}p_{ij}/\partial S_{i}\partial S_{j}>0$).
	(iv) The labor market returns to education and friendships are positive
	constants (\textup{$\partial Y_{i}/\partial E_{i}=r_{e}>0$} and $\partial Y_{i}/\partial F_{i}=r_{f}>0$)\textup{.}
\end{assumption}
\begin{proof}[Proof of Proposition \ref{prop:homophily}]
	Suppose that the change of $d_{ij}$ (due to the change in $X_{j}$)
	does not trigger any change in the equilibrium so that $\partial S_{k}^{*}/\partial d_{ij}=0$,
	for all $k\neq i,j$. Taking the derivatives of both sides of (\ref{eq:foc.study})
	and (\ref{eq:foc.social}) evaluated at optimal $(H_{i}^{*},S_{i}^{*})$,
	$i=1,\dots,n$, with respect to $d_{ij}$, we obtain
	\begin{eqnarray}
		\frac{\partial^{2}U_{i}}{\partial L_{i}^{2}}(S_{i}^{*},L_{i}^{*})\frac{\partial L_{i}^{*}}{\partial d_{ij}} & = & r_{e}\frac{\partial^{2}a_{i}}{\partial H_{i}^{2}}(H_{i}^{*})\frac{\partial H_{i}^{*}}{\partial d_{ij}}\label{eq:foc.study.dd}\\
		\frac{\partial^{2}U_{i}}{\partial L_{i}^{2}}(S_{i}^{*},L_{i}^{*})\frac{\partial L_{i}^{*}}{\partial d_{ij}} & = & r_{f}\frac{\partial^{2}p_{ij}}{\partial S_{i}\partial d_{ij}}(S_{i}^{*},S_{j}^{*})+\left(\frac{\partial^{2}U_{i}}{\partial S_{i}^{2}}(S_{i}^{*},L_{i}^{*})+r_{f}\sum_{k\neq i}\frac{\partial^{2}p_{ik}}{\partial S_{i}^{2}}(S_{i}^{*},S_{k}^{*})\right)\frac{\partial S_{i}^{*}}{\partial d_{ij}}\nonumber \\
		&  & +r_{f}\frac{\partial^{2}p_{ij}}{\partial S_{i}\partial S_{j}}(S_{i}^{*},S_{j}^{*})\frac{\partial S_{j}^{*}}{\partial d_{ij}},\label{eq:foc.social.dd}
	\end{eqnarray}
	where we have used $\partial^{2}U_{i}/\partial L_{i}\partial S_{i}=0$,
	\textit{$\partial^{2}a_{i}/\partial H_{i}\partial d_{ij}=0$, $\partial Y_{i}/\partial E_{i}=r_{e}$},
	and $\partial Y_{i}/\partial F_{i}=r_{f}$ by Assumption \ref{ass:S_d}(ii)
	and (iv), and the fact that $U_{i}$ and $p_{ik}$ (for $k\neq i,j$)
	do not depend on $d_{ij}$. Note that the presence of the term $\partial S_{j}^{*}/\partial d_{ij}$
	in equation (\ref{eq:foc.social.dd}) is because $S_{j}^{*}$ depends
	on the social distance $d_{ij}$ and $j$'s characteristics $X_{j}$.
	
	We start with the special case where we ignore the effect on $S_{j}^{*}$
	and assume that $\partial S_{j}^{*}/\partial d_{ij}=0$. Recall that
	$\partial^{2}U_{i}/\partial L_{i}^{2}<0$ and $\partial^{2}U_{i}/\partial S_{i}^{2}<0$.
	Under Assumption \ref{ass:S_d}(i) and (iii), we have $\partial^{2}p_{ij}/\partial S_{i}\partial d_{ij}<0$,
	$\partial^{2}p_{ik}/\partial S_{i}^{2}<0$ (for $k\neq i$) and $\partial^{2}a_{i}/\partial H_{i}^{2}<0$.
	Suppose that $\partial S_{i}^{*}/\partial d_{ij}\geq0$. Equations
	(\ref{eq:foc.social.dd}) and (\ref{eq:foc.study.dd}) then imply
	that $\partial L_{i}^{*}/\partial d_{ij}>0$ and $\partial H_{i}^{*}/\partial d_{ij}>0$.
	This contradicts the setup $H_{i}^{*}+S_{i}^{*}+L_{i}^{*}=1$: the
	amounts of time spent on studying, socializing, and leisure cannot
	all increase in $d_{ij}$ because they sum up to $1$. We conclude
	that $\partial S_{i}^{*}/\partial d_{ij}<0$, that is, socializing
	is decreasing in social distance $d_{ij}$. Moreover, from equation
	(\ref{eq:foc.study.dd}) we can see that $\partial L_{i}^{*}/\partial d_{ij}$
	and $\partial H_{i}^{*}/\partial d_{ij}$ have the same sign. Because
	$\partial S_{i}^{*}/\partial d_{ij}<0$, we thus derive that $\partial H_{i}^{*}/\partial d_{ij}>0$
	and $\partial L_{i}^{*}/\partial d_{ij}>0$. In sum, if the impact
	of social distance $d_{ij}$ on $j$'s time spent socializing is ignored,
	an increase in social distance $d_{ij}$ leads $i$ to socialize less
	and study more.
	
	In general, we have $\partial S_{j}^{*}/\partial d_{ij}\neq0$, that
	is, social distance $d_{ij}$ affects $j$'s time spent socializing.
	If we know $\partial S_{j}^{*}/\partial d_{ij}\leq0$, by $\partial^{2}p_{ij}/\partial S_{i}\partial S_{j}>0$
	(Assumption \ref{ass:S_d}(iii)) and the same reasoning as above,
	we can derive $\partial S_{i}^{*}/\partial d_{ij}<0$ and $\partial H_{i}^{*}/\partial d_{ij}>0$.
	However, the change in $X_{j}$ may lead to an increase in $j$'s
	time spent socializing ($\partial S_{j}^{*}/\partial d_{ij}>0$).
	In this case, the sign of $\partial S_{i}^{*}/\partial d_{ij}$ becomes
	ambiguous, as does the sign of $\partial H_{i}^{*}/\partial d_{ij}$.
	In other words, because the impact of social distance $d_{ij}$ (through
	$j$'s characteristics $X_{j}$) on $j$'s time spent socializing
	is ambiguous, we cannot predict its impact on $i$'s time spent socializing
	and studying.
	
	Even though the signs of $\partial S_{i}^{*}/\partial d_{ij}$ and
	$\partial H_{i}^{*}/\partial d_{ij}$ are ambiguous, from equations
	(\ref{eq:foc.study.dd}) and (\ref{eq:foc.social.dd}) we can derive
	that $\partial S_{i}^{*}/\partial d_{ij}\neq0$ and $\partial H_{i}^{*}/\partial d_{ij}\neq0$,
	that is, a change in social distance $d_{ij}$ affects $i$'s time
	spent socializing and studying. To see this, suppose that $\partial S_{i}^{*}/\partial d_{ij}=0$.
	Under Assumption \ref{ass:S_d}(i), (iii) and (iv), in general we
	have $r_{e}\partial^{2}p_{ij}/\partial S_{i}\partial d_{ij}+r_{f}\partial^{2}p_{ij}/\partial S_{i}\partial S_{j}\cdot\partial S_{j}^{*}/\partial d_{ij}\neq0$,
	so for equations (\ref{eq:foc.study.dd}) and (\ref{eq:foc.social.dd})
	to hold, we must have $\partial H_{i}^{*}/\partial d_{ij}\neq0$ and
	$\partial L_{i}^{*}/\partial d_{ij}\neq0$. Note that equation (\ref{eq:foc.study.dd})
	implies that $\partial H_{i}^{*}/\partial d_{ij}$ and $\partial L_{i}^{*}/\partial d_{ij}$
	have the same sign. This implies that $\partial(H_{i}^{*}+L_{i}^{*})/\partial d_{ij}\neq0$,
	which contradicts $\partial S_{i}^{*}/\partial d_{ij}=0$ because
	$H_{i}^{*}+S_{i}^{*}+L_{i}^{*}=1$. Moreover, because $\partial S_{i}^{*}/\partial d_{ij}\neq0$
	implies that the sum $H_{i}^{*}+L_{i}^{*}$ must change, and $\partial H_{i}^{*}/\partial d_{ij}$
	and $\partial L_{i}^{*}/\partial d_{ij}$ have the same sign, we can
	further derive that $\partial H_{i}^{*}/\partial d_{ij}\neq0$. To
	sum up, we have proved that $\partial S_{i}^{*}/\partial d_{ij}\neq0$
	and $\partial H_{i}^{*}/\partial d_{ij}\neq0$, that is, social distance
	$d_{ij}$ has nonzero effects on socializing and studying.
	
	Finally, because the expected number of friends and expected education
	depend on time spent socializing and studying ($\mathbb{E}[F_{i}|X,\upsilon,\omega]=\sum_{j\neq i}p_{ij}(S_{i},S_{j})$
	and $\mathbb{E}[E_{i}|X,\upsilon,\omega]=a_{i}(H_{i})$), a change
	in social distance $d_{ij}$ would in general affect friendships and
	education, through the change in socializing and studying. However,
	because the impacts of social distance $d_{ij}$ on socializing and
	studying are ambiguous, so are its impacts on friendships and education.
\end{proof}

\subsubsection{\label{app:Proof.endow}Proposition \ref{prop:endow}}

In this section, we investigate the roles of individual endowments
such as IQ and extroversion on socializing, studying and their friendship
and education outcomes. For convenience, we write $X_{i}=(X_{i,1},X_{i,-1})$,
where $X_{i,1}$ represents the endowment under investigation (e.g.,
IQ and extroversion) and $X_{i,-1}$ represents the remaining observed
characteristics. Assume that $X_{i,1}$ does not directly affect $i$'s
social distance to peers.

To speculate on the roles of IQ and extroversion, we make additional
assumptions on the utility and production functions. In particular,
we impose Assumptions \ref{ass:S_x} and \ref{ass:S_x.iq} when $X_{i,1}$
represents IQ and impose Assumptions \ref{ass:S_x} and \ref{ass:S_x.extro}
when $X_{i,1}$ represents extroversion.
\begin{assumption}
	\label{ass:S_x}(i) The marginal utility from leisure does not depend
	on socializing nor endowment $X_{i,1}$ ($\partial^{2}U_{i}/\partial L_{i}\partial S_{i}=0$
	and \textup{$\partial^{2}U_{i}/\partial L_{i}\partial X_{i,1}=0$}).
	(ii) The marginal product from socializing (or studying) is diminishing
	in socializing (or studying) (\textup{$\partial^{2}p_{ij}/\partial S_{i}^{2}<0$}
	and \textup{$\partial^{2}a_{i}/\partial H_{i}^{2}<0$}). (iii) The
	labor market returns to education and friendships are positive constants
	(\textup{$\partial Y_{i}/\partial E_{i}=r_{e}>0$} and $\partial Y_{i}/\partial F_{i}=r_{f}>0$)\textup{.}
\end{assumption}
\begin{assumption}[IQ]
	\label{ass:S_x.iq}The marginal product of studying is increasing
	in IQ (\textup{$\partial^{2}a_{i}/\partial H_{i}\partial X_{i,1}>0$}),
	that is, students with higher IQ attain higher levels of education
	from each unit of studying.
\end{assumption}
\begin{assumption}[Extroversion]
	\label{ass:S_x.extro}(i) The marginal product of socializing is
	increasing in the magnitude of extroversion (\textup{$\partial^{2}p_{ij}/\partial S_{i}\partial X_{i,1}>0$}),
	that is, students who are more extroverted make more friends from
	each unit of socializing. The marginal utility from socializing is
	also increasing in the magnitude of extroversion ($\partial^{2}U_{i}/\partial S_{i}\partial X_{i,1}>0$).
	(ii) The marginal product of studying does not depend on the the magnitude
	of extroversion (\textup{$\partial^{2}a_{i}/\partial H_{i}\partial X_{i,1}=0$}).
\end{assumption}
\begin{proof}[Proof of Proposition \ref{prop:endow}]
	Because $X_{i,1}$ does not directly affect $i$'s social distance
	to peer $j$, under the assumption that the equilibrium effects are
	negligible, we obtain $\partial S_{j}^{*}/\partial X_{i,1}=0$ for
	$j\neq i$, that is, $i$'s endowment does not affect $j$'s socializing.
	Taking the derivatives of both sides of (\ref{eq:foc.study}) and
	(\ref{eq:foc.social}) evaluated at optimal $(H_{i}^{*},S_{i}^{*})$,
	$i=1,\dots,n$, with respect to $X_{i,1}$, we obtain
	\begin{eqnarray}
		\frac{\partial^{2}U_{i}}{\partial L_{i}^{2}}(S_{i}^{*},L_{i}^{*})\frac{\partial L_{i}^{*}}{\partial X_{i,1}} & = & r_{e}\frac{\partial^{2}a_{i}}{\partial H_{i}\partial X_{i,1}}(H_{i}^{*})+r_{e}\frac{\partial^{2}a_{i}}{\partial H_{i}^{2}}(H_{i}^{*})\frac{\partial H_{i}^{*}}{\partial X_{i,1}}\label{eq:foc.study.dx}\\
		\frac{\partial^{2}U_{i}}{\partial L_{i}^{2}}(S_{i}^{*},L_{i}^{*})\frac{\partial L_{i}^{*}}{\partial X_{i,1}} & = & \frac{\partial^{2}U_{i}}{\partial S_{i}\partial X_{i,1}}(S_{i}^{*},L_{i}^{*})+r_{f}\frac{\partial^{2}p_{ij}}{\partial S_{i}\partial X_{i,1}}(S_{i}^{*},S_{j}^{*})\nonumber \\
		&  & +\left(\frac{\partial^{2}U_{i}}{\partial S_{i}^{2}}(S_{i}^{*},L_{i}^{*})+r_{f}\sum_{j\neq i}\frac{\partial^{2}p_{ij}}{\partial S_{i}^{2}}(S_{i}^{*},S_{j}^{*})\right)\frac{\partial S_{i}^{*}}{\partial X_{i,1}}\label{eq:foc.social.dx}
	\end{eqnarray}
	where we have used $\partial^{2}U_{i}/\partial L_{i}\partial S_{i}=0$\textit{,
		$\partial^{2}U_{i}/\partial L_{i}\partial X_{i,1}=0$, $\partial Y_{i}/\partial E_{i}=r_{e}$},
	and $\partial Y_{i}/\partial F_{i}=r_{f}$ by Assumption \ref{ass:S_x}(i)
	and (iii).
	
	\textbf{IQ.} We first consider the scenario where $X_{i,1}$ represents
	IQ. If we assume that 
	\begin{equation}
		\frac{\partial^{2}U_{i}}{\partial S_{i}\partial X_{i,1}}(S_{i}^{*},L_{i}^{*})+r_{f}\frac{\partial^{2}p_{ij}}{\partial S_{i}\partial X_{i,1}}(S_{i}^{*},S_{j}^{*})\leq0,\label{eq:MU.iq}
	\end{equation}
	which represents the case where students with higher IQ are less efficient
	in socializing, then we can derive that $\partial H_{i}^{*}/\partial X_{i,1}>0$
	and $\partial S_{i}^{*}/\partial X_{i,1}<0$, that is, smarter students
	study more and socialize less. To see this, suppose that $\partial H_{i}^{*}/\partial X_{i,1}\leq0$.
	Recall that $\partial^{2}U_{i}/\partial L_{i}^{2}<0$ and $\partial^{2}U_{i}/\partial S_{i}^{2}<0$.
	Under Assumptions \ref{ass:S_x}(ii) and \ref{ass:S_x.iq}, we have
	$\partial^{2}p_{ij}/\partial S_{i}^{2}<0$, $\partial^{2}a_{i}/\partial H_{i}^{2}<0$,
	and $\partial^{2}a_{i}/\partial H_{i}\partial X_{i,1}>0$. If $\partial H_{i}^{*}/\partial X_{i,1}\leq0$,
	equations (\ref{eq:foc.study.dx}) and (\ref{eq:foc.social.dx}) then
	imply that $\partial L_{i}^{*}/\partial X_{i,1}<0$ and $\partial S_{i}^{*}/\partial X_{i,1}<0$.
	This contradicts the setup $H_{i}^{*}+S_{i}^{*}+L_{i}^{*}=1$: the
	amounts of time spent on studying, socializing, and leisure cannot
	all decrease in $X_{i,1}$ because they sum up to $1$. We conclude
	that $\partial H_{i}^{*}/\partial X_{i,1}>0$, that is, studying is
	increasing in IQ. Moreover, suppose that $\partial S_{i}^{*}/\partial X_{i,1}\geq0$.
	From equation (\ref{eq:foc.social.dx}) we derive that $\partial L_{i}^{*}/\partial d_{ij}>0$,
	which again contradicts the setup that $H_{i}^{*}+S_{i}^{*}+L_{i}^{*}=1$.
	In sum, if condition (\ref{eq:MU.iq}) is satisfied, smarter students
	spend more time studying and less time socializing. This is in line
	with the idea of ``nerdy'' students who are not popular and prefer
	to spend their time studying.
	
	In general, without condition (\ref{eq:MU.iq}), the roles of IQ on
	socializing and studying become ambiguous. For example, if 
	\[
	0<\frac{\partial^{2}U_{i}}{\partial S_{i}\partial X_{i,1}}(S_{i}^{*},L_{i}^{*})+r_{f}\frac{\partial^{2}p_{ij}}{\partial S_{i}\partial X_{i,1}}(S_{i}^{*},S_{j}^{*})<r_{e}\frac{\partial^{2}a_{i}}{\partial H_{i}\partial X_{i,1}}(H_{i}^{*}),
	\]
	following similar reasoning we can derive $\partial H_{i}^{*}/\partial X_{i,1}>0$
	and $\partial L_{i}^{*}/\partial X_{i,1}<0$, but the sign of $\partial S_{i}^{*}/\partial X_{i,1}$
	is ambiguous. More interesting, if
	\[
	\frac{\partial^{2}U_{i}}{\partial S_{i}\partial X_{i,1}}(S_{i}^{*},L_{i}^{*})+r_{f}\frac{\partial^{2}p_{ij}}{\partial S_{i}\partial X_{i,1}}(S_{i}^{*},S_{j}^{*})>r_{e}\frac{\partial^{2}a_{i}}{\partial H_{i}\partial X_{i,1}}(H_{i}^{*}),
	\]
	we can derive that $\partial S_{i}^{*}/\partial X_{i,1}>0$ and $\partial L_{i}^{*}/\partial X_{i,1}<0$,
	but the sign of $\partial H_{i}^{*}/\partial X_{i,1}$ is ambiguous.
	This is the case where smarter students are relatively more efficient
	in socializing than studying: they enjoy more utility from one unit
	of socializing or are more productive in making friends. It is reasonable
	that these students socialize more by resting less, though it is unclear
	whether they study less as well.
	
	Because the roles of IQ on socializing and studying are in general
	ambiguous, so are its roles on their friendships and education.
	
	\textbf{Extroversion.} Next, we turn to the scenario where $X_{i,n}$
	represents extroversion. Unlike IQ, the roles of extroversion are
	predictable. In particular, we can show that $\partial S_{i}^{*}/\partial X_{i,1}>0$,
	$\partial H_{i}^{*}/\partial X_{i,1}<0$, and $\partial L_{i}^{*}/\partial X_{i,1}<0$,
	that is, students who are more extroverted socialize more, study less,
	and have less leisure. To see this, note that under Assumptions \ref{ass:S_x}(ii)
	and \ref{ass:S_x.extro}(i) and (ii), we have $\partial^{2}p_{ij}/\partial S_{i}^{2}<0$,
	$\partial^{2}a_{i}/\partial H_{i}^{2}<0$, $\partial^{2}U_{i}/\partial S_{i}\partial X_{i,1}>0$,
	$\partial^{2}p_{ij}/\partial S_{i}\partial X_{i,1}>0$, and $\partial^{2}a_{i}/\partial H_{i}\partial X_{i,1}=0$.
	Suppose that $\partial S_{i}^{*}/\partial X_{i,1}\leq0$. Equations
	(\ref{eq:foc.study.dx}) and (\ref{eq:foc.social.dx}) then imply
	that $\partial L_{i}^{*}/\partial X_{i,1}<0$ and $\partial H_{i}^{*}/\partial X_{i,1}<0$.
	This contradicts the setup $H_{i}^{*}+S_{i}^{*}+L_{i}^{*}=1$: the
	amounts of time spent on studying, socializing, and leisure cannot
	all decrease in $X_{i,1}$ because they sum up to $1$. We conclude
	that $\partial S_{i}^{*}/\partial X_{i,1}>0$, that is, socializing
	is increasing in the magnitude of extroversion. Moreover, from equation
	(\ref{eq:foc.study.dx}) we can see that $\partial L_{i}^{*}/\partial X_{i,1}$
	and $\partial H_{i}^{*}/\partial X_{i,1}$ have the same sign. Because
	$\partial S_{i}^{*}/\partial X_{i,1}>0$, we thus derive that $\partial H_{i}^{*}/\partial X_{i,1}<0$
	and $\partial L_{i}^{*}/\partial X_{i,1}<0$. In sum, more extroverted
	students spend more time socializing and less time studying and resting.
	
	Based on the roles of extroversion on time allocation, we expect that
	students who are more extroverted make more friends and attain lower
	levels of education.
\end{proof}

\section{\label{app:pairwise IV}Pairwise Instruments}

An alternative approach to constructing an individual-level instrument
is through a pairwise regression of friendship nominations. Specifically,
we run a Probit regression of whether $i$ normalizes $j$ as a friend
on pairwise homophily measures between $i$ and $j$ (such as age
distance). We also control for $i's$ characteristics ($X_{i}$),
the mean characteristics in $i's$ school-grade ($\bar{X}_{gs}$),
grade fixed effects ($\alpha_{g}$), and school fixed effects ($\lambda_{s}$).
In particular, the Probit regression of friendship nominations has
the specification
\[
\Pr(F_{ij}=1)=\Phi(\beta_{0}+\beta'_{1}d_{ij}+\beta'_{2}X_{i}+\beta'_{3}\bar{X}_{gs}+\alpha_{g}+\lambda_{s}),
\]
where $d_{ij}$ represents the pairwise homophily measures between
$i$ and $j$, and $\Phi$ represents the cdf of standard normal distribution.
To run this regression, we use the pairwise sample that consists of
all the pairs of individuals being in the same school and same grade.
From the pairwise regression, we obtain the predicted value of friendship
nomination for each pair of $i$ and $j$ ($\hat{\Pr}(F_{ij}=1)$).
The instrument for $i$'s friendship nominations is then given by
the sum of the predicted friendship nominations toward $i$, that
is, $\sum_{j\in c}\hat{\Pr}(F_{ij}=1)$, where the sum is over all
$j$ in $i$'s school-grade.

Compared with the average homophily measures, using the predicted
number of friends as an instrument can be more efficient. Following
\citet{NEWEY1994} and \citet{Wooldridge2010} we derive that under
homoscedasticity in the log earnings equation, an efficient instrument
for the number of friends takes the form of the predicted number of
friends. Using the predicted number of friends as an instrument instead
of a regressor in the second stage is also in line with what \citet[Section 4.6]{AngristHarmless}
suggested to avoid a forbidden regression in the case of a nonlinear
first stage.

\section{\label{app:ONET}Indices of Cognitive and Social Skills from the
	O{*}NET}

We use data from the O{*}NET to construct the indices of cognitive
and social skills.\footnote{The data can be found from https://www.onetonline.org/.}
We construct the indices as follows.

\textbf{Cognitive skills score.} From the O{*}NET we observe whether
an occupation requires the following cognitive abilities: Category
Flexibility; Deductive Reasoning; Flexibility of Closure; Fluency
of Ideas; Inductive Reasoning; Information Ordering; Mathematical
Reasoning; Memorization; Number Facility; Oral Comprehension; Oral
Expression; Originality; Perceptual Speed; Problem Sensitivity; Selective
Attention; Spatial Orientation; Speed of Closure; Time Sharing; Visualization;
Written Comprehension; and Written Expression. For each occupation
the O{*}NET reports the level of the skill that is needed. A higher
level indicates that the occupation requires a greater amount of that
skill. To compute the cognitive score for an occupation, we take the
average of the levels of these cognitive skills that O{*}NET reports
are needed for that particular occupation.

\textbf{Social skills score.} From the O{*}NET we observe whether
an occupation requires the following social skills: Coordination;
Instructing; Negotiation; Persuasion; Service Orientation; Social
Perceptiveness; Assisting and Caring for Others; Coaching and Developing
Others; Communicating with Persons Outside Organization; Communicating
with Supervisors, Peers or Subordinates; Coordinating the Work and
Activities of Others; Developing and Building Teams; Establishing
and Maintaining Interpersonal Relationships; Guiding, Directing and
Motivating Subordinates; Interpreting the Meaning of Information for
Others; Performing Administrative Activities; Performing for or Working
Directly with the Public; Provide Consultation and Advice to Others;
Resolving Conflicts and Negotiating with Others; Selling or Influencing
Others; Staffing Organizational Units; Training and Teaching Others.
For each occupation the O{*}NET reports the level of the skill that
is needed. A higher level indicates that the occupation requires a
greater amount of that skill. To compute the social score for an occupation,
we take the average of the levels of these social skills that are
needed in the occupation.

\section{\label{app:Code}STATA Code: Bound Estimates and Confidence Intervals}

\begin{lstlisting}[basicstyle={\footnotesize\ttfamily}]
* Set the significance level 
clear mata
mata alpha = 0.05
	
* Estimate the bounds by two-step GMM
xtivreg2 ln_earn_upper (in_grade=$iv) $xi $mean g2-g7 if insample, fe gmm cl(sid) 
mata bu = st_matrix("e(b)")' 
mata Vu = st_matrix("e(V)")
xtivreg2 ln_earn_lower (in_grade=$iv) $xi $mean g2-g7 if insample, fe gmm cl(sid) 
mata bl = st_matrix("e(b)")' 
mata Vl = st_matrix("e(V)")
	
* Construct the bound estimates and standard errors 
mata nb = rows(bu)                /* number of parameters      */
mata b_set = J(nb, 2, 0)          /* bound estimates           */
mata se_set = J(nb, 2, 0)         /* standard errors           */
mata d = J(nb, 1, 0)              /* widths of identified sets */
mata se_max = J(nb, 1, 0)         /* maximum of s.e.           */
mata for (i=1; i<=nb; i++){
    b_set[i,1] = bl[i,1] * (bu[i,1] > bl[i,1]) +       ///
                 bu[i,1] * (bu[i,1] < bl[i,1])
    b_set[i,2] = bu[i,1] * (bu[i,1] > bl[i,1]) +       ///
                 bl[i,1] * (bu[i,1] < bl[i,1])
    se_set[i,1] = Vl[i,i]^0.5 * (bu[i,1] > bl[i,1]) +  ///
                  Vu[i,i]^0.5 * (bu[i,1] < bl[i,1])
    se_set[i,2] = Vu[i,i]^0.5 * (bu[i,1] > bl[i,1]) +  ///                    
                  Vl[i,i]^0.5 * (bu[i,1] < bl[i,1])
    d[i,1] = b_set[i,2] - b_set[i,1]
    se_max[i,1] = se_set[i,1] * (se_set[i,1] > se_set[i,2]) +  /// 	  
                  se_set[i,2] * (se_set[i,1] < se_set[i,2])
}

* Calculate the critical values in Imbens and Manski (2004)
mata c = J(nb, 1, 0) 
mata c0 = invnormal(1 - alpha/2)     /* initial value of c */
mata void cfun(todo, c, d, se, alpha, f, g, H){
    f = (normal(c + d/se) - normal(-c) - (1 - alpha)) ^ 2
}
mata CV = optimize_init() 
mata optimize_init_which(CV, "min") 
mata optimize_init_evaluator(CV, &cfun()) 
mata optimize_init_params(CV, c0) 
mata optimize_init_argument(CV, 3, alpha)
mata for (i=1; i<=nb; i++){
    optimize_init_argument(CV, 1, d[i,1])
    optimize_init_argument(CV, 2, se_max[i,1])
    c[i,1] = optimize(CV)
}

* Construct the confidence interval in Imbens and Manski (2004)
mata ci = J(nb, 2, 0) 
mata for (i=1; i<=nb; i++){
    ci[i,1] = b_set[i,1] - c[i,1] * se_set[i,1]
    ci[i,2] = b_set[i,2] + c[i,1] * se_set[i,2]
}
mata round(b_set[1,.], 0.0001)   /* bounds               */
mata round(ci[1,.], 0.0001)      /* confidence interval  */
\end{lstlisting}

\end{document}